\documentclass[preprint]{imsart}

\arxiv{1612.00068}

\usepackage{fancyhdr,siunitx}
\usepackage{lmodern}% http://ctan.org/pkg/lm

\usepackage{amsmath}
\usepackage{color}
\usepackage{mathtools}
\mathtoolsset{showonlyrefs}
\usepackage{amsfonts,amstext, amssymb,amsthm, amsmath, rotating, latexsym,wasysym}
\usepackage{graphicx,bm,bbm,bigints}% ,

\usepackage[disable]{todonotes}
\usepackage[OT1]{fontenc}
\usepackage[numbers]{natbib}

\usepackage{multirow}
\usepackage{booktabs}
\usepackage{epstopdf,float}
\usepackage[]{geometry}
\usepackage{bigstrut}
\usepackage{mathrsfs}
\usepackage{enumitem}
\usepackage[thinlines,thiklines]{easybmat}
\usepackage[caption=false]{subfig}

\DeclareRobustCommand{\rchi}{{\mathpalette\irchi\relax}}
\newcommand{\irchi}[2]{\raisebox{\depth}{$#1\chi$}}
\usepackage[countmax]{subfloat}
\newtheorem{thm}{Theorem}
\newtheorem{lemma}{Lemma}

\newtheorem{remark}{Remark}

\makeatletter

\newcommand{\Rome}[1]{\expandafter\@slowromancap\romannumeral #1@}
\makeatother

\newcommand{\be}{\begin{equation}}
\newcommand{\ee}{\end{equation}}
\newcommand{\ba}{\begin{eqnarray}}
\newcommand{\ea}{\end{eqnarray}}
\newcommand{\bee}{\begin{equation*}}
\newcommand{\eee}{\end{equation*}}
\newcommand{\baa}{\begin{eqnarray*}}
\newcommand{\eaa}{\end{eqnarray*}}

\DeclareMathOperator*{\argmin}{arg\,min}

\newcommand{\R}{\mathbb{R}}

\newcommand{\p}{\mathbb{P}}

\newcommand{\w}{\mathcal{W}}
\newcommand{\g}{\mathbb{G}}
\newcommand{\f}{\mathcal{F}}
\newcommand{\E}{\mathbb{E}}

\newcommand{\G}{\mathcal{G}}
\newcommand{\h}{\mathcal{H}}

 \newcommand{\Ss}{\mathcal{S}}

\newcommand{\Var}{\mathrm{Var}}
\definecolor{lgray}{gray}{0.70}
\newcommand{\rpm}{\sbox0{$1$}\sbox2{$\scriptstyle\pm$}
  \raise\dimexpr(\ht0-\ht2)/2\relax\box2 }

\newcount\colveccount
\newcommand*\colvec[1]{
        \global\colveccount#1
        \begin{pmatrix}
        \colvecnext
}
\def\colvecnext#1{
        #1
        \global\advance\colveccount-1
        \ifnum\colveccount>0
                \\
                \expandafter\colvecnext
        \else
                \end{pmatrix}
        \fi
}
\DeclareMathOperator{\Span}{span}

\usepackage{tikz}
\usetikzlibrary{arrows}
\makeatletter
\newcommand\@erelb@r[1]{%
 \mathrel{\tikz[baseline=-.5ex]\draw[#1] (0,0)--(0.3,0);}
}
\newcommand{\erelbar}[1]{\@erelbar#1}
\def\@erelbar#1#2{%
 \ifcase\numexpr#1*4+#2\relax
 \@erelb@r{-}\or % 00
 \@erelb@r{->}\or % 01
 \@erelb@r{-|}\or % 02
 \@erelb@r{->|}\or % 03
 \@erelb@r{<-}\or % 10
 \@erelb@r{<->}\or % 11
 \@erelb@r{<-|}\or % 12
 \@erelb@r{<->}\or % 13
 \@erelb@r{|-}\or % 20
 \@erelb@r{|->}\or % 21
 \@erelb@r{|-|}\or % 22
 \@erelb@r{|<->|}\or % 23
 \@erelb@r{|<-}\or % 30
 \@erelb@r{|<->}\or % 31
 \@erelb@r{|<-|}\or % 32
 \@erelb@r{|<->|} % 33
 \else
 \@wrong
 \fi
}
\makeatother

\begin{document}
 \begin{frontmatter}
\title{Efficient Estimation in Single Index Models through Smoothing splines}
\runtitle{Smooth Single Index Models}
\begin{aug}
  \author{\fnms{Arun K.}  \snm{Kuchibhotla}\ead[label=e1]{arunku@upenn.edu}}
  \and
  \author{\fnms{Rohit K.} \snm{Patra}\ead[label=e2]{rohitpatra@ufl.edu}}
 
   % \ead[label=e3]{third@somewhere.com}%
    % \ead[label=u1,url]{http://www.foo.com}}

  \runauthor{Kuchibhotla and Patra}

  \affiliation{University of Pennsylvania and University of Florida}

  \address{University of Pennsylvania and University of Florida\\ \printead{e1,e2}}
% \thankstext{t1}{Corresponding author.}

\end{aug}
 % \author{ , Rohit Kumar Patra}
   % \author{\fnms{First}  \snm{Author}\thanksref{a,e1}\ead[label=e1,mark]{first@somewhere.com}},

% \author{Arun K. Kuchibhotla\\\texttt{arunku@upenn.edu}\\University of Pennsylvania \and Rohit K Patra \\\texttt{rohitpatra@ufl.edu}\\University of Florida}
% \affil[1]{University of Pennsylvania}
% \affil[2]{University of Florida}
% \email
% \author{\begin{tabular}{lcl}
% Arun K. Kuchibhotla & and & Rohit K. Patra \\
% \texttt{arunku@upenn.edu}&& \texttt{rohitpatra@ufl.edu}\\
% University of Pennsylvania &&University of Florida
% \end{tabular}}
% \maketitle
 % \runtitle{Smooth Single Index Model}
% \begin{aug}
% \author{\fnms{Rohit} \snm{K. Patra}\thanksref{a}\ead[label=e1]{rohitpatra@ufl.edu}}

% \address[a]{Department of Statistics, Univeristy of Florida, Gainesville, FL, U.S.A.
% \printead{e1}}
% \runauthor{Patra, R. K.}
% \affiliation{Univeristy of Florida}
% \end{aug}
% \end{frontmatter}

\begin{abstract}
 We consider estimation and inference in a single index regression model with an unknown but smooth link function. In contrast to the standard approach of using kernels or regression splines, we use smoothing splines to estimate the smooth link function. We develop a method to compute the penalized least squares estimators (PLSEs) of the parametric and the nonparametric components given independent and identically distributed (i.i.d.)~data. We prove the consistency and find the rates of convergence of the estimators. We establish asymptotic normality under under mild assumption and prove asymptotic efficiency of the parametric component under homoscedastic errors. A finite sample simulation corroborates our asymptotic theory. We also analyze a car mileage data set and a Ozone concentration data set. The identifiability and existence of the PLSEs are  also investigated.
\end{abstract}
% \begin{keyword}

% \textbf{Keywords:} least favorable submodel, penalized least squares, and semiparametric model.
 \end{frontmatter}
\section{Introduction}\label{smooth:sec:intro}
Consider a regression model where one observes  i.i.d.~copies of the predictor $X \in \R^d$ and the response $Y\in \R$ and is interested in estimating the regression function $\E(Y|X=\cdot).$ In nonparametric regression $\E(Y|X=\cdot)$ is generally assumed to satisfy some smoothness assumptions (e.g., twice continuously differentiable), but no assumptions are made on the form of dependence on $X$. While nonparametric models offer flexibility in modeling, the price for this flexibility can be high for two main reasons: the estimation precision decreases rapidly as  $d$ increases (``curse of dimensionality'') and the estimator can be hard to interpret when $d >1$.

A natural restriction of the nonparametric model that avoids the curse of dimensionality while still retaining some flexibility in the functional form of $\E(Y|X=\cdot)$ is the single index model. In single index models, one assumes the existence of $\theta_0 \in \R^d$ such that
$$\E(Y|X)=\E(Y| \theta_0^\top X), \quad \text{almost every (a.e.)} \, X,$$
where $\theta_0^\top X$ is called the index; the widely used generalized linear models (GLMs) are special cases. This dimension reduction gives single index models considerable advantages in applications when $d >1$ compared to  the general nonparametric regression model; see \cite{Horowitz09} and \cite{carroletal97} for a discussion. The aggregation of  dimension by the index enables us to estimate the conditional mean function at a much faster rate than in a general nonparametric model. Since \cite{Powelletal89}, single index models have become increasingly popular in many scientific fields including biostatistics, economics, finance, and environmental science and have been deployed in a variety of settings; see \cite{liracine07}. %These models are also used extensively in projection pursuit regression; see \cite{FriedmanStuetzle81}.

 Formally, in this paper, we consider the model
\begin{equation}\label{smooth:simsl}
Y = m_0(\theta_0^{\top}X) + \epsilon, \quad \mathbb{E}(\epsilon|X) = 0, \quad \text{a.e.} \, X,
\end{equation}
where $m_0:\mathbb{R} \to \mathbb{R}$ is called the link function, $\theta_0 \in \R^d$ is the index parameter, and $\epsilon$ is the unobserved mean zero error (with finite variance). We assume that both $m_0$ and $\theta_0$ are unknown and are the parameters of interest. For identifiability of~\eqref{smooth:simsl}, we assume that the first coordinate of $\theta_0$ is non-zero and % For identification of  $m_0$ and $\theta_0$ separately, we assume that the first coordinate of $\theta_{0}$ is non-zero and
\be\label{smooth:eq:theta}
\theta_0\in\Theta := \{\eta=(\eta_1, \ldots, \eta_d) \in\mathbb{R}^d: |\eta|=1\mbox{ and }\eta_{1}\ge 0\} \subset S^{d-1},
\ee
where  $| \cdot |$ denotes the Euclidean norm, and  $S^{d-1}$ is the Euclidean unit sphere in $\mathbb{R}^d$; see  \cite{carroletal97} and \cite{cuietal11} for a similar assumption.

%In this paper, we consider estimation of both $m_0$ and $\theta_0$ under the assumption that $m_0$ is a smooth function.

Most of the existing techniques for estimation in single index models can be broadly classified into two groups, namely, M-estimation and ``direct'' estimation. M-estimation  methods involve a nonparametric regression estimator of $m_0$ (e.g., kernel estimator \cite{ICHI93},  Bayesian B-splines \cite{Antoniadis04}, regression splines~\cite{MR2514187},   local-linear approximation~\cite{wu2010single}, and penalized splines \cite{YuRuppert02})  and a minimization of some appropriate criterion function (e.g., quadratic loss~\cite{MR2514187,YuRuppert02}, robust $L_1$ loss~\cite{zou2014m}, modal regression~\cite{liu2013robust}, and quantile regression~\cite{wu2010single}) with respect to the index parameter to obtain an estimator of $\theta_0$. The so-called direct estimation methods include average derivative estimators \cite{chaudhuri1997average,Hristacheetal01,Powelletal89,Stoker86}, methods based on the conditional variance of $Y$~\cite{Xia06,Xiaetal02},  dimension reduction techniques,  such as sliced inverse regression~\cite{Li91,LiDuan89}, and partial least squares \cite{ZhouHe08}. Another prominent direct method is  a  kernel-based fixed point iterative scheme  to compute an efficient estimator of $\theta_0$ \cite{cuietal11}. In these methods one tries to  directly estimate $\theta_0$  without estimating $m_0$, e.g., in \cite{Hristacheetal01} the authors use the estimate of the  derivative of the local linear approximation to $\E(Y|X=\cdot)$ and not  the estimate of $m_0$ to estimate $\theta_0.$

% In the last few decades various approaches have been proposed in the statistical literature for estimation in the smooth single index model. %for $\sqrt{n}$ consistent estimation of $\theta_0$  have been proposed in last decades.
% \cite{Stoker86} and \cite{Powelletal89} developed the so-called average derivative method for direct estimation of $\theta_0$; also see~\cite{HorwitzHardlee96}. The average derivative approach involves high-dimensional kernel smoothing that generally leads to poor performance.~\cite{Hristacheetal01} tries to address the poor performance by gradually lowering the dimension of the kernel smoother.
% \cite{Xiaetal02} propose estimating $\theta_0$ by minimizing the average conditional variance of $Y$; also see \cite{Xia06}. %{\clg  \cite{Xia06} proved  asymptotic efficiency of the refinements of the above two methods.}
% \cite{cuietal11} propose a  kernel-based fixed point iterative scheme  to compute an efficient estimator of $\theta_0$; however, they require an adaptive choice of tuning parameter(s) at each step of their iterative scheme.  %However, in these so-called direct methods the estimation of $m_0$ is not well studied.
% Other methods for estimation include sliced inverse regression (see \cite{LiDuan89} and \cite{Li91}) and partial least squares (see \cite{ZhouHe08}).

%Although nonparametric methods have been used to estimate the link function $m_0$, most methods use kernel smoothing techniques or regression splines (or B-splines??); see e.g.,~\cite{ICHI93},~\cite{YuRuppert02}, .... Dimension reduction techniques. ADE

In this paper we propose an M-estimation technique based on smoothing splines to simultaneously estimate the link function $m_0$ and the index parameter $\theta_0$.  When $\theta_0$ is known,~\eqref{smooth:simsl} reduces to a one-dimensional function estimation problem and smoothing splines offer a fast and easy-to-implement nonparametric estimator of the link function --- $m_0$ is generally estimated by minimizing  a penalized least squares criterion with a (natural) roughness penalty of integrated squared second derivative~\cite{greensilverman94,wahba90}. However, in the case of single index models, the problem is considerably harder as both the link function and the index parameter are unknown and intertwined (unlike in partial linear regression model \cite{HardleLiang07}).

In other words, given i.i.d.~data $\{(y_i, x_i)\}_{ 1\le i\le n}$ from model~\eqref{smooth:simsl},  we propose minimizing the following penalized loss:
\begin{align}\label{smooth:eq:temp_loss}
{} \frac{1}{n}\sum_{i=1}^n \big(y_i-m( \theta^\top x_i)\big)^2 + \lambda^2\int |m''(t)|^2dt\quad (\lambda \neq 0)
\end{align}
over $\theta \in \Theta$ and `smooth' functions $m$; we will make this more precise in Section~\ref{smooth:sec:prelim}.
 % {\clg such that the integrated (over the domain of the function) squared second derivative is finite}.
  Here $\lambda$ is known as the smoothing parameter --- high values of $|\lambda|$ lead to smoother estimators. 
The theory developed in this paper allows for the tuning parameter $\lambda$ in \eqref{smooth:eq:temp_loss}  to be data dependent. Thus data-driven  procedures such as  cross-validation can be used to choose  an optimal $\lambda$; see Section~\ref{smooth:sec:simul}. As opposed to average derivative methods discussed earlier \cite{Hristacheetal01,Powelletal89}, the optimization problem in \eqref{smooth:eq:temp_loss}  involves only 1-dimensional nonparametric function estimation. %; see e.g., \cite{wahba90} and \cite{silverman85} for further discussion on the effect of $\lambda$ on estimation.

To the best of our knowledge, this is the first work that uses smoothing splines in the single index paradigm, under (only) smoothness constraints. We show that the penalized least squares loss leads to a minimizer $(\hat m, \hat \theta)$. We study the asymptotic properties, i.e., consistency, rates of convergence, of the estimator $(\hat m, \hat \theta)$ under data dependent choices of the tuning parameter $\lambda$. We show that under sub-Gaussian errors $\hat \theta$ is asymptotically normal and, further, under homoscedastic errors $\hat \theta$ achieves the optimal semiparametric efficiency bound in the sense of \cite{BickelEtAl93}.
 % The constraint $\int |m''(t)|^2dt <\infty$ on the class of link functions  makes the nuisance tangent space difficult to characterize. As a result the efficient score is a difficult task.
%In Section \ref{smooth:sec:Assymplse}, we find an efficient estimate of $\theta_0$ with data dependent choice of $\lambda$.

\cite{ICHI93} developed a semiparametric least squares estimator of $\theta_0$ using  kernel estimates of the link function. However, the choice  of tuning parameters (e.g., the bandwidth for estimation of the link function)  make this procedure difficult to implement  \cite{DelecroixEtal06,HardleEtAl93}  and its numerical instability is well documented; see e.g., \cite{YuRuppert02}. To address these issues~\cite{MR2514187,YuRuppert02} used B-splines and penalized splines  to estimate $m_0$, respectively. However, in their proposed procedure the practitioner is required to choose the number and placement of knots for every $\theta$.  Smoothing splines avoid the choice of number of knots and their placement.   Furthermore, smoothing splines (or more generally RKHS based regression estimators) are unique in that they are defined as minimization over a Hilbert space rather than as a local average. Even though smoothing splines can be approximated by kernel regression estimators or can be seen as a linear smoother, they are obtained under global smoothness constraint. This viewpoint makes them readily usable (at least in principle) when more constraints (such as monotonicity, non-negativity, unimodality,  convexity, and $k$-monotonicity) need to be imposed. Several works including~\cite{henderson2009imposing},~\cite{pevsta2015shape}, and~\cite{yatchew1997nonparametric} advocate the use of smoothing splines for this reason. The above works also propose numerical methods for computing the constrained smoothing splines estimator in the case univariate nonparametric regression; also see~\cite{dontchev2002newton,ELF,meegaskumbura2011control,shen2015shape}. These works suggest that, in addition to the convenience in problem formulation, the proof techniques  for establishing consistency and asymptotic normality of the estimator for the finite-dimensional parameter in the constrained single index model will be almost the same as those for the smooth single index model studied here.  

In contrast, other regression estimators such as kernel (or Nadaraya-Watson) estimator, series expansion, and regression splines imposing almost any type of (shape) constraint requires rethinking of the methods from scratch. This difficulty has posed several interesting works that consider  estimation in  constrained one dimensional nonparametric regression models; see~\cite{beresteanu2004nonparametric},~\cite{hall2001nonparametric},~\cite{MR2516802}, and~\cite{racine2009constrained}. \cite{hall2001nonparametric}  modifies the kernel regression estimator by including probability weights for the summands and choosing these weights so as to satisfy monotonicity constraints. \cite{racine2009constrained} further extends this by allowing for negative weights and thus enlarging the possible set of constraints; the computation, however, becomes difficult. \cite{MR2516802} provides specific spline basis such that monotonicity and convexity constraints on functions can be converted into simple linear inequality constraints on the coefficients. However, this explicit basis construction for other general constraints (as discussed in~\cite{yatchew1997nonparametric}) seems out of reach at present and the extension of these methods to the case of single index model does not follow directly from existing work. 

% This motivates us to  use smoothing splines for estimation in the single index model. 

% \cite{beresteanu2004nonparametric} on the other hand use the classical B-splines and restricts the coefficients so as to impose more general constraints on the derivatives.  

% \end{itemize}
% 

This paper  gives a systematic and rigorous study of a smoothing splines based estimator for the single index model under minimal assumptions and fills an important gap in the literature. The assumptions for $m_0$ in this paper are weaker than those considered in the literature.  We assume that the link function has an absolutely continuous derivative as opposed to the assumed (almost) three times differentiability of $m_0$ \cite{cuietal11,ICHI93,Powelletal89,MR2514187}.  We study the model under the assumption that $\theta \in S^{d-1}$.   In contrast, when the first coordinate is assumed to be $1$, the parameter space is unbounded and  consistent estimation  of $\theta_0$ requires further assumptions, see e.g., \cite{LiPatilea15}.  \cite{cuietal11} points out that the assumption $\theta \in S^{d-1}$ makes the parameter space irregular and the construction of paths on the sphere is hard. In this paper we construct paths on the unit sphere to study the semiparametric efficiency of the finite dimensional parameter and provide a closed form expression for the variance of $\hat\theta$; see Theorem~\ref{smooth:thm:Main_rate}.

Our exposition is organized as follows. In Section \ref{smooth:sec:prelim} we introduce some notation, formally define our estimator, and study its existence. In Section \ref{smooth:sec:Assymplse}, we prove consistency (see Theorem \ref{smooth:thm:cons}) and provide the rates of convergence (see Theorems \ref{smooth:thm:mainc} and  \ref{smooth:thm:ratest}) for our estimator.
%In Section \ref{smooth:sec:SemiInf} we show that the smoothness assumption on $m$ does not affect the inference of $\theta_0$ asymptotically (in the Fisher sense) ???.
We show  that the  estimator  for $\theta_0$ is  asymptotically normal and  semiparametrically efficient; see Theorem \ref{smooth:thm:Main_rate} in Section~\ref{smooth:sec:SemiInf}. In Section \ref{smooth:sec:simul} we provide finite sample simulation study of the proposed estimator and compare performance with existing methods in the literature. In Section~\ref{sec:real_data}, we apply the methodology developed to the car mileage data and the Ozone concentration data. In Section~\ref{sec:summ}, we briefly summarize the results in the paper and provide some remarks on future directions of research. % We discuss an algorithm to compute our estimators in Section \ref{smooth:sec:compute}.
Appendices \ref{smooth:app:sec:Assymplse}--\ref{smooth:app:SemiInf} contain proofs of the some of the results in the paper. The proofs of the results not given in the Appendices can be found in the on-line supplementary material.

%\begin{itemize}
%\item Why is it hard?
%\item talk about the smoothness assumption effect.
%\item
%\end{itemize}
\section{Preliminaries} \label{smooth:sec:prelim}
Suppose that $\{(y_i, x_i)\}_{  1\le i\le  n}$  is an i.i.d.~sample from model~\eqref{smooth:simsl}. We start with some notation.  Let  $\rchi \subset \mathbb{R}^d$ denote the support of $X$.  Let $D$ be the set of possible index values and $D_0$ be the set of possible index values at $\theta_0$, i.e., \[D:= \{\theta^{\top}x :  x\in\rchi, \theta\in \Theta\}\quad \text{and} \quad D_0:= \{\theta_0^{\top}x :  x\in\rchi\} .\]   We denote the class of all real-valued functions with absolutely continuous first derivative on $D$ by $\mathcal{S}$, i.e.,$$\mathcal{S}:= \{m:D\to\mathbb{R}|\, m^\prime \text{ is absolutely continuous}\}.$$ We use $\p$ to denote the probability of an event, $\E$ for the expectation of a random quantity, and  $P_X$ for the  distribution of $X$. For $g : \rchi \to \mathbb{R}$, define
$$\|g\|^2 := \int_\chi g^2 dP_X \qquad \text{ and} \qquad \|g\|_n^2= \frac{1}{n}\sum_{i=1}^n g^2(x_i).$$
 Let $P_{\epsilon, X}$ denote the joint distribution of $(\epsilon, X)$ and $P_{\theta,m}$  denote the joint distribution of $(Y,X)$ when  $Y= m(\theta^\top X) +\epsilon.$ In particular,  $P_{\theta_0, m_0}$ denotes the joint distribution of $(Y,X)$ when $(Y,X)$ satisfy \eqref{smooth:simsl}. For any function $g: I \subset\R^p \to \mathbb{R}$, let $\|g\|_{\infty} := \sup_{u \in I} |g(u)|.$
Moreover, for $I_1 \subset I,$ we define $\|g\|_{I_1} := \sup_{u \in I_1} |g(u)|.$ For any set $I\subset \R$, $\diameter(I)$ denotes the diameter of the set $I$. For any $a\in \R^d$ and $r>0$, $B_a(r)$ denotes the Euclidean ball of radius $r$ centered at $a$. The notation $a\lesssim b$ is used to express that $a$ is less than $b$ up to a positive constant multiple.
For any function $f:\rchi\rightarrow\R^r, r\ge1$, let $\{f_i\}_{1\le i \le r}$ denote each of the components, i.e.,  $f(x)= (f_1(x), \ldots, f_r(x)), r\ge 1$ and  $f_i: \rchi \to \R$. We define $\|f\|_{2,2}:= \sqrt{ \sum_{i=1}^r \|f_i\|^2}$  and $\| f\|_{2, \infty}:= \sqrt{ \sum_{i=1}^r \|f_i\|^2_\infty}.$
For any real-valued function $m$ and $\theta \in \Theta$, we define
$$(m\circ\theta)(x) := m(\theta^{\top}x), \qquad \mbox{for all } x \in \rchi.$$
  % where $(g,\theta)\in \mathcal{G} \times \Theta.$ %We consider the penalized least squares to estimation problem with a smoothness penalty on the regression function $m$.
For any function $f: D\subset \R \rightarrow \R$ with absolutely continuous first derivative, we define the roughness penalty
\begin{equation}\label{smooth:eq:J}
J^2(f) := \int_D |f''(t)|^2 dt.
\end{equation}
% We assume that for the true link function $m_0$, $J(m_0) <\infty$ (see assumption \ref{smooth:a1} in Section \ref{smooth:sec:Assymplse}).

 \noindent The penalized loss for $(m,\theta) \in \Ss\times\Theta$ (and $\lambda\neq 0$) is defined as
\begin{align}\label{smooth:eq:L_n}
\mathcal{L}_n(m,\theta; \lambda) &:={} \frac{1}{n}\sum_{i=1}^n \big(y_i-m( \theta^\top x_i)\big)^2 + \lambda^2 J^2(m).
\end{align}
%See \cite{VanDeGeer90} for such a loss function in the case of 1-dimensional nonparametric regression.
For simplicity of notation, we define
$$Q_n(m,\theta) := \frac{1}{n}\sum_{i=1}^n \big(y_i-m ( \theta^\top x_i)\big)^2.$$
In this paper we study the following penalized least square estimator (PLSE):
\begin{equation}\label{smooth:simpls}
(\hat{m}, \hat{\theta}) := \argmin_{(m,\theta)\in\mathcal{S}\times\Theta}\mathcal{L}_n(m,\theta; \lambda).
\end{equation}
 Here we suppress the dependence of $(\hat{m}, \hat{\theta})$ on $\lambda$, for notational convenience. { The estimator $(\hat{m}, \hat{\theta})$ may not be unique. However, the analysis in the rest of the paper works for \textit{any} minimizer of $\mathcal{L}_n(m,\theta;\lambda)$.} The following theorem (proved in Section \ref{smooth:app:thm:existance_proof}  of the supplementary material) proves the existence of $(\hat{m},\hat{\theta})$  for every $\lambda\neq 0.$
\begin{thm}\label{smooth:thm:existance}
$\hat{\theta} \in \Theta$ and $\hat{m} \in \mathcal{S},$ where $\hat{\theta}$ and $\hat{m}$ are defined in \eqref{smooth:simpls}. Moreover, $\hat{m}$ is a natural cubic spline with knots at $\{\hat{\theta}^{\top} x_i\}_{1\le i\le n}$.
\end{thm}
% We now outline the identification of the composite population parameter $m_0\circ\theta_0$.   Define $Q(m,\theta) := \E[Y -m(\theta^\top X)]^2.$  The following argument shows that $(m_0,\theta_0)$ is the  minimizer of $Q$ and is well-separated (with respect to the $L_2(P_X)$-norm) from other elements in $\mathcal{S}\times\Theta$. Choose arbitrarily small $\delta > 0$, and pick any $(m,\theta)\in\mathcal{S}\times\Theta$ such that $\|m\circ\theta - m_0\circ\theta_0\|^2 > \delta^2$. Then
% \[
% Q(m,\theta) = \E [Y-m_0(\theta_0^\top X)]^2 + \E[m_0(\theta_0^\top X) - m(\theta^\top X)]^2,
% \]
% since $\E(\epsilon|X) = 0$. Thus, we have
% \[
% \inf_{\|m\circ\theta - m_0\circ\theta_0\|^2 > \delta^2} Q(m,\theta) - Q(m_0,\theta_0) > \delta^2.
% \]
% It is easy to see that the composite population parameter $m_0\circ\theta_0$ is identifiable. However, this does not guarantee that both $m_0$ and $\theta_0$ are separately identifiable.

Note that the identification of $m_0\circ\theta_0$ does not guarantee that both $m_0$ and $\theta_0$ are separately identifiable. \citet{ICHI93} (also see \citet[Pages 12--17]{Horowitz98} and~\citet[Proposition 8.1]{liracine07}) find sufficient conditions on the distribution/domain of $X$ under which  $m_0$ and $\theta_0$ can be separately identified:
%{\color{red} \cite{ICHI93} gives conditions for identifiability of $\theta$ and $m$.  Should we just use the same conditions?}
\begin{enumerate}[label=\bfseries (A\arabic*)]
\setcounter{enumi}{-1}
\item The function $m_0(\cdot)$ is non-constant, non-periodic, and a.e. differentiable.  The first coordinate of $\theta_0$ is positive, i.e., $\theta_{0,1} >0$.  The components of $X_1\sim P_X$ (i.e.,  $X_{1,1}, \ldots, X_{1,d-1}$ and  $X_{1,d}$) cannot have a perfect linear relationship.  There exists an integer $d_1\in \{1,2,\ldots,d\}$, such that  $X_{1,1},\ldots,X_{1,d_1-1},$ and $X_{1,d_1}$ have continuous distributions  and  $X_{1, d_1+1},\ldots, X_{1,d-1},$ and $X_{1,d}$ be discrete random variables. Furthermore, there exist an open interval $\mathcal{I}$ and constant vectors  $c_0,c_1,\ldots, c_{d-d_1} \in \mathbb{R}^{d-d_1}$ such that
\begin{itemize}
\item $c_l-c_0$ for $l\in \{1 ,\ldots, d-d_1\}$ are linearly independent,
\item $\mathcal{I} \subset \bigcap_{l=0}^{d-d_1} \big\{ \theta_0^\top x:  x\in \rchi \text{ and } (x_{d_1+1}, \ldots x_d) =c_l \big\}.$
\end{itemize}
\label{smooth:a0}
\end{enumerate}
\citet{ICHI93} and~\citet{Horowitz98} prove by examples that each part of Assumption~\ref{smooth:a0} is necessary for identifiability of $m_0$ and $\theta_0.$ Further discussion on alternative identifiability assumptions when $X$ has a Lebesgue density, we refer to~\citet[Section 2]{KuchPat17}.

\section{Asymptotic analysis of the PLSE}\label{smooth:sec:Assymplse}

In this section, we will list the assumptions under which we will establish consistency and find  the rates of convergence of our estimators.  Note that we will study $(\hat{m},\hat{\theta})$ for any (possibly data-driven) choice of $\lambda$ satisfying two rate conditions; see assumption \ref{smooth:a4} below.

\begin{enumerate}[label=\bfseries (A\arabic*)]
\setcounter{enumi}{0}
\item The link function $m_0$ satisfies $J(m_0) < \infty$. \label{smooth:a1} %{\color{red}[Triple diff for use in equation 50]}
 %We denote the penalized least squares estimators by $(\hat{m},\hat{\theta})$

\item $\rchi$, the support of $X$, is a compact subset of $\mathbb{R}^d$ and   $\sup_{x\in \rchi} |x| \le T.$\label{smooth:a2}
%is such that $\theta^{\top}\rchi$ is an interval in $\mathbb{R}$ for each $\theta\in\Theta$.
\item The error $\epsilon$ in model~\eqref{smooth:simsl} is conditionally  sub-Gaussian, i.e.,  there exists $K>0$ such that  $$\E\left[\exp\big({\epsilon^2}/{K}\big)|X\right]\le 2\quad \text{ a.e.}\quad X.$$ As stated in \eqref{smooth:simsl},  we also assume that $\mathbb{E}(\epsilon|X) = 0$ a.e.~$X$.  \label{smooth:a3}
\item The smoothing parameter $\lambda$ can be chosen to be  a random variable. For the rest of the paper, we denote it by $\hat{\lambda}_n$. Assume that $\hat{\lambda}_n$ satisfies the rate condition:\label{smooth:a4}
\begin{equation}\label{smooth:eq:conl}
\hat{\lambda}_n^{-1} = O_p(n^{2/5})\qquad \mbox{and}\qquad \hat{\lambda}_n = o_p(n^{-1/4}).
\end{equation}

% \item $\Var(X)$ is a positive definite matrix.
% \label{smooth:a7}

% \item $\mathbb{E}\big[XX^{\top}|m_0'(\theta_0^{\top}X)|^2\big]$ is a nonsingular matrix. \label{smooth:a6}

\end{enumerate}
The assumptions deserve comments.  In \ref{smooth:a1} our assumption on $m_0$ is quite minimal --- we essentially require $m_0$ to have an absolutely continuous derivative. Most previous works assume $m_0$ to be three times differentiable; see e.g., \cite{NeweyStroker93,Powelletal89}. Note that the assumption $J(m_0) < \infty$ in combination with compact support of $X$ implies that $m_0$ is bounded and we set $M_1 := \|m_0\|_{\infty}$.
% {\clg  For a fixed $\theta$ the support of the minimizer of $\mathcal{L}_n(m,\theta; \lambda)$ depends on $\theta$.}
 \ref{smooth:a2} assumes that  the support of the covariates is bounded.  As the class of functions $\Ss$ is not uniformly bounded, we use assumption \ref{smooth:a3}  to provide control over the tail behavior of $\epsilon$; see Chapter 8 of  \cite{VANG} for a discussion on this. Observe that \ref{smooth:a3} allows for heteroscedastic errors.  Assumption \ref{smooth:a4} allows our tuning parameter to be data dependent, as opposed to a sequence of constants.  This allows for data driven choices of $\hat{\lambda}_n$, such as cross-validation. We will show that for any choice  of $\hat{\lambda}_n$  satisfying \eqref{smooth:eq:conl}, $\hat\theta$ will be an asymptotically efficient estimator of $\theta_0$.  We use empirical process methods (e.g., see \cite{VdV98}) to prove the consistency and to find the rates of convergence of $\hat{m}\circ \hat{\theta}$.

In Theorem \ref{smooth:thm:mainc} we show that $(\hat{m},\hat{\theta})$ is a consistent estimator of $(m_0, \theta_0)$ and  $\hat{m}\circ\hat{\theta}$ converges to $m_0\circ\theta_0$ at  rate $\hat{\lambda}_n$ (with respect to the $L_2(P_X)$-norm).
\begin{thm}\label{smooth:thm:mainc}
Under assumptions \ref{smooth:a0}--\ref{smooth:a4}, the PLSE satisfies $J(\hat{m}) = O_p(1),$ $\|\hat{m}\|_\infty=O_p(1),$ and $\|\hat{m}\circ\hat{\theta} - m_0\circ\theta_0\| = O_p(\hat{\lambda}_n).$
\end{thm}

Next  we prove the consistency of $\hat{m}$ and $\hat{\theta}$. We prove that  $\hat{m}$ is consistent under the Sobolev norm, which for any set $I\subset\mathbb{R}$ and any function $g:I\to\mathbb{R}$ is defined as
\[
\|g\|^S_I = \sup_{t\in I}|g(t)| + \sup_{t\in I}|g'(t)|.
\]
\begin{thm}\label{smooth:thm:cons}
Under assumptions \ref{smooth:a0}--\ref{smooth:a4}, $\hat{\theta}\overset{P}{\to}\theta_0,$ $\|\hat{m} - m_0\|^S_{D_0}\overset{P}{\to}0,$ and $\|\hat{m}^\prime \|_\infty =O_p(1).$
\end{thm}

The above result shows that not only is $\hat m$ consistent but its derivative $\hat m'$ also converges uniformly to $m_0'$. Proof of Theorem \ref{smooth:thm:mainc} is in Appendix~\ref{smooth:app:thm:mainc_proof} and proof of  Theorem \ref{smooth:thm:cons} is given in Section~\ref{smooth:sec:proof_cons} the supplementary material. We next introduce further notation  and provide upper bounds on the rates of convergence of $\hat{\theta} $ and $\hat{m}$ separately.  

% For any $a\in \R^d$, let $a_{-1}$ denote the last $d-1$ coordinates of $a$. Another common reparameterization of the finite dimensional parameter in~\eqref{smooth:simsl} is to write $\theta= (1, \theta_{-1})$, where $\theta_{-1}\in\R^{d-1}$. However in this alternative parameterization, the finite dimensional parameter space is no longer bounded. As most estimators for $\theta$ are minimizers/solutions  of some criterion function, further assumptions on the estimator of $\theta_0$ are needed to make sure that the estimator does not diverge; see e.g., Section~2 of~\cite{MR1647661} and \cite{LiPatilea15}.
% \cite{cuietal11} considers the reparameterization $\theta= (\sqrt{1-|\theta_{-1}|^2}, \theta_{-1})$, where $|\theta_{-1}| <1.$ Under this parameterization the parameter space  is bounded; however, calculations of  the parametric score become unnecessarily tedious.
 % In this paper we consider a local parameterization so that perturbations of $\theta_0$ lie in $\Theta.$ 
 Recall  that $\Theta$ is a closed subset of $\R^d$ and the interior of $\Theta$  in $\R^d$ is the null set. Thus we will define a ``local parameterization matrix'' that will help us create linear perturbations of $\theta_0$ that lie in $\Theta.$  For every real matrix  $G \in \R^{m\times n}$, we define $\|G\|_2:= \max_{ x\in S^{n-1}} |G x|$. This is sometimes called the operator or matrix 2-norm; see e.g., page 281 of \cite{Meyer01}.  The following lemma proved\footnote{Our proof is constructive.} in Section~\ref{smooth:sec:lem:h_lip_proof} of the supplementary file shows that  the ``local parameterization matrix" as a function of $\theta$ is Lipschitz at $\theta_0$ with respect to the operator norm.

\begin{lemma}\label{smooth:lem:H_lip}
There exists a  set of matrices $\{H_\theta \in\mathbb{R}^{d\times(d-1)} : \theta \in \Theta\}$ satisfying the following properties:
\begin{enumerate}[label=(\alph*)]
\item $\xi\mapsto H_\theta\xi$ are bijections from $\mathbb{R}^{d-1}$ to the hyperplanes $\{x\in\mathbb{R}^d: \theta^\top  x=0\}$.
\item The columns of $H_\theta$ form an orthonormal  basis for $\{x\in\mathbb{R}^d: \theta^\top  x=0\}$.
\item   $\|H_\theta - H_{\theta_0}\|_2\le|\theta- \theta_0 |.$
\item For all  distinct $\eta,\beta \in \Theta\setminus\theta_0$, such that $|\eta-\theta_0|\le1/2$ and $|\beta-\theta_0|\le1/2,$
\begin{equation}\label{eq:H_lip_gen}
\|H_{\eta}^\top-H_{\beta}^\top\|_2 \le 8(1+8/\sqrt{15})\frac{ |\eta-\beta| }{ |\eta-\theta_0|+|\beta-\theta_0|}.
\end{equation}
\end{enumerate}%
\end{lemma}

Note that for each $\theta \in \Theta,$ $H_{\theta}^\top $ is the Moore-Penrose pseudo-inverse of $H_\theta$, e.g., $H_\theta^\top  H_\theta = \mathbb{I}_{d-1}$ where $\mathbb{I}_{d-1}$ is the identity matrix of order $d-1$; see Section 5.2 of \cite{bootstrap} for a similar construction.

The following distributional assumption on $X$ is used to find the  upper bounds on the rates of convergence of $\hat{\theta} $ and $\hat{m}$ separately. 

\begin{enumerate}[label=\bfseries (A\arabic*)]
\setcounter{enumi}{4}
\item $ H_{\theta_0}^\top\E\big[\Var(X|\theta_0^\top X) \{m_0'(\theta_0^{\top}X)\}^2 \big]H_{\theta_0}$ is a positive definite matrix.\label{smooth:a6}
\end{enumerate}
\todo[inline]{This assumptions makes sure that the problem is regular. The dimension of $\Var(X|\theta_0^\top X)$ and ``inherent'' dimension of $\beta$ must match to maintain finite information.}
If one of the continuous covariates with a nonzero index parameter has a density  (with respect to the Lebesgue measure) that is bounded away from zero (on its support) then assumption \ref{smooth:a6} is satisfied. Note that \ref{smooth:a6} fails if $m_0$ is a constant function; however a single index model is not identifiable if $m_0$ is constant (see \ref{smooth:a0}).
 The following bounds  (proved in Section~\ref{smooth:sec:proof_ratest} of the supplementary file)  will help us compute the asymptotic distribution of $\hat{\theta}$ in Section \ref{smooth:sec:SemiInf}.

\begin{thm}\label{smooth:thm:ratest}
 Under \ref{smooth:a0}--\ref{smooth:a6}, $\hat{m}$ and $\hat{\theta}$ satisfy
\[
|\hat{\theta} - \theta_0| = O_p(\hat{\lambda}_n) \quad \text{and}\quad \|\hat{m}\circ\theta_0 - m_0\circ\theta_0\| = O_p(\hat{\lambda}_n).
\]
\end{thm}
% \noindent Proof of Theorem \ref{smooth:thm:mainc} is in Appendix~\ref{smooth:app:thm:mainc_proof} and proofs of  Theorems \ref{smooth:thm:cons} and \ref{smooth:thm:ratest} are given in Sections  \ref{smooth:sec:proof_cons} and \ref{smooth:sec:proof_ratest} of the supplementary material, respectively.

\section{Semiparametric inference} \label{smooth:sec:SemiInf}
In this section we show that $\hat{\theta}$ is asymptotically normal and is a semiparametrically  efficient estimator of $\theta_0$ under homoscedastic errors. Before going into the derivation of the limit law of $\hat{\theta}$, we need to introduce some further notation and some regularity assumptions. For every $\theta\in\Theta$, let us define $D_\theta:= \{\theta^\top x: x\in \rchi\}$. Assumption~\ref{smooth:a0} implies that  there exists $r>0$ such that for all $\theta \in S^{d-1}\cap B_{\theta_0}(r)$ we have
\begin{equation}\label{eq:D_r}
  D_\theta \subsetneq D^{(r)}:= \bigcup_{\theta \in S^{d-1}\cap B_{\theta_0}(r)} D_\theta.
\end{equation}
See Section~\ref{rem:StrictSubset} of the supplementary file for a proof of this. 
For the rest of the paper we redefine $D:= D^{(r)}$. For every $\theta\in\Theta$, define $h_\theta: D \rightarrow \R^{ d}$ as
 \begin{align}
h_\theta(u) &:= \E[X |\theta^\top X=u].\label{smooth:eq:h_beta}
\end{align}
We use the following additional assumptions in the proof of asymptotic normality of $\hat\theta.$ 
 \begin{enumerate}[label=\bfseries (B\arabic*)]
\item $ h_{\theta}(\cdot)$ is twice continuously differentiable except possibly at a finite number of points, and for every $\theta_1$ and $\theta_2$ in $\Theta$,
\be \label{smooth:eq:lip_h_beta}
\|h_{\theta_1} -h_{\theta_2}\|_\infty \le \bar{M} |\theta_1-\theta_2|,
\ee
where $\bar{M}$ is a fixed finite constant.\label{smooth:b3}
\end{enumerate}
Let $p_{\epsilon,X}$ denote the  joint density (with respect to some dominating measure $\mu$ on $\R \times \rchi$) of $(\epsilon, X)$. Let $p_{\epsilon|X} (e,x)$ and $p_X(x)$ denote the corresponding conditional probability density of $\epsilon$ given $X$ and the marginal density of $X$, respectively. We define $\sigma : \rchi \rightarrow \R$ by  $\sigma^2( x):=\E(\epsilon^2|X=x).$
%{\clr  Discrete covariates ! where does it talk about the density assumption on X being irrelvalent with some base meausres}
%
%Basically it will be all functions of X with $\E(f(x))=0$ and the projection would of the score would be zero and not to be worried for need solid reference
 \begin{enumerate}[label=\bfseries (B\arabic*)]
\setcounter{enumi}{1}
\item  $p_{\epsilon|X} (e,x)$ is differentiable with respect to $e$, $\|\sigma^2(\cdot)\|_\infty < \infty$ and $\|1/\sigma^2(\cdot)\|_\infty <\infty$.\label{smooth:dens}
%\item We assume that $\E(\epsilon^2|X)=\sigma^2( X)$ \label{smooth:homo}
 %\item Assume $ \epsilon \sim N(0,\sigma^2)$ and  $\epsilon \independent X,$ $i.e., \epsilon$ and $X$ are independent. \label{smooth:b4}
\end{enumerate}
%{\color{red}[ and$\{ h_{\beta} , \beta \in \Theta\}$ is a Donsker class.]}
The assumptions \ref{smooth:b3} and \ref{smooth:dens} deserve comments.    The function $h_\theta$ plays a crucial role in the construction of ``least favorable" paths; see Section \ref{smooth:sec:App_least_fav_model}.   For the functions in the path to be in $\Ss$, we use the smoothness assumptions on $h_\theta$. In a way we need smoothness of $m_0$ or the distribution of $X$ to be smooth to be able to establish semiparametric efficiency.   \ref{smooth:dens} gives lower and upper bounds on the variance of $\epsilon$ as we are using a un-weighted least squares method to estimate parameters in a (possibly) heteroscedastic model.

In the sequel we will use standard empirical process theory notation. For any function $f: \R\times \rchi\to \R$ and  $(m, \theta) \in \mathcal{S}\times\Theta$, we define \[P_{\theta, m} f = \int f dP_{\theta, m}.\] Note that   $P_{\theta, m} f$ can be a random variable if $\theta$ (or $m$) is random. Moreover, for any function $f: \R\times \rchi\to \R$, we define \[\p_nf := \frac{1}{n}\sum_{i=1}^n f(y_i, x_i) \quad \text{and}\quad \g_n f:= \frac{1}{\sqrt{n}} \sum_{i=1}^n \big[f(y_i, x_i) -P_{\theta_0,m_0} f\big].\]

%
% The independence assumption in \ref{smooth:b4} is crucial as otherwise the model in \eqref{smooth:simsl} has an additional  nuisance parameter namely, the form of dependence between $X$ and $\epsilon$. With out additional restriction on the dependence structure of $X$ and $\epsilon$  the nuisance tangent space  is complicated and linear span (with respect to $L^2(P_{\theta_0,m_0})$ is intractable making semiparametric inference hard; see e.g., Chapter 4 of  \cite{Tsiatis06}. The Gaussian assumption can be relaxed as in Section 6 of \cite{WellnerZhang07}. However, in this paper we assume Gaussianity of the error variables to keep  the exposition focused.

\subsection{Efficient score}\label{smooth:sec:eff_score}
As a first step in showing that $\hat{\theta}$ is an efficient estimator, in the following we find the efficiency bound for $\theta_0$ in model~\eqref{smooth:simsl}. To compute the score for the model, we will first consider parametric paths on $\Theta$. For any $\eta\in\mathbb{R}^{d-1}$ and $\theta\in\Theta$,  we  now define a path $s\mapsto \zeta_s(\theta,\eta),$ for $s\in \R $ and $|s| \le |\eta|^{-1},$ as
\be\label{smooth:eq:path_para}
\zeta_s(\theta,\eta):=\sqrt{1-s^2|\eta|^2}\, \theta + s H_\theta \eta.
\ee
Note that $\theta^\top H_\theta =0_{d-1}$ and $|H_\theta \eta| =|\eta|$ for all $\eta \in \R^{d-1}$. When $|s| \le 1/|\eta|$ we have  $\zeta_s(\theta,\eta)\in S^{d-1}$. For every fixed $s \ne 0$, as $\eta$ varies in $B_0^{d-1}(|s|^{-1})$, $\zeta_s(\theta,\eta)$ takes all values in the set $\{\beta \in S^{d-1}: \theta^\top \beta > 0\}$ and $s H_\theta \eta$ is the orthogonal projection of $\zeta_s(\theta,\eta)$ onto the hyperplane $\{x\in\mathbb{R}^d: \theta^\top  x=0\}$.

We now attempt to calculate the efficient score for
\begin{equation}\label{smooth:eq:score_general_model}
     Y=m(\theta^\top X)+\epsilon
 \end{equation} for some $(m,\theta)\in \Ss\times\Theta$ under assumptions \ref{smooth:a3} and \ref{smooth:dens}. The log-likelihood  of the model is
\begin{equation}\label{smooth:eq:loglik}
l_{\theta, m}(y,x)= \log\left[ p_{\epsilon|X} \big(y-m(\theta^\top x),x\big) p_X(x)\right].
\end{equation}
\begin{remark}\label{rem:Notations}
Note that under \eqref{smooth:eq:score_general_model}, we have $\epsilon=Y-m(\theta^\top X).$  For every function $b(e,x): \R\times \rchi \to \R$ in $L_2(P_{\epsilon, X}),$ there exists an ``equivalent'' function $\tilde b (y,x): \R\times \rchi\to \R$ in $L_2(P_{\theta,m})$  defined as $\tilde b(y,x):= b(y-m(\theta^\top x), x)\in L_2(P_{\theta,m})$. In this section, we use the function arguments $(e,x)$ $(L_2(P_{\epsilon, X}))$ and $(y,x)$ $(L_2(P_{\theta,m}))$ interchangeably.
\end{remark}
For  $\eta \in S^{d-2} \subset \R^{d-1}$, consider the path defined in~\eqref{smooth:eq:path_para}. Note that this is  a valid path through $\theta$ as $\zeta_0(\theta,\eta)=\theta$.   The score function for this submodel (the parametric score) is
\begin{equation}
\left. \frac{\partial l_{\zeta_s(\theta,\eta), m} (y,x) }{\partial s}\right\vert_{s=0}= \eta^\top S_{\theta,m}(y,x), \text{ where }S_{\theta,m}(y,x):=   -\frac{p'_{\epsilon|X} \big(y-m(\theta^\top x),x\big)}{p_{\epsilon|X} \big(y-m(\theta^\top x),x\big)}m^\prime(\theta^\top x) H_\theta ^\top x. \label{smooth:eq:paraScore}
\end{equation} We now define a parametric submodel for  the unknown nonparametric components: 
\begin{align} \label{smooth:eq:nonp_path}
\begin{split}
m_{s,a}(t)&=m(t) - s a(t),\\
p_{\epsilon|X;s, b}(e, x) &= p_{\epsilon|X}(e, x) (1 + s b(e, x)),\\
p_{X; s,q}(x) &=p_X(x)(1+s q(x)),
\end{split}
\end{align} where $s\in \R$, $b: \R \times \rchi\rightarrow\R$ is a bounded function such that $\E(b(\epsilon, X) |X)=0$ and $\E(\epsilon b(\epsilon, X) |X)=0$,  $a \in \mathcal{S}$ such that $J(a)< \infty$ and $q: \rchi \rightarrow \R$ is a bounded function such that $\E(q(X))=0.$ Consider the following  parametric submodel of \eqref{smooth:simsl},
\begin{equation} \label{smooth:eq:semi_path}
s \mapsto ( \zeta_s(\theta,\eta), \,m_{s,a}, \,p_{\epsilon|X;s, b}, \,  p_{X; s,q}(x) )
\end{equation}
where $\eta \in  S^{d-2}$. Differentiating the log-likelihood of the submodel in \eqref{smooth:eq:semi_path} with respect to $s$,  we get that the score along the submodel in \eqref{smooth:eq:semi_path} is
$$  \eta^\top S_{\theta,m} (y,x) + \frac{p'_{\epsilon|X} \big(y-m(\theta^\top x),x\big)}{p_{\epsilon|X} \big(y-m(\theta^\top x),x\big)} a(\theta^\top x) + b(y-m(\theta^\top x), x) +q(x).$$
It is now  easy to see that the nuisance tangent space, denoted by $\Lambda$, of the model is
\begin{align} \label{smooth:eq:nuisance_socre}
\begin{split}
\Lambda:= \overline{\mathrm{lin}}\,& \Big\lbrace f \in L_2(P_{\epsilon,X}): f(e,x) = \frac{p'_{\epsilon|X} \big(e,x\big)}{p_{\epsilon|X} \big(e,x\big)} a(\theta^\top x) +b(e, x) + q(x), \text{ where} \\
& \quad a\in \Ss, J(a)<\infty, b: \R \times \rchi\rightarrow \R \text{ and } q: \rchi\rightarrow\R \text{ are bounded functions},\\ & \hspace{.5in} \E(\epsilon b(\epsilon,X)|X)=0,\, \E(b(\epsilon,X)|X)=0, \text{ and } \E(q(X))=0 \Big\rbrace,
 \end{split}
\end{align}
where for any set $A\subset L_2(P_{\theta,m}),$ $\overline{\mathrm{lin}}\,A$ denotes the closure in $L^2(P_{\theta,m})$ of the linear span of functions in $A$; see \cite{Newey90} for a review of the construction of the nonparametric tangent set as a closure of scores of parametric submodels of the nuisance parameter.
By Corollary A.1 of \cite{Gyorfi02}, we have that the class of infinitely often differentiable functions on $D$ is dense in $L_2(\mathbf{m})$, where $\mathbf{m}$ denotes the Lebesgue measure on $D$. Thus we have that
\begin{align*}
\overline{\mathrm{lin}}\{ a\in \Ss : J(a)<\infty\}&= \{a: D \rightarrow \R|\, a\in L_2(\mathbf{m})\},\\
\overline{\mathrm{lin}}\{ q: \rchi \rightarrow \R|\, q \text{ is a bounded function and }  \E(q(X))=0\}&= \{ q \in L_2(P_X)|\,  \E(q(X))=0\},
\end{align*}
and
\begin{align*}
\overline{\mathrm{lin}}\{ b: \R \times\rchi \rightarrow \R|\, b \text{ is a }&\text{bounded function, } \E(\epsilon b(\epsilon,X)|X)=\E(b(\epsilon,X)|X)=0\}\\
&=\{  b \in L_2(P_{\epsilon,X})  |\, \E(\epsilon b(\epsilon,X)|X)=\E(b(\epsilon,X)|X)=0\}.
\end{align*}
 Thus, it is easy to see that under assumptions \ref{smooth:a0}--\ref{smooth:a6}, \ref{smooth:b3}, and \ref{smooth:dens},  the nuisance tangent space of \eqref{smooth:simsl} is
\begin{align*}
\Lambda= \Big\lbrace & f\in L_2(P_{\epsilon,X}) :\, f(e,x) = \frac{p'_{\epsilon|X} \big(e,x\big)}{p_{\epsilon|X} \big(e,x\big)} a(\theta^\top x) +b(e, x) + q(x), \text{ where }\\
& \hspace{.3in}   a\in L_2(\mathbf{m}), b \in L_2(P_{\epsilon,X}), q\in L_2(P_X), \E(\epsilon b(\epsilon,X)|X)=0,\\
&\hspace{1.5in} \E(b(\epsilon,X)|X)=0, \text{ and } \E(q(X))=0 \Big\rbrace;
\end{align*}
  see Theorem 4.1 in \cite{NeweyStroker93} and Proposition 1 of \cite{MaZhu13} for a similar nuisance tangent space. Observe that the efficient score  is  the $L_2(P_{\epsilon,X})$ projection of $S_{\theta,m}(y,x)$ onto $\Lambda^{\perp}$, where $\Lambda^\perp$ is the orthogonal complement of $\Lambda$ in $L_2(P_{\epsilon,X})$. \cite{NeweyStroker93} and \cite{MaZhu13} show that
  \begin{align}   \label{smooth:eq:Lamperp_char}
  \Lambda^\perp = \Big\{& f\in L_2(P_{\epsilon,X}):\, f(e,x)= \big[g(x)- \E\big(g(X)|\theta^\top X=\theta^\top x\big)\big] e, \text{ for some } g:\rchi \rightarrow\R\Big\}.
  \end{align}

 Using calculations similar those in Proposition 1 in  \cite{MaZhu13}, it can be shown that
{ \be \label{smooth:eq:trueEffScore}
\Pi(S_{\theta,m}  |\Lambda^\perp)(y,x)=  \frac{(y-m(\theta^\top x))}{\sigma^2(x)} m^\prime(\theta^\top x)  H_\theta^\top \left\lbrace x -  \frac{\E(\sigma^{-2}(X)X|\theta^\top X=\theta^\top x)}{\E(\sigma^{-2}(X)|\theta^\top X=\theta^\top x)}\right\rbrace,
\ee}
where for any $f\in L_2(P_{\epsilon,X})$,  $\Pi(f|\Lambda^\perp)$ denotes the $L_2(P_{\epsilon,X})$ projection of $f $ onto the space $\Lambda^\perp$. $\Pi(S_{\theta,m}|\Lambda^\perp)$ is sometimes denoted by $S^{eff}_{\theta,m}$.
It is important to note that the optimal estimating equation depends on $\sigma^2(\cdot)$. Since in the semiparametric model $\sigma^2(\cdot)$ is left unspecified, it is unknown. Without additional assumptions, nonparametric estimators of $\sigma^2(\cdot)$  have a slow rate of convergence to $\sigma^2(\cdot)$, especially if $d$ is large. Thus if we substitute $\hat{\sigma}(x)$ in the efficient score equation, the solution of the modified score equation would lead to poor finite sample performance; see \cite{Tsiatis06}.

To focus our presentation on the main concepts, briefly consider the case when  $\sigma^2(\cdot)\equiv\sigma^2$. In this case the efficient score $\Pi(S_{\theta,m}  |\Lambda^\perp)(y,x)$ is
\begin{equation}\label{eq:eff_score_0}
 \frac{1}{\sigma^{2}}(y-m(\theta^\top x)) m^\prime(\theta^\top x)  H_\theta^\top \left\lbrace x - h_\theta(\theta^\top x)\right\rbrace,
\end{equation}
where $h_\theta(\theta^\top x)$ is defined in \eqref{smooth:eq:h_beta}. Asymptotic normality and efficiency of $\hat{\theta}$ would follow if we can show that $(\hat{m},\hat{\theta})$  satisfies the efficient score equation \textit{approximately}, i.e.,
\begin{equation}\label{eq:eff_app_11}
\p_n \left[  \frac{1}{\sigma^{2}}(Y-\hat{m}(\hat\theta^\top X)) \hat{m}^\prime(\hat{\theta}^\top X)  H_{\hat{\theta}}^\top \big\{X - h_{\hat{\theta}}(\hat{\theta}^\top X)\big\}\right]= o_p(n^{-1/2})
\end{equation}
and a class of functions formed by the efficient score indexed by  $(\theta, m)$ in a ``neighborhood'' of $(\theta_0,m_0)$ satisfies some ``uniformity'' conditions, e.g., it is a Donsker class. We formalize this notion of efficiency in Theorem \ref{smooth:thm:Main_rate} below.
\subsection{Efficiency of  \texorpdfstring{$\hat{\theta}$}{Lg}}\label{smooth:sec:eff_theta}
 \begin{thm} \label{smooth:thm:Main_rate}
Assume that $(Y,X)$ satisfies \eqref{smooth:simsl} and assumptions \ref{smooth:a0}--\ref{smooth:a6}, \ref{smooth:b3}, and \ref{smooth:dens} hold. Define
\be \label{smooth:eq:EffScore}
\tilde{\ell}_{\theta,m}(y,x):=\big(y-m(\theta^\top x)\big) m^\prime(\theta^\top x)  H_\theta^\top \left\lbrace x - h_\theta(\theta^\top x)\right\rbrace.
\ee
If $V_{\theta_0,m_0}:=P_{\theta_0,m_0}(\tilde{\ell}_{\theta_0,m_0} S^{\top}_{\theta_0,m_0})$ is a nonsingular matrix in $\R^{(d-1) \times (d-1)}$, then
\begin{equation}\label{smooth:eq:globalEff}
\sqrt{n} (\hat{\theta}- \theta_0)\stackrel{d}{\rightarrow}   N(0,H_{\theta_0} V_{\theta_0,m_0}^{-1}  \tilde{I}_{\theta_0,m_0}(H_{\theta_0} V_{\theta_0,m_0}^{-1})^\top),
\end{equation}
where  $\tilde{I}_{\theta_0,m_0} := P_{\theta_0,m_0} (\tilde{\ell}_{\theta_0,m_0}\tilde{\ell}^\top_{\theta_0,m_0})$. If we further assume that $\sigma^2(\cdot)\equiv\sigma^2$ and  if the efficient information matrix, $\tilde{I}_{\theta_0,m_0}$, is nonsingular, then $\hat{\theta}$ is an efficient estimator  of $\theta_0$, i.e.,
\begin{equation}\label{smooth:eq:localeffestim}
\sqrt{n} (\hat{\theta}-\theta_0) \stackrel{d}{\rightarrow} N(0, \sigma^4 H_{\theta_0}\tilde{I}^{-1}_{\theta_0,m_0} H_{\theta_0}^\top).
\end{equation}
 \end{thm}

\begin{remark} Note that  even if $\E(\epsilon^2|X)\not\equiv \sigma^2$, $\hat{\theta}$ is a consistent and asymptotically normal estimator of $\theta$. When the constant variance assumption provides a good approximation to the truth, estimators similar to $\hat{\theta}$ have been known to have high relative efficiency with respect to the optimal semiparametric efficiency bound; see Page 94 of \cite{Tsiatis06} for a discussion. When $\sigma^2(x)=V^2(\theta_0^\top x)$ for some unknown real-valued function $V,$ we can define a weighted PLSE as
$$(\tilde{m},\tilde{\theta}):=\argmin_{(m,\theta)\in\mathcal{S}\times\Theta} \frac{1}{n}\sum_{i=1}^n \hat{w}(x_i) \big(y_i-m( \theta^\top x_i)\big)^2 + \hat\lambda_n^2 J^2(m),$$
where $\hat{w}(x)$ is a consistent estimator of $V^{-2}(\theta_0^\top x)$. Theorem \ref{smooth:thm:Main_rate} can be easily generalized to show that $\tilde{\theta}$ is an efficient estimator of $\theta_0$ under this specific heteroscedastic structure.
\end{remark}

\begin{remark}
The asymptotic variance of $\sqrt{n} (\hat{\theta}-\theta_0)$ is the same as that obtained in Section 2.4 of \cite{HardleEtAl93} and~\cite{chang2010asymptotically} (under assumption~\ref{smooth:a4}). However both require stronger smoothness assumptions on $m_0$ for their estimators. 
\end{remark}

\begin{remark}\label{rem:Degenracy}
Observe that the variance of the limiting distribution (for both the heteroscedastic and homoscedastic models) is singular. This can be attributed to the fact that $\Theta$ is a Stiefel manifold of dimension $\R^{d-1}$ and has an empty interior in $\R^d.$
 \end{remark}
% \begin{remark}\label{smooth:rem:Hetero}
% The theorem here can be generalized to generalized to heteroscedastic model  by considering weighted least squares, with weights being the inverse of the variance function. The calculations will be similar to the homoscedastic model considered in Section \ref{smooth:sec:SemiInf} and all the results in the paper would hold.
% \end{remark}
% \begin{remark} \todo{do it}
% Confidence set for $\theta_0$.
% \end{remark}
\subsubsection{Proof of Theorem \ref{smooth:thm:Main_rate}}\label{sec:proof_main_thm} % without loss of generality

%For simplicity of notation and without loss of generality, we assume $\sigma^2=1.$
In the following we give a sketch of the proof of \eqref{smooth:eq:globalEff}.  Some of the steps are proved in the following sections.

\begin{enumerate}[label=\bfseries Step \arabic*]
\item In Theorem~\ref{smooth:thm:eff_equation} we will show that $(\hat{m},\hat{\theta})$  satisfy the efficient score equation \textit{approximately}, i.e.,\label{smooth:item:step0}
\begin{equation}\label{smooth:eq:app_score}
\sqrt{n} \p_n \tilde{\ell}_{\hat{\theta}, \hat{m}}= o_p(1).
\end{equation}
 \item In Section  \ref{smooth:thm:nobias_proof} of the supplementary material, we prove that $\tilde{\ell}_{\hat{\theta},\hat{m}}$ is unbiased in the sense of \cite{VdV02}, i.e.,
\begin{equation} \label{smooth:eq:nobias_main_th}
 P_{\hat{\theta}, m_0} \tilde{\ell}_{\hat{\theta},\hat{m}} =0.
\end{equation}  \label{smooth:item:step1}
Similar conditions have appeared before in proofs of asymptotic normality of the MLE (e.g., see \cite{MR1394975}) and the construction of efficient one-step estimators (see~\cite{MR913573}); see~Section 3 of \cite{MR1803168} for further discussion.
 \item We prove \label{smooth:item:step2}
\begin{equation} \label{smooth:eq:Emp_proc_zero}
 \g_n ( \tilde\ell_{\hat{\theta},\hat{m}}- \tilde\ell_{\theta_0,m_0}) =o_p(1)
\end{equation}
in Theorem \ref{smooth:thm:ConsistencyofG_n}. In  view of  \eqref{smooth:eq:app_score} and \eqref{smooth:eq:nobias_main_th} an equivalent formulation of \eqref{smooth:eq:Emp_proc_zero} is
\begin{equation}\label{smooth:eq:Final_eqpart1}
 \sqrt{n} (P_{\hat{\theta}, m_0} - P_{\theta_0, m_0}) \tilde \ell_{\hat{\theta},\hat{m}} =\g_n \tilde{\ell}_{\theta_0,m_0} + o_p(1).
\end{equation}
\item To complete the proof of \eqref{smooth:eq:globalEff}, it is enough to show that \label{smooth:item:step3}
\begin{equation} \label{smooth:eq:ParaScore_approx}
\sqrt{n} (P_{\hat{\theta}, m_0} - P_{\theta_0, m_0}) \tilde\ell_{\hat{\theta},\hat{m}} = \sqrt{n}V_{\theta_0,m_0}  H_{\theta_0}^\top (\hat{\theta}- \theta_0) + o_p(\sqrt{n} |\hat{\theta} -\theta_0|).
\end{equation}
A proof of slightly simplified version of \eqref{smooth:eq:ParaScore_approx} can be found in the proof of Theorem 6.20 of \cite{VdV02}. However, for the sake of completeness we give a proof  of \eqref{smooth:eq:ParaScore_approx} in Section \ref{smooth:app:proof_step3} of the supplementary material.

 \end{enumerate}
 \noindent Observe that  \eqref{smooth:eq:Final_eqpart1} and \eqref{smooth:eq:ParaScore_approx} imply
 \begin{align} \label{smooth:eq:final_thm}
  \begin{split}
\sqrt{n}V_{\theta_0,m_0}  H_{\theta_0}^\top(\hat{\theta}- \theta_0)  ={}&\g_n \tilde\ell_{\theta_0,m_0} + o_p(1+\sqrt{n} |\hat{\theta}-\theta_0|),\\
 \Rightarrow  \sqrt{n} H_{\theta_0}^\top (\hat{\theta}- \theta_0) ={}& V_{\theta_0,m_0}^{-1}\g_n \tilde\ell_{\theta_0,m_0} + o_p(1) \stackrel{d}{\rightarrow} V_{\theta_0,m_0}^{-1} N(0, \tilde{I}_{\theta_0,m_0}).
 \end{split}
 \end{align}
The proof of the theorem will be complete  if we can show that  $$\sqrt{n}(\hat{\theta}- \theta_0)=H_{\theta_0} \sqrt{n} H_{\theta_0}^\top(\hat{\theta}- \theta_0) +o_p(1).$$
 Let $\hat{\eta}$ be the unique vector in $\R^{d-1}$ that satisfies the following  equation:
\begin{equation} \label{smooth:eq:local_hat}
\hat{\theta} =\sqrt{1-|\hat{\eta}|^2}\, \theta_0+ H_{\theta_0} \hat{\eta},
\end{equation}
note that such an $\hat{\eta}$ will always exists as $\hat{\theta} \stackrel{P}{\rightarrow} \theta_0.$ As $H_{\theta_0}^\top \theta_0=0$ and $H_{\theta_0}^\top H_{\theta_0}= \mathbb{I}_{d-1}$, pre-multiplying both sides of the previous equation  by $H_{\theta_0}^\top$ we get
\begin{equation}\label{smooth:eq:eta_hat}
\hat{\eta}= H_{\theta_0}^\top(\hat{\theta}- \theta_0).
\end{equation}
Substituting the above expression of $\hat\eta$ in  \eqref{smooth:eq:local_hat} and subtracting $\theta_0$ from both sides of \eqref{smooth:eq:local_hat} we get
\begin{equation}
\hat{\theta}-\theta_0= \Big[ \sqrt{1-|H_{\theta_0}^\top(\hat{\theta}- \theta_0)|^2}-1\Big] \theta_0+ H_{\theta_0} H_{\theta_0}^\top(\hat{\theta}- \theta_0).
\end{equation}
By  \eqref{smooth:eq:final_thm} we have that $\sqrt{n} H_{\theta_0}^\top (\hat{\theta}- \theta_0)=O_p(1).$ Moreover, note that $\sqrt{1-x^2}-1 =O(x^2),$ as $x\rightarrow0.$ Combining the above facts, we get
\begin{align*}
\sqrt{n}(\hat{\theta}- \theta_0)&=\sqrt{n} O_p(|H_{\theta_0}^\top(\hat{\theta}- \theta_0)|^2) + \sqrt{n} H_{\theta_0}H_{\theta_0}^\top(\hat{\theta}- \theta_0)\\
&= H_{\theta_0} \sqrt{n} H_{\theta_0}^\top(\hat{\theta}- \theta_0)+ O_p(n^{-1/2}).
\end{align*}
%By efficiency of  $\sqrt{n}H_{\theta_0}^\top (\hat{\theta}- \theta_0)$ we have that  $ H_{\theta_0} \tilde{I}_{\theta_0,m_0}^{-1} H_{\theta_0}^\top$ is the best variance that $\sqrt{n} (\hat{\theta}- \theta_0)$ can achieve asymptotically and thus is Fisher efficient.
Now we prove \eqref{smooth:eq:localeffestim}. Assume that $\sigma^2(\cdot)\equiv \sigma^2$. Observe that, by \eqref{smooth:eq:trueEffScore} and \eqref{smooth:eq:EffScore}, we have
\begin{align*}
 S_{\theta_0,m_0}&= \Pi(S_{\theta_0,m_0} |\Lambda^\perp) + \big(  S_{\theta_0,m_0}-\Pi(S_{\theta_0,m_0} |\Lambda^\perp)\big)\\
 &= \frac{1}{\sigma^2} \tilde{\ell}_ {\theta_0,m_0} + \big(  S_{\theta_0,m_0}-\Pi(S_{\theta_0,m_0}  |\Lambda^\perp)\big).
\end{align*}
 Thus \eqref{smooth:eq:localeffestim} follows from \eqref{smooth:eq:globalEff} by observing that
\begin{align*}
V_{\theta_0,m_0}&= P_{\theta_0,m_0} \big( \tilde{\ell}_{\theta_0,m_0} S^{\top}_{\theta_0,m_0}\big)= \frac{1}{\sigma^2} \tilde{I}_{\theta_0,m_0}.
\end{align*}

\subsubsection{``Least favorable'' path for \texorpdfstring{$m$}{Lg}}
\label{smooth:sec:App_least_fav_model}
We will now show that \eqref{smooth:eq:app_score} holds. Recall the definition \eqref{smooth:eq:path_para}. For any $(\theta,m)\in \Theta\times\{ m\in \Ss| J(m)<\infty\}$ and $\eta \in S^{d-2}$, let  $t \mapsto (\zeta_t(\theta,\eta) , \xi_t(\cdot;\theta,\eta,m))$ denote a path in $\Theta \times \{ m\in\Ss| J(m)<\infty\}$ that goes through $(\theta,m)$, i.e., $(\zeta_0(\theta,\eta), \xi_0(\cdot;\theta,\eta,m))= (\theta,m)$; see \eqref{smooth:eq:LeastFavmodel} below for definition.  Recall that $(\hat{\theta},\hat{m})$ minimizes $\mathcal{L}_n(m,\theta,\hat{\lambda}_n).$ Hence, for every  $\eta \in S^{d-2}$, the function $t \mapsto \mathcal{L}_n( \xi_t(\cdot;\hat\theta,\eta,\hat{m}),\zeta_t(\hat\theta,\eta) ,\hat{\lambda}_n)$ is minimized at $t=0$. In particular, if the above function is differentiable in a neighborhood of $0$, then
\begin{equation}\label{smooth:eq:lf_1}
\frac{\partial}{\partial t} \mathcal{L}_n( \xi_t(\cdot;\hat\theta,\eta,\hat{m}),\zeta_t(\hat\theta,\eta) ,\hat{\lambda}_n) \bigg\vert_{t=0}=  0.
\end{equation}
Moreover if  $(\zeta_t(\hat\theta,\eta) , \xi_t(\cdot;\hat\theta,\eta,\hat{m}))$ satisfies
\begin{align}\label{smooth:eq:lf_2}
\begin{split}
\frac{\partial}{\partial t} \big(y-\xi_t(\zeta_t(\hat\theta,\eta)^\top x;\hat{\theta},\eta,\hat{m})\big)^2 \bigg\vert_{t=0} &=\eta^\top \tilde{\ell}_{\hat{\theta},\hat{m}}(y,x), \\
  \left.\frac{\partial }{\partial t}J^2(\xi_t(\cdot;\hat\theta,\eta,\hat{m})) \right\vert_{t =0} &=O_p(1).
 \end{split}
\end{align}
for all $\eta\in S^{d-2}$, then we get \eqref{smooth:eq:app_score} as $\hat\lambda_n^2 =o_p(n^{-1/2})$; see assumption \ref{smooth:a4}.

Observe that $\hat{\theta}$ is a consistent estimator of $\theta_0.$ As we are concerned with the path $t \mapsto \mathcal{L}_n( \xi_t(\cdot;\hat\theta,\eta,\hat{m}),\zeta_t(\hat\theta,\eta) ,\hat{\lambda}_n)$, we will try to construct a path for any  $(\theta,m)\in \{\Theta \cap B_{\theta_0}(r)\}\times\{ m\in \Ss| J(m)<\infty\}$ that satisfies the above requirements. For any set $A\subset \R$ and any $\nu>0$ let us define $A^\nu :=\cup_{a\in A} B_a(\nu)$ and  let $\partial A$ denote the boundary of $A$.  Fix $\nu>0$. By~\eqref{eq:D_r}, for every $\theta   \in \Theta \cap B_{\theta_0}(r),$ $\eta \in S^{d-2},$ and $t \in \R$  sufficiently close to zero, there exists a strictly increasing  function $\phi_{\theta,\eta, t} : D^\nu \rightarrow \R$ with
\begin{align}\label{smooth:eq:BdryCond}
\begin{split}
\phi_{\theta,\eta, t}(u) ={}& u, \quad u\in  D_\theta \\
 \phi_{\theta,\eta, t}(u+  (\theta- \zeta_t(\theta,\eta))^\top h_{\theta}(u)) ={}& u, \quad u\in \partial D,
 \end{split}
 \end{align}
 where  $h_\theta(u)$ and $\zeta_t( \theta,\eta)$ are defined in  \eqref{smooth:eq:h_beta} and \eqref{smooth:eq:path_para}, respectively. Furthermore, we can ensure that $\phi_{\theta,\eta, t}(u)$ is infinitely differentiable for $u \in D$ and that $\left. \frac{\partial}{\partial t} \phi_{\theta,\eta, t}\right|_{t=0}$ exists. Note that $\phi_{\theta,\eta, t}(D)=D$. Moreover, $\phi_{\theta,\eta,t}$ cannot be the identity function for $t\neq 0$ if $(\theta- \zeta_t( \theta,\eta))^\top h_{\theta}(u) \neq 0$ for $u \in\partial D.$
 % We define a path $(\zeta_t(\theta,\eta), \xi_t(\cdot;\theta,\eta, m))$ as $t$ varies in a neighborhood of $0$ in the model \eqref{smooth:simsl} such that $\xi_t(\theta,\eta, m)$ is $m$ at $t=0$. %We denote $$\frac{\partial \xi_t(x)}{\partial x}=\daleth_t(x), \qquad \text{and} \quad \frac{\partial g(x)}{\partial x}=f(x).$$
 Now, we can  define the following path through $m$:
 \begin{equation} \label{smooth:eq:LeastFavmodel}
 \xi_t(u;\theta,\eta,m):= m\circ \phi_{\theta,\eta,t}(u+ (\theta- \sqrt{1-t^2 |\eta|^2}\, \theta - t H_\theta\eta)^\top h_{\theta}(u)).
 \end{equation}
 The function $\phi_{\theta,\eta,t}$ helps us control the  partial derivative in the second equation of  \eqref{smooth:eq:lf_2}.  In the following theorem (proved in Appendix~\ref{smooth:app:proof_eff_equation}), we show that $( \zeta_t( \hat{\theta},\eta), \xi_t(\cdot; \hat{\theta},\eta,\hat{m}))$ is a path through $(\hat{\theta}, \hat{m})$ and satisfies \eqref{smooth:eq:lf_1} and \eqref{smooth:eq:lf_2}. Here $\eta$ is  the ``direction" for the path $ t \mapsto \zeta_t( \theta,\eta)$ and $(\eta,h_\theta(u)$) defines the ``direction" for the path $ t \mapsto \xi_t(\cdot; \theta,\eta,m)$.

 \begin{thm} \label{smooth:thm:eff_equation} Under assumptions \ref{smooth:a0},\ref{smooth:a1}, \ref{smooth:a4}, and \ref{smooth:b3}, $( \zeta_t( \hat{\theta},\eta), \xi_t(\cdot; \hat{\theta},\eta,\hat{m}))$ is a valid parametric submodel, i.e., $(\zeta_t( \hat{\theta},\eta), \xi_t(\cdot; \hat{\theta},\eta,\hat{m}))  \in \Theta\times\{m\in\Ss| J(m)<\infty\}$ for all  $t$ in some neighborhood of $0$. Moreover $(\zeta_t( \hat{\theta},\eta), \xi_t(\cdot; \hat{\theta},\eta,\hat{m}))$ satisfies \eqref{smooth:eq:lf_2} and   $\mathcal{L}_n( \xi_t(\cdot;\hat\theta,\eta,\hat{m}),\zeta_t(\hat\theta,\eta) ,\hat{\lambda}_n)$, as function of $t,$ is differentiable at $0$ and $\sqrt{n} \p_n \tilde{\ell}_{\hat{\theta}, \hat{m}}= o_p(1).$
 \end{thm}

 \subsubsection{Asymptotic equicontinuity of \texorpdfstring{$\tilde{\ell}_{\theta,m}$ at $(\theta_0,m_0)$}{Lg}}
For notational  convenience we define
\begin{equation}
K_1(x;\theta) :=H_\theta^\top (x- h_{\theta} (\theta^\top x)).\label{smooth:eq:K_def}
\end{equation}
With the  above notation, from \eqref{smooth:eq:EffScore} we have
\bee
\tilde{\ell}_{\theta,m} (y,x)= (y-m(\theta^\top x)) m^\prime(\theta^\top x)  K_1(x;\theta).
\eee

\begin{thm} \label{smooth:thm:ConsistencyofG_n}
Under assumptions \ref{smooth:a0}--\ref{smooth:a6}, \ref{smooth:b3}, and \ref{smooth:dens}, $\g_n (\tilde{\ell}_{\hat{\theta},\hat{m}} -\tilde{\ell}_{\theta_0,m_0})=o_p(1).$
\end{thm}
\begin{proof}

We divide  the proof Theorem~\ref{smooth:thm:ConsistencyofG_n} into two lemmas. First observe that
\begin{align}\label{smooth:eq:split_eqcont_l}
&\g_n (\tilde{\ell}_{\hat{\theta},\hat{m}} -\tilde{\ell}_{\theta_0,m_0})\nonumber\\
={}& \g_n \big[  \big(Y-\hat{m}(\hat{\theta}  ^\top X)\big) \hat{m}^\prime (\hat{\theta}^\top X) K_1(X;\hat{\theta}) - \big(Y-m_0(\theta_0^\top X)\big) m_0^\prime (\theta_0^\top X) K_1(X;\theta_0)\big]\nonumber\\
={}& \g_n \big[  \big(\epsilon +m_0(\theta_0^\top X) -\hat{m}(\hat{\theta}  ^\top X)\big) \hat{m}^\prime (\hat{\theta}^\top X) K_1(X;\hat{\theta})- \epsilon m_0^\prime (\theta_0^\top X) K_1(X;\theta_0)\big]\nonumber\\
\begin{split}
={}& \g_n \big[  \big(m_0(\theta_0^\top X) -\hat{m}(\hat{\theta}  ^\top X)\big) \hat{m}^\prime (\hat{\theta}^\top X) K_1(X;\hat{\theta})\big]\\
&\qquad +\g_n \big[ \epsilon \big( \hat{m}^\prime (\hat{\theta}^\top X) K_1(X;\hat{\theta})-m_0^\prime (\theta_0^\top X) K_1(X;\theta_0)\big)\big].
\end{split}
\end{align}

 The proof of Theorem \ref{smooth:thm:ConsistencyofG_n} will be complete, if we can show that both the terms in  \eqref{smooth:eq:split_eqcont_l} converge to 0 in probability. We begin with some definitions. Let $a_n$ be a sequence of real numbers such that $a_n\rightarrow \infty$ as $n\rightarrow \infty$  and $a_n\|\hat{m}-m_0\|_{D_0}^S =o_p(1).$ We can always find such a sequence $a_n,$ as  we have $\|\hat{m}-m_0\|_{D_{0}}^S =o_p(1)$ (see Theorem \ref{smooth:thm:cons}). For all $n \in \mathbb{N}$, define \footnote{The notations with $*$ denote the classes  of functions that do not depend on $n$ while the ones with $n$ denote shrinking neighborhoods around $(m_0, \theta_0)$.}
\begin{align}\label{smooth:eq:c_n}
\begin{split}
\mathcal{C}^{m*}_{M_1,M_2, M_3} &:=\Big\{ m \in \Ss: \|m\|_\infty< M_1,\;\|m^\prime\|_\infty<M_2, \text{ and } J(m)<M_3\Big\}, \\
\mathcal{C}^m_{M_1,M_2, M_3}(n) &:=\Big\{ m \in \mathcal{C}^{m*}_{M_1,M_2, M_3} : a_n \|m- m_0\|_{D_0}^S\le 1\Big\},\\
\mathcal{C}^\theta(n)&:=\Big\{\theta \in \Theta\cap B_{\theta_0}(1/2):  \hat{\lambda}_n^{-1/2} |\theta_0-\theta| \le1 \Big\}, \\
\mathcal{C}_{M_1,M_2, M_3}(n)&:= \Big\{ (m,\theta):\theta \in \mathcal{C}^\theta(n) \text{ and } m\in \mathcal{C}^m_{M_1,M_2, M_3}(n)\Big\},\\
\mathcal{C}^*_{M_1,M_2, M_3}&:= \Big\{ (m,\theta): \theta \in \Theta\cap B_{\theta_0}(1/2)  \text{ and } m\in \mathcal{C}^{m*}_{M_1,M_2, M_3}\Big\}.
\end{split}
\end{align}
 Let us  consider the first term of \eqref{smooth:eq:split_eqcont_l}. Fix $\delta>0$. For every fixed $M_1, M_2,$ and  $M_3$,
%\begin{align}
%&\p\Big( \big|\g_n \big[  (m_0 (\theta_0 ^\top X) -{\hat{m}} ({\hat{\theta}} ^\top X)) {\hat{m}}^\prime ({\hat{\theta}} ^\top X) K_1(X;\hat{\theta})\big]\big|> \delta\Big) \nonumber \\
%\le{}& \p\Big(   \big|\g_n \big[  (m_0 (\theta_0 ^\top X) -{\hat{m}} ({\hat{\theta}} ^\top X)) {\hat{m}}^\prime ({\hat{\theta}} ^\top X) K_1(X;\hat{\theta})\big]\big|> \delta, (\hat{m},\hat{\theta}) \in \mathcal{C}_{M_1,M_2, M_3}(n)\Big)\nonumber \\
%& \quad + \p\big( (\hat{m},\hat{\theta}) \notin \mathcal{C}_{M_1,M_2, M_3}(n)\big)\nonumber \\
%\le{}& \p\Big( \sup_{(m,\theta) \in \mathcal{C}_{M_1,M_2, M_3}(n)} \big|\g_n \big[  (m_0 (\theta_0 ^\top X) -m (\theta ^\top X)) m^\prime (\theta ^\top X) K_1(X;\theta)\big]\big|> \delta\Big)\nonumber \\
%&\quad + \p\big((\hat{m},\hat{\theta})\notin \mathcal{C}_{M_1,M_2, M_3}(n)\big).  \label{smooth:eq:step_2_part1}
%\end{align}
\begin{align}  \label{smooth:eq:insetoutset}
\begin{split}
&\p\Big( \big|\g_n \big[  {\hat{m}}^\prime\circ{\hat{\theta}} \,(m_0 \circ\theta_0  -{\hat{m}} \circ\hat{\theta} )  K_1(\cdot;\hat{\theta})\big]\big|> \delta\Big)  \\
\le{}& \p\Big(   \big|\g_n \big[   {\hat{m}}^\prime\circ{\hat{\theta}}\,  (m_0 \circ\theta_0 -{\hat{m}}\circ\hat{\theta})  K_1(\cdot;\hat{\theta})\big]\big|> \delta, (\hat{m},\hat{\theta}) \in \mathcal{C}_{M_1,M_2, M_3}(n)\Big)\\
& \quad + \p\big( (\hat{m},\hat{\theta}) \notin \mathcal{C}_{M_1,M_2, M_3}(n)\big) \\
\le{}& \p\Big( \sup_{(m,\theta) \in \mathcal{C}_{M_1,M_2, M_3}(n)} \big|\g_n \big[  m^\prime \circ\theta \,  (m_0 \circ\theta_0  -m \circ\theta )  K_1(\cdot;\theta)\big]\big|> \delta\Big)\\
&\quad + \p\big((\hat{m},\hat{\theta})\notin \mathcal{C}_{M_1,M_2, M_3}(n)\big).
\end{split}
\end{align}
Recall that $(\hat{m}, \hat{\theta})$ is a consistent estimator of  $(m_0, \theta_0)$  and $\|\hat{m}^\prime\|_\infty$ is $O_p(1)$; see Theorem \ref{smooth:thm:cons}. Furthermore, we have that  both $\|\hat{m}\|_\infty$ and $J(\hat{m})$ are  $O_p(1)$ (see Theorem \ref{smooth:thm:mainc})  and $\hat{\lambda}_n^{-1/2}|\hat{\theta}-\theta_0|=o_p(1)$ (see Theorem \ref{smooth:thm:ratest}).  Thus  for any $\varepsilon >0,$ there exists $M_1,M_2,$ and $M_3$ (depending on $\varepsilon$) such that
$$\p\left((\hat{m},\hat{\theta}) \notin\mathcal{C}_{M_1,M_2, M_3}(n)\right) \le \varepsilon,$$ for all sufficiently large $n.$
%Hence, to prove $$\p\Big( |\g_n (m_0- \hat{m})({\theta_0}^\top X) U_{\hat{\theta},\hat{m}}(X)|\ > \delta\Big) \le \varepsilon$$
Hence, it is enough to show that for the above choice of $M_1,M_2,$ and  $M_3$,  we have
$$ \p\Big( \sup_{(m,\theta) \in \mathcal{C}_{M_1,M_2, M_3}(n)} \big|\g_n \big[  m^\prime \circ\theta\, (m_0\circ\theta_0  -m\circ\theta )  K_1(\cdot;\theta)\big]\big|> \delta\Big) \le \varepsilon$$ for sufficiently large $n.$ Lemma~\ref{smooth:lem:first_term} (proved in Section \ref{smooth:app:proof_first_term} of the supplementary material) shows this. Moreover, Lemma~\ref{smooth:lem:second_term} (proved in Section \ref{smooth:app:proof_second_term} of the supplementary material) shows that the second term  on the right hand side of \eqref{smooth:eq:split_eqcont_l} converges to zero in probability. Thus our proof is complete.
\end{proof}

\begin{lemma}\label{smooth:lem:first_term}
Fix $M_1,M_2,M_3,$ and $ \delta>0.$ For $n\in \mathbb{N},$ let us define two classes of functions from $\rchi$ to $\R^{d}$
\begin{align}\label{smooth:eq:g_beta0_U}
\mathcal{D}_{M_1, M_2,M_3}(n) &:=\left\lbrace m^\prime \circ\theta (m_0 \circ\theta_0 -m\circ\theta)  K_1(\cdot;\theta): (m,\theta) \in \mathcal{C}_{M_1,M_2, M_3}(n) \right\rbrace,\nonumber\\
\mathcal{D}^*_{M_1, M_2,M_3} &:= \left\lbrace  m^\prime \circ\theta (m_0 \circ\theta_0 -m\circ\theta)  K_1(\cdot;\theta): (m,\theta) \in \mathcal{C}^*_{M_1,M_2, M_3} \right\rbrace.
\end{align}
$\mathcal{D}_{M_1, M_2,M_3}(n)$ is a Donsker class and
\begin{equation}\label{smooth:eq:g_beta0_U_envelope}
\sup_{f \in \mathcal{D}_{M_1, M_2,M_3}(n)} \|f\|_{2,\infty}  \le 2TM_2 (a_n^{-1}+T M_2 \hat\lambda_n^{1/2})=:D_{M_1, M_2,M_3}(n).
\end{equation}
Moreover, $J_{[\,]}(\gamma,\mathcal{D}_{M_1, M_2,M_3}(n),\|\cdot\|_{2,2}) \lesssim \gamma^{1/2}$, where for any class of functions $\mathcal{F},$ $J_{[\,]}$ is the entropy integral (see e.g., Page 270, \cite{VdV98}) defined as \be \label{smooth:eq:def_J}
J_{[\,]}(\delta,\mathcal{F},\|\cdot\|_{2,2}):=\int_{0}^\delta  \sqrt{\log N_{[\,]}(t, \mathcal{F}, \|\cdot\|_{2,2})} dt.
\ee  Finally, we have
\[ \p\bigg( \sup_{f \in \mathcal{D}_{M_1, M_2,M_3}(n) }|\g_n f| > \delta \bigg) \rightarrow 0 \qquad \text{as } n\rightarrow \infty.\]
\end{lemma}

%\g_n \big[ \epsilon \big( \hat{m}^\prime (\hat{\theta}^\top X) K_1(X;\hat{\theta})-m_0^\prime (\theta_0^\top X) K_1(X;\theta_0)\big)\big]
\begin{lemma}\label{smooth:lem:second_term}
Let us define $U_{\theta,m} : \rchi \to \R^{d-1}$, $U_{\theta,m}(x):= m^\prime (\theta^\top x) K_1(x;\theta)$.
Fix $M_1,M_2,M_3,$ and $ \delta>0.$ For $n\in \mathbb{N},$ let us define
\begin{align}\label{smooth:eq:W_beta0_U}
\w_{M_1, M_2,M_3}(n) &:=\left\lbrace    U_{\theta,m}-U_{\theta_0,m_0}: (m,\theta) \in \mathcal{C}_{M_1,M_2, M_3}(n) \right\rbrace,\nonumber\\
\w^*_{M_1, M_2,M_3} &:= \left\lbrace    U_{\theta,m}-U_{\theta_0,m_0}: (m,\theta) \in \mathcal{C}^*_{M_1,M_2, M_3} \right\rbrace.
\end{align}
Then $\w_{M_1, M_2,M_3}(n)$ is a Donsker class such that
\begin{equation}\label{smooth:eq:W_beta0_U_envelope}
 \sup_{f\in \w_{M_1, M_2,M_3}(n)} \|f\|_{2,\infty}\le \big[ 2 T^{3/2} M_3 \hat{\lambda}_n^{1/4}+  2 T a_n^{-1} +M_2 (2T+\bar{M}) \hat{\lambda}_n^{1/2}\big]=: W_{M_1, M_2,M_3}(n).
\end{equation}
Moreover,
$J_{[\,]}(\gamma,\w_{M_1, M_2,M_3}(n),\|\cdot\|_{2,2}) \lesssim \gamma^{1/2}$.  Hence, as $n\rightarrow \infty,$ we have
\be\label{smooth:eq:second_part_main}
 \p\Big(\Big| \g_n \big[ \epsilon \big( U_{\hat{\theta},\hat{m}}-U_{\theta_0,m_0}\big)\big]\Big| > \delta \Big) \rightarrow 0.
\ee
\end{lemma}

\section{Simulation study}\label{smooth:sec:simul}
To investigate the finite sample performance of $(\hat{m},\hat{\theta}),$ we carry out several simulation experiments. We also compare the finite sample performance of the proposed estimator with the EFM estimator (estimating function method \cite{cuietal11}), the EDR estimator (effective dimension reduction \cite{Hristacheetal01}), and the estimator proposed in~\cite{MR2514187} (denoted by WY). \cite{cuietal11} compares the performance of the EFM estimator to existing estimators such as the refined minimum average variance estimator (rMAVE) \cite{Xiaetal02} and the EDR estimator and argues that EFM has improved overall performance compared to existing estimators. Thus we do not include the rMAVE estimator in our simulation study. The code to compute the EDR estimates can be found in the R package \texttt{EDR}. Moreover, the authors of \cite{cuietal11} and~\cite{MR2514187} kindly provided us with the R codes to evaluate the EFM and the WY estimators, respectively.  The codes used to implement our procedure are available in the \texttt{simest} package in R; see \citet{simest}. In what follows, we chose the penalty parameter $\hat{\lambda}_n$  for the PLSE through generalized cross validation (GCV), i.e., choose $\hat\lambda_n$ by minimizing $\text{GCV}: \R \rightarrow\R$
\[\text{GCV}(\lambda):= \frac{Q_n(\hat{m}_\lambda,\hat{\theta}_\lambda)}{1-n^{-1} \text{trace}(A(\lambda))},\] where $(\hat{m}_\lambda,\hat{\theta}_\lambda):= \argmin_{(m,\theta)\in\mathcal{S}\times\Theta}\mathcal{L}_n(m,\theta; \lambda)$  and $A(\lambda)$ is the {\it hat} matrix for $\hat{m}_\lambda$ (see e.g., Sections 3.2 and 3.3  of~\cite{greensilverman94} for a detailed description of $A(\lambda)$ and its connection to the GCV); see  \cite{RuppertBook} for an extensive discussion on why the GCV is an attractive choice for choosing the penalty parameter in the single index model. We choose $\hat{\lambda}_n$ by minimizing the GCV score over a grid of values that satisfy assumption~\ref{smooth:a4}.  For all the other methods considered in the paper we have used the suggested values of tuning parameters. In the following, we consider three different data generating mechanisms. The codes used for the simulation examples can be found at~\url{http://stat.ufl.edu/~rohitpatra/research}.
% \todo[inline]{More about GCV}
% ; see \cite{RuppertBook} for extensive discussion on why GCV is an attractive choice for choosing the penalty parameter.
% and study the finite sample behavior of the proposed estimates.
 \begin{figure}[h!]
\centering
\includegraphics[width= .5\textwidth]{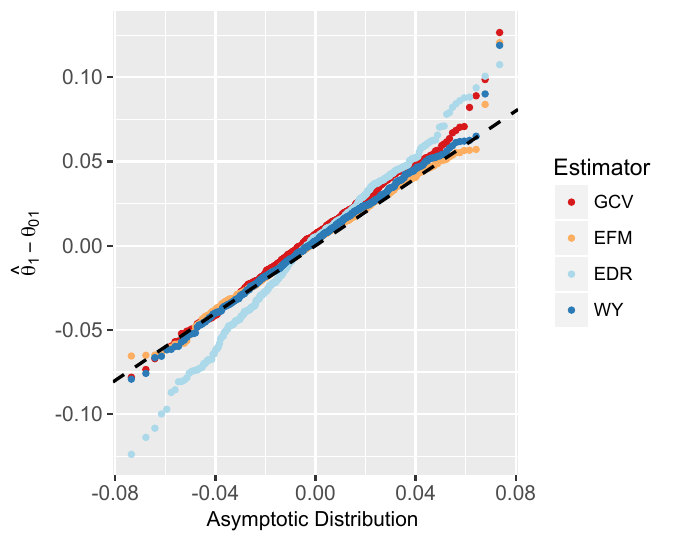}
\caption[QQ-plot.]{Quantile Quantile plot plot $\hat{\theta}_1 -\theta_0$ from 500 replications  with the true asymptotic distribution  of the $\hat{\theta}_{1}-\theta_{0,1}$ on the X-axis when we have 500 i.i.d.~samples from \eqref{eq:ex_1}.}
\label{smooth:fig:ex1} %% label for entire figure
\end{figure}

\subsection{A simple model}
We start with a simple model.  Assume that $(X_1, X_2) \in \R^2$, $X_1 \sim \text{Uniform}[-2,2]$,  $X_2\sim \text{Uniform}[0,1]$, $\epsilon\sim N(0, .5^2)$, and
\begin{equation}\label{eq:ex_1}
Y=( \theta_0^\top X)^2+\epsilon,  \text{ where } \theta_0=(1, -1)/\sqrt{2}.
\end{equation}

Observe that for this example,  $H_{\theta_0}^\top=[1,1]/\sqrt{2}$ (see Section~\ref{smooth:sec:lem:h_lip_proof} of the supplementary file) and the analytic expression of the efficient information is
\[ \tilde{I}_{\theta_0,m_0}=4 \Var(\epsilon) \mathbb{E}\Big( \theta_0^\top X H_{\theta_0}^\top \big[ X - \E\big(X| \theta_0^\top X\big)\big]\Big)^2=4 \Var(\epsilon) \mathbb{E}\left| (\theta_0^\top X)^2  \big[H_{\theta_0}^\top \Var(X | \theta_0^\top X\big) H_{\theta_0}\big]\right|. \] Using the above expression, we calculated the  asymptotic variance of $\sqrt{n} (\hat{\theta}_1 -\theta_{0,1})$ to be $0.328$. Figure~\ref{smooth:fig:ex1}  provides numerical evidence of asymptotic normality of the estimators discussed in this section. In this example, it is worth noting that the EFM and the WY have the better overall performance than the PLSE and the EDR.

\subsection{Dependent covariates}\label{sec:dep_cov}
We now consider a simulation scenario where covariates are dependent and the predictor $X \in \R^6$ contains discrete components. More precisely, $(X_1, \ldots, X_6)$ is generated according to the following law: $X_1 \sim \text{Uniform} [-1,1]$, $ X_2 \sim \text{Uniform}[-1,1]$, $ X_3:=0.2 X_1+ 0.2 (X_2+2)^2+0.2 Z_1$, $X_4:=0.1+ 0.1(X_1+X_2)+0.3(X_1+1.5)^2+  0.2 Z_2$, $X_5 \sim \text{Bernoulli}(\exp(X_1)/\{1+\exp(X_1)\})$, and  $ X_6 \sim \text{Bernoulli}(\exp(X_2)/\{1+\exp(X_2)\}).$ Here $Z_1$ and $Z_2$ are two Uniform$[-1,1]$ random variables independent of $X_1$ and $X_2$. Finally, we let 
\[ Y=  (\theta_0^\top X)^2+\epsilon, \]
 where $\theta_0$ is $(1.3, -1.3, 1, -0.5, -0.5, -0.5)/\sqrt{5.13}$.
In the following, we consider three different scenarios based on different error distributions:
\begin{enumerate}[label=\bfseries (2.\arabic*)]
    \item$\epsilon \sim N(0,1),$   \hfill (Homoscedastic, Gaussian Error)  \label{smooth:item:HomoGau}
    \item$\epsilon|X \sim N\left(0, \log(2+ (X^\top \theta_0)^2)\right),$ \hfill (Heteroscedastic, Gaussian Error) \label{smooth:item:HetereoGau}
    \item  $\epsilon| \xi \sim (-1)^\xi \text{Beta}(2,3),$ where $\xi\sim \text{Ber}(.5).$ \hfill (Homoscedastic, Non-Gaussian Error)  \label{smooth:item:HomoBeta}
\end{enumerate}
\begin{figure}[h!]
\vspace{0cm}
\centering
\includegraphics[width=\textwidth]{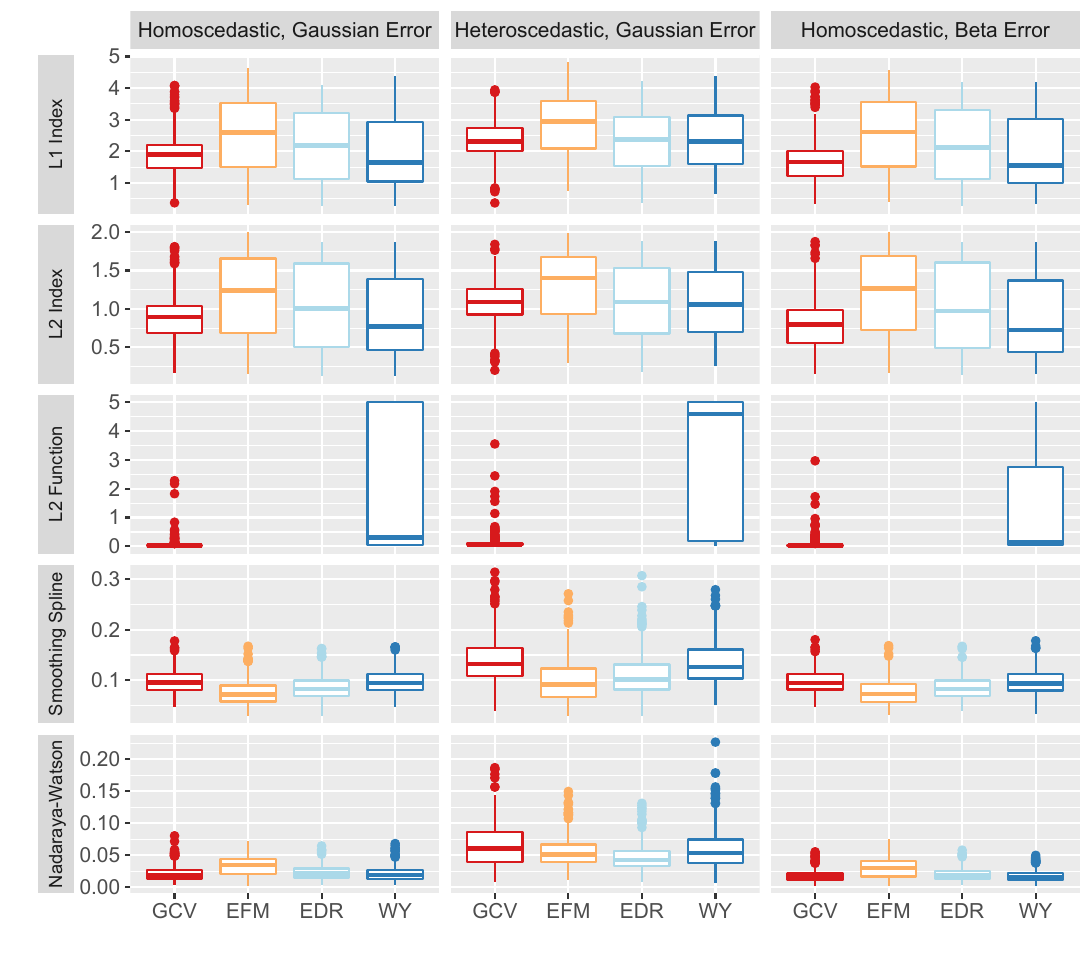}
\caption[Box plots of $L_1$ error of estimates of $\theta_0$ for Section~\ref{sec:dep_cov}.]{ Box plots (over 500 replications) of various errors based on $200$ observations from models \ref{smooth:item:HomoGau}, \ref{smooth:item:HetereoGau}, and \ref{smooth:item:HomoBeta} in the left, the middle, and the right columns, respectively. First two rows display $L_1$ and $L_2$ errors of estimates of $\theta_0$. The third row corresponds to $\|\hat{m}\circ\theta_0-m_0\circ\theta_0\|_{n}$ for the estimators proposed in~Section~\ref{smooth:sec:prelim} and~\cite{MR2514187}. The fourth and fifth rows corresponds to $\|\tilde{m}\circ\theta_0-m_0\circ\theta_0\|_{n}$ for one-dimensional smoothing splines and Nadaraya-Watson estimators based on the estimated index $\{(\hat{\theta}^{\top}X_i, Y_i), 1\le i\le n\}$, respectively. } %% label for entire figure
\label{smooth:fig:ex2}
\vspace*{0cm}
\end{figure}

 Our proposed method and~\cite{MR2514187} provide estimators for both the link function and the index parameter.  In the third row of Figure~\ref{smooth:fig:ex2}, we display the box plot of the in-sample $L_2(\p_n)$ loss ($\|\hat{m}\circ\theta_0^\top -m_0\circ\theta_0^\top \|_{n}$) for the PLSE and the WY estimators. EFM and EDR are not included as they do not provide estimators for the link function. However,  results of~\cite{gu2015oracally} show that for any root-${n}$ consistent estimator $\hat\theta$ of the index estimator,  the kernel (or Nadaraya-Watson) regression estimator on the data $\{(\hat{\theta}^{\top}X_i, Y_i), 1\le i\le n\}$ is asymptotically indistinguishable from the kernel estimator based on $\{(\theta_0^{\top}X_i, Y_i), 1\le i\le n\}$. This oracle type property led us to compute the estimators  of the link function based the  data $\{(\hat{\theta}^{\top}X_i, Y_i), 1\le i\le n\}$ for  GCV, EFM, EDR, and WY. The plot of the error (see fifth row of Figure~\ref{smooth:fig:ex2}) in  the estimation of the link function based on the Nadaraya-Watson estimator\footnote{Here we used the~\texttt{np} package \cite{np} to compute the bandwidth choice for the nonparametric regression estimator.}  provides numerical confirmation of the oracle property proved in~\cite{gu2015oracally}.  We also estimate the one-dimensional link function based on smoothing splines\footnote{We used \texttt{smooth.spline} command in \texttt{R} and choose $\lambda$ by the GCV procedure proposed in~\cite{greensilverman94}.} applied to the data $\{(\hat{\theta}^{\top}X_i, Y_i), 1\le i\le n\}$; see fourth row of Figure~\ref{smooth:fig:ex2}. The results of~\cite{gu2015oracally} do not directly imply a similar oracle type phenomenon for smoothing splines based estimators. However, the fourth row of Figure~\ref{smooth:fig:ex2} provides some numerical evidence for this oracle type property for the smoothing splines estimators.  The proof of the oracle type property developed in~\cite{gu2015oracally} crucially uses the smoothness  of the Nadaraya-Watson estimator (as a function of the index)  and we have not  been able to extend it to the case of smoothing splines estimators in single index models\footnote{This is a very interesting research direction and we plan to study it in the near future.}.

 The relative poor performance of EDR, EFM, and WY in estimating $\theta_0$ can possibly be attributed to the dependency between covariates.     Scenarios \ref{smooth:item:HomoGau} and \ref{smooth:item:HetereoGau} are similar to simulation scenarios considered in  \cite{LiPatilea15} and \cite{MaZhu13}. The codes to compute the estimator proposed in~\cite{LiPatilea15} were not available to us.
\subsection{High dimensional covariates}\label{sec:HDC}
For the final simulation scenario, we consider a setting similar to that of Example 4 in \citet[Section 3.2]{cuietal11}. We consider $d$-variate covariates for $d=10, \, 50,$ and $ 100.$ For each $d$, we assume that $X \sim \text{Uniform}[0,5]^d,$ $\epsilon \sim N(0, 0.2^2),$ $\theta_0=(2,1, \textbf{0}_{d-2})^\top/\sqrt{5}$, and have $400$ observations from the following model:
\begin{equation}\label{eq:ex3}
Y=\sin(a X^\top \theta_0) +\epsilon, \text{ where } a= \pi/2, 3\pi/4, \text{ and }3\pi/2.
\end{equation}
Note that here a higher value of $a$ represents a more oscillating link function.
 Figure~\ref{fig:Ex:sin} summarizes the finite sample performance of the estimators considered in this section. The performance of all the estimators worsen as the $a$ increases. When $a$ is $\pi/2$ or $3\pi/4$, GCV significantly outperforms the estimators considered in the simulation study. The IQR bars for the GCV in the first two panels of~Figure~\ref{fig:Ex:sin} are not visible because they are very small (relative to the scale of the plot).

\begin{figure}[h!]
\vspace{0cm}
\centering
\includegraphics[width=\textwidth]{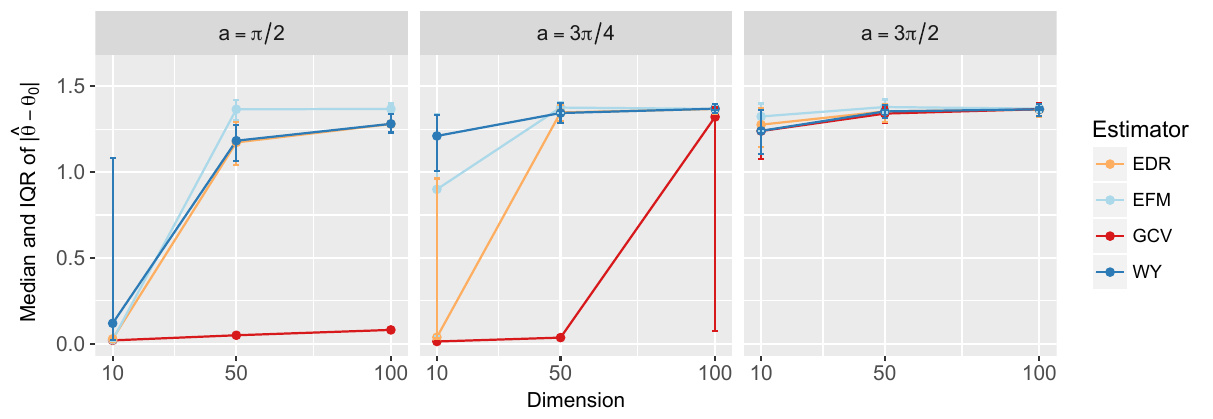}
\caption[Finite sample performance for example 3]{\label{tab:Ex:sin}  The quartiles of $ |\hat{\theta} -\theta_{0}|$ from 500 replications for $n=400$ from \eqref{eq:ex3}. }\label{fig:Ex:sin}
\vspace*{0cm}
\end{figure}

\section{Real data analysis} \label{sec:real_data}
% In this section we apply the methodology developed to two real datasets to fit an appropriate single index model and compare the performance of the model with earlier works.
\subsection{Car mileage data}\label{sec:car}
 In this sub-section, we model the mileages ($Y$) of $392$ cars using the covariates ($X$):  displacement (D),   weight (W),  acceleration (A),   and horsepower (H);  see \url{http://lib.stat.cmu.edu/datasets/cars.data} for the data set. For our data analysis, we have scaled and centered each  of covariates to have mean $0$ and variance $1.$ To compare the prediction capabilities of the linear model to that of the single index model for this data set, we randomly split the data set into a training set of size $260$ and a test set of size $132$ and  compute the prediction error for both the linear model fit and the single index model fit. The average prediction error over $1000$ such random splits was $4.3$ for the linear model fit and $3.8$ for the single index model fit. The results indicate that the single index model is a better fit.

  In the left panel of  Figure~\ref{fig:Ozone_fit}, we have the scatter plot of $ \{(\hat{\theta}^\top x_i, y_i)\}_{i=1}^{392}$ overlaid with the plot of $\hat{m}(\hat{\theta}^\top x)$. In Table~\ref{tab:real_dat}, we display the estimates of $\theta_0$  based on the methods considered in the paper. The MAVE, the EFM estimator, and the PLSE give similar estimates while the EDR gives a different estimate of the index parameter.
\begin{table}[h]
\caption[Estimates of $\theta_0$ for the data sets in Sections~\ref{sec:car} and \ref{sec:Ozone}. ]{Estimates of $\theta_0$ for the data sets in Sections~\ref{sec:car} and \ref{sec:Ozone}. }\label{tab:real_dat}
\centering
\begin{tabular}{*{9}{c}}
 \toprule
\multirow{2}{*}{Method}  &\multicolumn{4}{c}{Car mileage data}& &\multicolumn{3}{c}{Ozone data} \bigstrut \\
\cmidrule(rl){2-5} \cmidrule(rl){7-9}
  & D & W  &A&H& &  R & W & T\\
 \midrule
GCV & 0.48 & 0.18 & 0.11 & 0.85  && 0.32 & -0.62 & 0.71 \\
  EFM & 0.44 & 0.18 & 0.13 & 0.87  &  & 0.29 & -0.60 & 0.75 \\
  EDR & 0.33 & 0.11 & 0.15 & 0.93  & & 0.22 & -0.64 & 0.73 \\
  rMAVE & 0.48 & 0.17 & 0.17 & 0.84  && 0.31 & -0.58& 0.75\\
  % MAVE2 & 0.43 & 0.13 & 0.13 & 0.88 && & & \\
   % -0.84 & 0.54\\
  % \cite{KongXia07}  &\NA&\NA&\NA&\NA&&  0.34&-0.62 &0.71\\
  \bottomrule
\end{tabular}
\end{table}

\subsection{Ozone concentration data}\label{sec:Ozone}
 For the second real data example, we study the relationship between Ozone concentration ($Y$) and three meteorological variables ($X$): radiation level ($R$), wind speed ($W$), and temperature ($T$). The data consists of $111$ days of complete measurements from May to September, $1973,$ in New York city. The data set can be found in the \texttt{EnvStats} package in R. \cite{YuRuppert02} fit a linear model, an additive model, and a fully nonparametric model and conclude that the single index model fits the data best. To fit a single index model to the data \cite{YuRuppert02} fix $10$ knots and fit cubic penalized splines to the data.
 % However for the estimator proposed in this paper the knots of the splines are chosen automatically by the data.
  The right panel of Figure~\ref{fig:Ozone_fit} shows the scatter plot of $ \hat{\theta}^\top X$ and $Y$ overlaid with the plot of $\hat{m}(\hat{\theta}^\top X)$. As in the previous example, we have scaled and centered each of the covariates such that they have mean $0$ and variance $1.$ We see that all the considered methods in the paper give similar estimates for $\theta_0$; see~Table~\ref{tab:real_dat}.

\begin{figure}[h!]
\vspace{0cm}
\centering
\includegraphics[width=\textwidth]{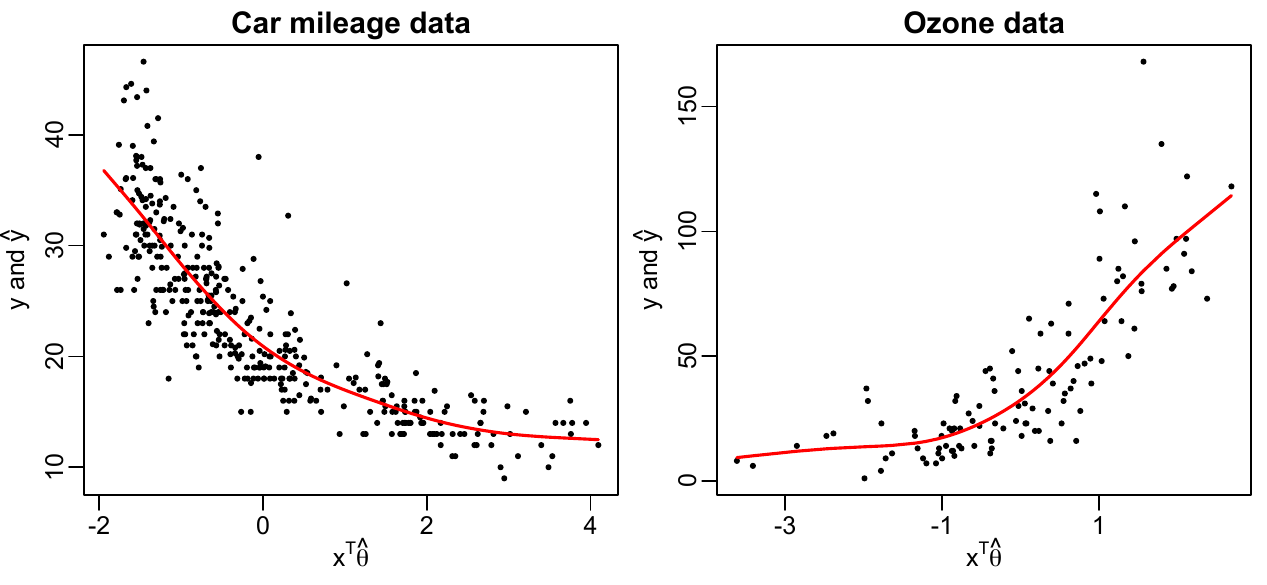}
\caption[Box plots of $L_1$ error of estimates of $\theta_0$ for Section~\ref{sec:dep_cov}.]{ Scatter plots of $\{(x_i^\top\hat{\theta}, y_i)\}_{i=1}^n$ overlaid with the plots of $\hat{m}$ (in solid red line)  for the two real data sets considered. Left panel: the car mileage data (Section~\ref{sec:car}); right panel: Ozone concentration data (Section~\ref{sec:Ozone}).}
\label{fig:Ozone_fit}
\vspace*{0cm}
\end{figure}

% \todo{What to do with this discussion?}
\section{Concluding remarks}\label{sec:summ}
In this paper we propose a simple penalized least squares based estimator $(\hat{m}, \hat\theta)$ for the unknown link function, $m_0,$ and the index parameter, $\theta_0,$ in the single index model under mild smoothness assumptions on $m_0.$ We prove that $\hat{m}$ is rate optimal (for the given smoothness) and $\hat{\theta}$ is $\sqrt{n}$-consistent and asymptotically normal. Moreover under homoscedastic errors, we show that $\hat\theta$ is efficient in the sense of \cite{BickelEtAl93}.  We have developed the R package~\texttt{simest} to compute the proposed estimators. We observe that the PLSE has superior finite sample performance compared to most competing methods.

 Several interesting future directions follow. Estimation and inference adapting to the smoothness of the link function is an interesting direction.~\cite{lepski2014adaptive} proposes an estimator for the single index model that adapts to the smoothness of the true link function, but the estimator depends on true (unknown) density of $X$ and requires independence between $\epsilon$ and $X$; also see \cite{lepski2013adaptive}. In the context of one dimensional smoothing splines~\citet[Chapter 21]{Gyorfi02} consider adaptation to smoothness using complexity regularization and extension of such a procedure to the case of single index model is an interesting direction of future research. In line with the recent literature on high-dimensional asymptotics in the single index model~\cite{li2017b,wang2014new,zhu2011variable}, it would be interesting to prove analogues of our results in a finite sample setting and under sparsity inducing regularization of the index parameter.~\cite{li2017b,peng2011penalized} consider variable selection in the single index model via a (additional) SCAD penalty (on the index parameter) on local linear and regression splines based estimation methods, respectively. They suggest that for the single index model, SCAD based variable selection methods have better performance when compared to LASSO based methods studied in~\cite{wang2014new,zhu2011variable}. Variable selection by incorporating a SCAD penalty on~\eqref{smooth:eq:L_n} is an exciting direction of research  and we plan to pursue this in the near future. 

 % {\clr Also include partially linear single index models.}

% The results provided in this paper were derived using the fact that model (1) is the truth, which possibly is a good approximation in the sense that the conditional expectation is close a function of projection of x in some direction. These results can be extended to the case where the target of inference $(m_0, \theta_0)$ can be defined as minimizer of $\E[(Y - m(\theta^{\top}X))^2]$ over $m\in \mathcal{S}$ and $\theta\in\Theta$. In this respect our results are more generally applicable where there are no assumptions on a true model and $\hat{m}\circ\hat{\theta}$ can be seen as estimating the best prediction under least squares loss based on the class $\{m\circ\theta:m\in \mathcal{S}, \theta\in\Theta\}$. Comparing with direct method estimator, this kind of interpretation is lost and the average derivative estimator (may be you should write something better) of the index parameter may not be useful as a good predictor or a good dimension reduction step.

% When $J(m_0)$ is zero it is possible to estimate $m_0$ at a parametric rate. However such an procedure would require a different choice of $\hat{\lambda}_n$ (than in \ref{smooth:a4}) that depends on knowledge whether $J(m_0)$ is 0 or not. However as $m_0$ is unknown the usability of such a model is drastically reduced. Thus, while  the proof of Theorem~\ref{smooth:thm:mainc} can easily generalized  to include the case of $J(m_0)=0$ (e.g., see \cite{VanDeGeer90}), we do not consider this in our work.

 \section*{Acknowledgments} We would like to thank Bodhisattva Sen for many helpful discussions and for his help in writing this paper. We would also like to thank Promit Ghosal and  David A. Hirshberg for helpful discussions. Finally, we would like thank the anonymous referee, associate editor, and  Editor for their work in refereeing the paper. Their suggestions led to an improved paper. 

% \appendix
% \appendixpage
% \addappheadtotoc
% %\section{Lemmas used in this paper}\label{smooth:app:used_lemmas}

%\begin{lemma}(Theorem 2.1 of \cite{UTR})
%\end{lemma}

%
\begin{appendix}

\section{Proofs of results in Section \ref{smooth:sec:Assymplse}} \label{smooth:app:sec:Assymplse}

We start with two useful lemmas concerning the properties of functions in $\mathcal{S}.$
\begin{lemma}\label{smooth:lem:bounds}(Lemma 3.6 of~\cite{VANC})
Let $m \in \{g\in \Ss: J(g) < \infty\}$. Then $|m'(s) - m'(s_0)|\le J(m)|s-s_0|^{1/2}$ for every $s,s_0\in D$.
\end{lemma}
% \end{proof}
\begin{lemma}\label{smooth:lem:dbounds}
% For any set $A\in\R^p$ $p\ge1$, let $\diameter(A)$ denote the diameter of the set $A.$
 Let $m \in \{g\in \Ss: J(g) < \infty \text{ and } \|g\|_\infty \le M\}$, where $M$ is a finite constant. Then $$\|m'\|_{\infty}\le 2M/\diameter(D)+ (1+ J(m)) \diameter(D)^{1/2},$$ where $\diameter(D)$ is the diameter of $D.$  Moreover if $\diameter(D) <\infty$, then $$\|m'\|_{\infty}\le C (1+J(m)),$$ where $C$ is a finite constant depending only on $M$ and $\diameter(D)$.
%{\color{red}[why is there a $J(m_0)? is it needed afterwards?$]}
\end{lemma}

\begin{proof}
Fix $s_0\in D$. Integrating the inequality
\[
-J(m)|t - s_0|^{1/2}\le m'(t) - m'(s_0) \le J(m)|t - s_0|^{1/2}
\]
with respect to $t$, we get
\[
|m(s) - m(s_0) - m'(s_0)(s - s_0)|\le J(m)\diameter(D)^{3/2},
\]
where $\diameter(D)$ is the diameter of $D.$  Since $\|m\|_{\infty}\le M$, we get that $$|m'(s_0)(s-s_0)| \le 2M+ J(m) \diameter(D)^{3/2}.$$  If we choose $s$ such that $|s-s_0| =\diameter(D)/2$, then we have
$$\|m'\|_{\infty} \le 2M/\diameter(D)+ (1+J(m)) \diameter(D)^{1/2}.$$ The rest of the lemma follows by choosing $C=2M/\diameter(D)  +\diameter(D)^{1/2}$.
\end{proof}

\subsection{Proof of Theorem \ref{smooth:thm:mainc}}\label{smooth:app:thm:mainc_proof}
Our proof of Theorem \ref{smooth:thm:mainc} is  along the lines of the proofs of Lemma 3.1 in \cite{MammenGeer97} and Theorem 10.2 in \cite{VANG}. Since $(\hat{m},\hat{\theta})$ minimizes $Q_n(m,\theta)+\hat{\lambda}_n^2 J^2(m)$, we have
\be\label{smooth:eq:proof_2}
Q_n(\hat{m},\hat{\theta})+  \hat{\lambda}_n^2 J^2(\hat{m})\leq   Q_n(m_0,\theta_0)  +\hat{\lambda}_n^2 J^2(m_0).
\ee
Observe that by definition of $Q_n(m, \theta)$, we have that \eqref{smooth:eq:proof_2} implies
\begin{align}
\|\hat{m} \circ \hat{\theta}-m_0\circ\theta_0\|_n^2 + \hat{\lambda}_n^2 J^2(\hat{m})
 % &\le \frac{2}{n} \sum_{i=1}^n (y_i-m_0(\theta_0^\top x_i)) (\hat{m}(\hat{\theta}^\top x_i)-m_0(\theta_0^\top x_i)) + \hat{\lambda}_n^2 J^2(m_0) \nonumber\\
&= \frac{2}{n} \sum_{i=1}^n \epsilon_i (\hat{m}(\hat{\theta}^\top x_i)-m_0(\theta_0^\top x_i)) + \hat{\lambda}_n^2 J^2(m_0) \nonumber
\end{align}
To find  the rate of convergence  of $\|\hat{m} \circ \hat{\theta}-m_0\circ\theta_0\|_n$ we will try to find upper bounds for $\sum_{i=1}^n \epsilon_i (\hat{m}(\hat{\theta}^\top x_i)-m_0(\theta_0^\top x_i))$ in terms of $\|\hat{m} \circ \hat{\theta}-m_0\circ\theta_0\|_n$ (modulus of continuity); see Section 1 of \cite{VanDeGeer90} for a similar proof technique.  To be able to find such a bound, we first study the behavior of $\hat{m}\circ\hat{\theta}$. Observe that by  Cauchy-Schwarz inequality we have
\begin{align} \label{smooth:eq:consistency_1}
\begin{split}
&Q_n(m_0,\theta_0)- Q_n(\hat{m},\hat{\theta}) \\
={}& \frac{2}{n} \sum_{i=1}^n \epsilon_i (\hat{m}(\hat{\theta}^\top x_i)-m_0(\theta_0^\top x_i))  - \frac{1}{n}\sum_{i=1}^n (\hat{m}(\hat{\theta}^\top x_i)-m_0(\theta_0^\top x_i))^2 \\
\leq{}& \left(\frac{4}{n} \sum_{i=1}^n \epsilon_i^2\right)^{1/2} \|\hat{m} \circ \hat{\theta}-m_0\circ\theta_0\|_n - \|\hat{m} \circ \hat{\theta}-m_0\circ\theta_0\|_n^2.\\
\end{split}
\end{align}
Note that by \ref{smooth:a3}, $(1/n) \sum_{i=1}^n \epsilon_i^2 =O(1)$ almost surely. On the other hand, since $(\hat{m},\hat{\theta})$ minimizes $Q_n(m,\theta)+\hat{\lambda}_n^2 J^2(m)$, we have
\be \label{smooth:eq:consistency_2}
 Q_n(m_0,\theta_0)-Q_n(\hat{m},\hat{\theta}) \geq \hat{\lambda}_n^2( J^2(\hat{m})-J^2(m_0))\geq -\hat{\lambda}_n^2 J^2(m_0) \geq o_p(1),
\ee
as $\hat{\lambda}_n=o_p(1).$ Combining \eqref{smooth:eq:consistency_1} and \eqref{smooth:eq:consistency_2}, we have
\begin{align*}
\|\hat{m} \circ \hat{\theta}-m_0\circ\theta_0\|_n^2 \leq \|\hat{m} \circ \hat{\theta}-m_0\circ\theta_0\|_n O_p(1) + o_p(1).
\end{align*}
Thus we have $\|\hat{m} \circ \hat{\theta}-m_0\circ\theta_0\|_n=O_p(1).$ We also have $\|\hat{m} \circ \hat{\theta}\|_n=O_p(1)$ as $\|m_0\circ \theta_0\|_\infty <\infty.$

 We will now use the Sobolev embedding theorem to get a bound on $\|\hat{m}\|_\infty$ in terms of $J(\hat{m})$.
\begin{lemma}\label{smooth:lem:SobolevEmbedding} (Sobolev embedding theorem, Page 85, \cite{Oden12}) Let $m: I \to \R$ ($I \subset \R$ is an interval) be a function such that $J(m) < \infty$. We can write $$m(t)=m_1(t)+m_2(t),$$ with $m_1(t)= \beta_1+\beta_2 t$ and $\|m_2\|_{\infty} \leq J(m) \diameter(I).$
\end{lemma}
Thus, by the above lemma, we can find functions $\hat{m}_1$ and $\hat{m}_2$ such that
$$\hat{m}(t)=\hat{m}_1(t)+ \hat{m}_2(t),$$ where $ \hat{m}_1=\hat{\beta}_1+ \hat{\beta}_2 t,$ and $\|\hat{m}_2\|_\infty \leq J(\hat{m})\diameter(D).$ Then
\begin{align}
\begin{split} \label{smooth:eq:bound_m1}
\frac{\|\hat{m}_1 \circ \hat{\theta}\|_n}{1+J(m_0)+J(\hat{m})} &\leq \frac{\|\hat{m}\circ \hat{\theta}\|_n}{1+J(m_0)+J(\hat{m})} + \frac{\|\hat{m}_2\circ \hat{\theta}\|_n}{1+J(m_0)+J(\hat{m})}  =O_p(1).
% &\le \frac{\|\hat{m}\circ \hat{\theta}\|_n}{1+J(m_0)+J(\hat{m})} + \frac{ \|\hat{m}_2\|_\infty }{1+J(m_0)+J(\hat{m})}
\end{split}
\end{align}
Let us  define
\bee
\mathbb{A}_n(\theta):=  \frac{1}{n}\sum_{i=1}^n \varphi_\theta(X_i) \varphi_\theta^\top(X_i) \qquad \text{and} \qquad A(\theta) :=\int \varphi_\theta(x) \varphi_\theta(x)^\top dP_X(x),
\eee
where  $\varphi_\theta(x):=(1,\theta^\top x)^\top.$ Furthermore, we  denote the smallest eigenvalues of $\mathbb{A}_n(\theta)$ and $A(\theta)$ by $\vartheta_n(\theta)$ and $\vartheta(\theta)$ respectively.  Since $\Theta$ is a bounded subset of $\R^d$, by the Glivenko-Cantelli Theorem, we have
\be
\sup_{\theta\in\Theta} |\vartheta_n(\theta) -\vartheta(\theta)| =o_p(1).
\ee
Let $\vartheta_0:= \min_{\theta \in \Theta} \vartheta(\theta).$ { By assumption \ref{smooth:a0} and and the fact that $|\theta|=1$, we have $\text{det}(A(\theta))=\theta^\top \Var(X) \theta$ and $\inf_{\theta\in \Theta} \text{det}(A(\theta)) >0.$ It follows that $\vartheta_0>0$  and}
\begin{align*}
\|\hat{m}_1 \circ \hat{\theta}\|_n^2&= (\hat{\beta}_1,\hat{\beta}_2) \mathbb{A}_n(\theta)(\hat{\beta}_1,\hat{\beta}_2)^\top\\
  &\geq \vartheta_n(\hat{\theta}) (\hat{\beta}_1^2+ \hat{\beta}_2^2)\\
&= \big[\vartheta_n(\hat{\theta})- \vartheta(\hat{\theta}) \big](\hat{\beta}_1^2+ \hat{\beta}_2^2) + \vartheta(\hat{\theta}) (\hat{\beta}_1^2+ \hat{\beta}_2^2)\\
&\geq o_p(\hat{\beta}_1^2+ \hat{\beta}_2^2) +  \vartheta_0 (\hat{\beta}_1^2+ \hat{\beta}_2^2)\\
&\geq o_p(\hat{\beta}_1^2+ \hat{\beta}_2^2) +  \vartheta_0 \max(\hat{\beta}_1,\hat{\beta}_2)^2 \\
\end{align*}
Thus by \eqref{smooth:eq:bound_m1} we have
\be\label{smooth:eq:m1_coeffbnd}  \frac{\max(\hat{\beta}_1,\hat{\beta}_2) }{1+J(m_0)+J(\hat{m})}=O_p(1).\ee
Moreover, since $D$ is a bounded set,  by \eqref{smooth:eq:m1_coeffbnd}  we have $\|\hat{m}_1\|_\infty/ (1+J(m_0)+J(\hat{m}))= O_p(1).$ Combining this with Lemma~\ref{smooth:lem:SobolevEmbedding}, we get
\begin{equation} \label{smooth:eq:LinfBound}
\frac{\|\hat{m}\|_\infty}{1+J(m_0)+J(\hat{m})} \le \frac{\|\hat{m}_1\|_\infty}{1+J(m_0)+J(\hat{m})}+\frac{\|\hat{m}_2\|_\infty}{1+J(m_0)+J(\hat{m})}= O_p(1).
\end{equation}
% $$A(\theta) :=\int \varphi_\theta(X) \varphi_\theta(X)^\top dP_X =\left( \begin{array}{cc}
%1 & \theta^\prime E(X) \\
%\theta^\prime E(X) & \theta^\prime E(X X^\prime) \theta
%\end{array} \right). $$
%Hence, we have that
%$$ \frac{1}{n} \sum_{i=1}^n \varphi_\theta(X_i) \varphi_\theta^\top(X_i) \rightarrow A(\theta) \quad \text{almost surely for every } \theta.$$
%Let $\vartheta(\theta)$ be the smallest eigenvalue value of $A(\theta)$. By assumption \ref{smooth:a7}, we have that $A(\theta)$ is nonsingular. We have that $\vartheta(\theta)> min_{\theta\in\Theta}\{ trace(A(\theta))/2 - \sqrt{trace(A(\theta))^2/4 -det(A(\theta)) }\}:= \vartheta_0.$ If $c_0$ be the minimum eigenvalue of $Var(X),$ then $det(A(\theta))= \theta^\prime Var(X) \theta > c_0>0.$ It can be easily seen that $ \vartheta_0>0.$
%$$\vartheta_0^2 \frac{(\hat{\beta}_1^2 + \hat{\beta}_2^2)}{J(m_0)+J(\hat{m})} =O_p(1),$$
%which implies that $\|\hat{\beta}\|_\infty/ (J(m_0)+J(\hat{m}))= O_p(1)$.
%%To prove Theorem \ref{smooth:thm:mainc}, we use the following lemma (proved in the Section) that gives us entropy bound on function class $\mathcal{B}_C.$
%%\begin{lemma}\label{smooth:thm:entrpoy_modified}
%%\be \label{smooth:eq:entropy_modified_class}
%%\log N \left(\delta,\mathcal{B}_C, \|\cdot\|_\infty\right) \lesssim \delta^{-1/2}.
%%\ee
%%\end{lemma}
%%

Now define the class of functions
\begin{align*}
%\mathcal{H}_{M_1,M_{2}} &:= \{m(\theta^{\top}x):\;\theta\in\Theta,\;m\in\Ss, \;\|m\|_{\infty} \le M_{1},\mbox{ and }J(m) \le M_{2}\},\\
\mathcal{B}_C&:=\left\lbrace \frac{m\circ\theta-m_0\circ\theta_0}{1+J(m_0)+J(m)}: m\in \Ss,\; \theta \in \Theta,\; \text{and } \frac{\|m\|_\infty}{1+J(m_0)+J(m)} \le C\right\rbrace.
\end{align*}Observe that by \eqref{smooth:eq:LinfBound}, we can find a $C_\varepsilon$ such that

\be \label{smooth:eq:bc_prob}
\p\left(\frac{\hat{m}\circ\hat\theta-m_0\circ\theta_0}{1+J(m_0)+J(\hat{m})} \in \mathcal{B}_{C_\varepsilon}\right) \ge 1-\varepsilon, \qquad \forall n.
\ee

The following lemma  in \cite{VANG} gives a upper bound for $\sum_{i=1}^n \epsilon_i g(x_i)$, in terms of entropy of the class of functions $g$.

\begin{lemma} \label{smooth:lem:ContEmpProcess} (Lemma 8.4, \cite{VANG})
Suppose $\mathcal{G}$ be a class of functions. If $\log N_{[\;]} (\delta, \mathcal{G}, \| \cdot \|_\infty) \le A \delta^{-\alpha},$ $\sup_{g\in\mathcal{G}} \|g\|_n \le R,$ and $\epsilon$ satisfies assumption \ref{smooth:a3}, for some constants $0<\alpha<2,$ $A,$ and $R.$ Then for some constant $c,$ we have for all $T\ge c,$
\begin{equation}
\p\left( \sup_{g\in \mathcal{G}} \frac{ |\frac{1}{\sqrt{n}} \sum_{i=1}^n \epsilon_i g(x_i)| }{\|g\|_n^{1-\frac{\alpha}{2}}} \ge T\right) \le c \exp\left[\frac{-T^2}{c^2}\right]
\end{equation}
\end{lemma}
Lemma \ref{smooth:lem:entropy1}, proved in Section \ref{smooth:app:lem:entrop1_proof} of the supplementary material, finds the bracketing number for the class of functions $\mathcal{B}_C$.
\begin{lemma}\label{smooth:lem:entropy1} For every fixed positive $M_1,M_2,$ and $C$, we have
\begin{align*}
%\log N_{[\;]}(\varepsilon,\h_{M_1,M_2},\|\cdot\|_{\infty})& \lesssim \left(\frac{M_1+M_2}{\varepsilon}\right)^{1/2},\\ % \label{smooth:eq:ent_H_m_1m_2}\\
\log N \left(\delta,\mathcal{B}_C, \|\cdot\|_\infty\right) &\lesssim \delta^{-1/2}. %\label{smooth:eq:entropy_modified_class}
\end{align*}
\end{lemma}

%{\clg
%This  lemma, proved in the Section,  uses the entropy calculated in the following lemma.
%\begin{lemma}\label{smooth:lem:entropy}(Theorem 2.4, \cite{VANG})
%Let $\mathcal{F}$ be a class of functions $f:I\to\mathbb{R}$ for $I$ a compact interval in $\mathbb{R}$ such that for some $M_1,M_2<\infty$, $\|f\|_{\infty}\le M_1,$ the first $k-1$ derivatives are absolutely continuous and $\Big[ \int_I [f^{(k)} (x)]^2 dx \Big]^{1/2} \le M_2$. Then there exists a constant $C$ depending on $I$ such that,
%\[
%\log N_{[\;]}(\varepsilon,\mathcal{F},\|\cdot\|_{\infty})\le C\left(\frac{M_1+M_2}{\varepsilon}\right)^{1/k},\quad \mbox{for all }\epsilon > 0.
%\]
%\end{lemma}}
In the view of \eqref{smooth:eq:bc_prob}, Lemmas \ref{smooth:lem:ContEmpProcess}  and \ref{smooth:lem:entropy1} allow us to conclude
\be \label{smooth:eq:main_cosistency}
\frac{(1/n) \sum_{i=1}^n \epsilon_i (\hat{m}(\hat{\theta}^\top x_i)-m_0(\theta_0^\top x_i)) } {\|\hat{m}\circ\hat{\theta}- m_0\circ\theta_0\|_n ^{3/4} (1+J(m_0)+J(\hat{m}))^{1/4}} =O_p(n^{-1/2}).
\ee
Together, \eqref{smooth:eq:consistency_2} and \eqref{smooth:eq:main_cosistency} imply
\begin{align} \label{smooth:eq:basicIneq}
\begin{split}
&\hat{\lambda}_n^2 (J^2(\hat{m})-J^2(m_0))\\
 \le{}& Q_n(m_0,\theta_0)-Q_n(\hat{m},\hat{\theta})\\
={}& \frac{2}{n} \sum_{i=1}^n (y_i-m_0(\theta_0^\top x_i)) (\hat{m}(\hat{\theta}^\top x_i)-m_0(\theta_0^\top x_i)) - \|\hat{m} \circ \hat{\theta}-m_0\circ\theta_0\|_n^2\\
\le{}& \|\hat{m}\circ\hat{\theta}- m_0\circ\theta_0\|_n ^{3/4} (1+J(m_0)+J(\hat{m}))^{1/4} O_p(n^{-1/2}) - \|\hat{m} \circ \hat{\theta}-m_0\circ\theta_0\|_n^2.
\end{split}
\end{align}
We will now consider two cases.

\noindent\textbf{Case 1}: Suppose $J(\hat{m}) > 1+J(m_0)$. By the proof of Theorem 10.2 of~\cite{MR1915446} with $\nu=2$ and $\alpha= 1/2$, we have that  
\begin{equation}\label{smooth:case1_J}
 J(\hat{m}) = O_p(n^{-1/2}) \hat{\lambda}_n^{-5/4}. \text{  and  } \|\hat{m} \circ \hat{\theta}-m_0\circ\theta_0\|_n= O_p(n^{-1/2}) \hat{\lambda}_n^{-1/4}.
\end{equation}
However, by assumption \ref{smooth:a3}, we have that $\hat{\lambda}_n^{-1}= O_p(n^{2/5}).$ Hence the conclusion follows.

\noindent\textbf{Case 2}: When $J(\hat{m}) \le 1+J(m_0),$ \eqref{smooth:eq:basicIneq} implies,
\[ \|\hat{m} \circ \hat{\theta}-m_0\circ\theta_0\|_n^2 \le \|\hat{m}\circ\hat{\theta}- m_0\circ\theta_0\|_n ^{3/4} (1+J(m_0))^{1/4} O_p(n^{-1/2} )+ \hat{\lambda}_n^2 J^2(m_0).\]
Therefore, it follows that either
\bee \label{smooth:eq:case2_part1}
\|\hat{m} \circ \hat{\theta}-m_0\circ\theta_0\|_n \le (1+J(m_0))^{1/5} O_p(n^{-2/5} )= O_p(\hat{\lambda}_n)
\eee
or
\bee \label{smooth:eq:case2_part2}
\|\hat{m} \circ \hat{\theta}-m_0\circ\theta_0\|_n \le O_p(1) \hat{\lambda}_n J(m_0) =O_p(\hat{\lambda}_n) J(m_0).
\eee
Thus we have that $J(\hat{m})=O_p(1),$ $\|\hat{m} \circ \hat{\theta}-m_0\circ\theta_0\|_n= O_p(\hat{\lambda}_n),$ and, by \eqref{smooth:eq:LinfBound}, $\|\hat{m}\|_\infty= O_p(1).$  To find the rates of convergence of $\|\hat{m} \circ \hat{\theta}-m_0\circ\theta_0\|,$ we use the following lemma.
\begin{lemma} \label{smooth:lem:EqualityNorms} (Lemma 5.16, \cite{VANG})
Suppose $\mathcal{G}$ is a class of uniformly bounded functions and for some $0 < \nu <2,$
$$ \sup_{\delta>0} \delta^\nu \log N_{[\;]} (\delta,\mathcal{G}, \|\cdot \|_\infty) < \infty .$$
Then for every given $\alpha >0$ there exists a constant  $C>0$ such that
$$ \limsup_{n\rightarrow \infty} \mathbb{P}\left( \sup_{g\in \mathcal{G}, \|g\|> C n^{-1/ (2+\nu)}} \left| \frac{\|g\|_n}{\|g\|} -1\right| >\alpha \right)=0.$$
\end{lemma}

\section{Proofs of results in Section \ref{smooth:sec:SemiInf} }\label{smooth:app:SemiInf}

\subsection{Proof of Theorem \ref{smooth:thm:eff_equation}}\label{smooth:app:proof_eff_equation}
We will first show that $\xi_t(u;\theta,\eta,m)$ is a valid submodel. Note  that $\phi_{\theta,\eta, 0}(u+(\theta-\theta)^\top h_{\theta}(u))=u,$ $\forall u \in D.$ Hence,
\begin{align*}
\xi_0(\theta^\top x;\theta,\eta,m) = m\circ \phi_{\theta,\eta,0}(\theta ^\top x)= m({\theta^\top x}).
\end{align*}
Now we will prove that $J^2(\xi_t(\cdot;\theta,\eta,m))  < \infty.$ Let us define $$\psi_{\theta,\eta,t}(u):= \phi_{\theta,\eta,t}(u+(\theta-\zeta_t(\theta,\eta))^\top h_{\theta}(u)),$$ then  $\xi_t(u;\theta,\eta,m) = m\circ\psi_{\theta,\eta,t}(u)$ Observe that
\begin{align}  \label{smooth:eq:J_diff}
\begin{split}
 J^2(\xi_t(\cdot; \theta,\eta,m) )={}& \int_D \big|\xi_t^{\prime\prime}(u;\theta,\eta,m)\big|^2du\\
={}&\int_D \big[m^{\prime\prime} \circ\psi_{\theta,\eta,t}(u) \psi_{\theta,\eta,t}^{\prime}(u) ^2+m^\prime \circ\psi_{\theta,\eta,t}(u) \psi_{\theta,\eta,t}^{\prime\prime}(u)\big]^2 du\\
 =& \int_D \big[m^{\prime\prime}(u) (\psi_{\theta,\eta,t}^{\prime}  \circ \psi_{\theta,\eta,t}^{-1}(u)) ^2+m^\prime (u) \psi_{\theta,\eta,t}^{\prime\prime} \circ \psi_{\theta,\eta,t}^{-1}(u)\big]^2 \frac{du}{\psi_{\theta,\eta,t}^{\prime} \circ \psi_{\theta,\eta,t}^{-1}(u) }
 \end{split}
\end{align}
where $\psi_{\theta,\eta,t}^{\prime}(u)= \frac{\partial}{\partial u} \psi_{\theta,\eta,t}(u)$.  Thus, we have that $J^2(\xi_t(\cdot; \theta,\eta,m) )=O(1)$ whenever $J(m) =O(1)$, $\|m\|_\infty =O(1)$, and $t$ in a small neighborhood of $0$ (as $\psi_{\theta,\eta,t}(\cdot)$ is a strictly increasing function when $t$ is small). Next we evaluate  $ \partial \xi_t(\zeta_t(\theta,\eta)^\top x;\theta,\eta,m)/\partial t $ to help with the calculation of the score function for the submodel $\{\zeta_t(\theta, \eta), \xi_t(\cdot; \theta,\eta,m)\}.$ Note that
\begin{align*}
& \frac{\partial}{\partial t}   \xi_t(\zeta_t(\theta,\eta)^\top x;\theta,\eta,m)\\
 ={}&  \frac{\partial}{\partial t} m\circ \phi_{\theta,\eta,t}\left(\zeta_t(\theta,\eta)^\top x+ (\theta- \zeta_t(\theta,\eta))^\top  h_{\theta}(\zeta_t(\theta,\eta)^\top x)\right)
 \\
 ={}& m^{\prime}\circ \phi_{\theta,\eta,t}\big(\zeta_t(\theta,\eta)^\top x+(\theta-\zeta_t(\theta,\eta))^\top    h_{\theta}(\zeta_t(\theta,\eta)^\top x)\big)\\
 &\quad \bigg[ \dot{ \phi}_{\theta,\eta,t}\big[\zeta_t(\theta,\eta)^\top x+[\theta- \zeta_t(\theta,\eta)]^\top   h_{\theta}(\zeta_t(\theta,\eta)^\top x)\big] \\ &\qquad +\phi^\prime_{\theta,\eta,t}\big[\zeta_t(\theta,\eta)^\top x+(\theta- \zeta_t(\theta,\eta))^\top    h_{\theta}(\zeta_t(\theta,\eta)^\top x)\big]  \frac{\partial \zeta_t(\theta,\eta)}{\partial t}^\top \Big[x \\
& \hspace{6cm}+(\theta- \zeta_t(\theta,\eta))^\top   h_{\theta}^\prime (\zeta_t(\theta,\eta)^\top x) x - h_{\theta}(\zeta_t(\theta,\eta)^\top x) \Big]\bigg],
\end{align*}
 where $\dot{ \phi}_{t,\theta} (u) =\partial\phi_{\theta,\eta,t} (u)/{\partial t}.$ We will now show that the score function of the submodel $\{ t,\xi_t(\cdot; ,\theta, \eta, m)\}$ is $\tilde{\ell}_{\theta,m}(y,x).$  Using the facts that $\phi_{\theta,\eta,t}^\prime(u)=1$ and $\dot{\phi}_{\theta,\eta,t}(u)=0$ for all $u \in D$ (follows from the definition \eqref{smooth:eq:BdryCond}) and ${\partial \zeta_t(\theta,\eta)}/{\partial t}= (-2 t /\sqrt{1-t^2|\eta|^2})\,\theta +H_\theta\eta$, we get
\begin{align}
\begin{split} \label{smooth:eq:eval_score}
\left.  \frac{\partial}{\partial t} (y-\xi_t(\zeta_t(\theta,\eta)^\top x;\theta,\eta,m))^2 \right|_{t=0}={}&-2(y-\xi_t(\zeta_t(\theta,\eta)^\top x;\theta,\eta,m)) \left.\frac{\partial   \xi_t(\zeta_t(\theta,\eta)^\top x;\theta,\eta,m)}{\partial t}\right|_{t=0}\\
={}&-2(y-m(\theta^\top x)) m^{\prime}(\theta ^\top x)\eta^\top H_\theta^\top (x -h_{\theta}(\theta ^\top x))
\end{split}
\end{align}

%{ \color{red}To help with the calculations note that Thus $t=\theta$, we have
%\begin{align*}
%\left. \frac{\partial \xi_t(t^\top x)}{\partial t} \right\vert_{t=\theta} ={}& g^{\prime}(\theta ^\top x) x_2 + \int_{s_0}^{\theta^\top x} g^{\prime\prime} (u) [- k(u)] du\\
%&\qquad +(g^\prime (s_0) -m_0^\prime (s_0) )k(s_0) + m_0^\prime (s_0) h_{\theta_0}(s_0).
%\end{align*}}
Observe that  $(\hat{m},\hat{\theta})$ minimizes the penalized loss function  in \eqref{smooth:simpls} and $\xi_0( \zeta_0(\hat{\theta},\eta)^\top x;\hat{\theta},\eta,\hat m)=\hat{m}(\hat{\theta}^\top x),$ where  $\zeta_t(\hat{\theta},\eta)=\sqrt{1-t^2|\eta|^2}\, \hat\theta + s H_{\hat{\theta}} \eta$.  Hence, for every  $\eta \in \R^{d-1}$, the function
\begin{align}\label{smooth:eq:path_func}
 t \erelbar{21} \frac{1}{n}\sum_{i=1}^n \big(y_i-\xi_t(\zeta_t(\hat\theta,\eta)^\top x;\hat{\theta},\eta,\hat{m})\big)^2+ \hat{\lambda}_n^2\int_D \Big| \frac{\partial^2}{\partial u^2}\xi_t(u;\hat{\theta},\eta,\hat{m})\Big|^2du
 \end{align}
on a  some small neighborhood of $0$ (that depends on $\eta$) is minimized at $t=0$.
%\begin{align} \label{smooth:eq:step_1_1}
%\begin{split}
%\p_n S_{\hat{\theta}, \hat{m}} -\frac{\hat{\lambda}_n^2}{2} \left.\frac{\partial J^2(\xi_t(\hat{m},\hat{\theta}))}{\partial t} \right\vert_{t =\hat{\theta}}=0.
% \end{split}
%\end{align}
Moreover, using some tedious algebra it can be shown that $ J^2(\xi_t(\cdot; \theta,\eta,m) )$ is differentiable and
\begin{align*}%\label{smooth:eq:J_bound_on_diff}
\begin{split}
\left.\frac{\partial }{\partial t}J^2(\xi_t(\cdot;\theta,\eta,m)) \right\vert_{t =0} \lesssim \int_D |m''(p)|^2dp.
\end{split}
\end{align*}
This we have that the function in \eqref{smooth:eq:path_func} is differentiable at $t=0.$  Conclude that, for all $\eta \in \R^{d-1}$ we have  \[
\eta^\top \p_n \tilde{\ell}_{\hat{\theta}, \hat{m}} -\hat{\lambda}_n^2 \left.\frac{\partial J^2(\xi_t(\cdot;{\theta},\eta, m))}{\partial t} \right\vert_{t =\hat{\theta}}=0.
 \]
The proof of the theorem is now complete as $\hat{\lambda}_n^2= o_p(n^{-1/2}).$
\end{appendix}
\noeqref{smooth:eq:L_n}
\noeqref{smooth:eq:eta_hat}
\noeqref{smooth:eq:g_beta0_U_envelope}
\noeqref{smooth:eq:second_part_main}
\noeqref{smooth:eq:insetoutset}
\noeqref{smooth:eq:nonp_path}

\bibliographystyle{apalike}
\bibliography{SigNoise}
\newpage
\begin{center}
\Large\textbf{Supplementary Material}
\end{center}
\setcounter{section}{0}
\setcounter{thm}{7}
\setcounter{lemma}{9}
\renewcommand{\thesection}{S.\arabic{section}}
\section{Proof of Theorem \ref{smooth:thm:existance}}\label{smooth:app:thm:existance_proof}
The minimization problem considered is
\[
\inf_{\theta\in\Theta,m\in\mathcal{S}}\mathcal{L}_n(m,\theta;\lambda),
\]
where $\mathcal{L}_n$ is defined in \eqref{smooth:eq:L_n}.  For any fixed vector $\theta \in \Theta$, define $t_i^{\theta}:=\theta^\top x_i$, for $i=1,\ldots,n.$
Then we have
\[
\mathcal{L}_n(m,\theta; \lambda)= \left[\frac{1}{n}\sum_{i=1}^n \big(y_i-m(t_i^{\theta})\big)^2 + \lambda^2\int_D \big| m''(t)\big|^2dt\right]
\]
and the minimization can be equivalently written as
$\inf_{\theta\in\Theta}\inf_{m\in\mathcal{S}} \mathcal{L}_n(m,\theta; \lambda).$
%\[
%\inf_{\theta\in\Theta}\inf_{m\in\mathcal{S}} \left[\frac{1}{n}\sum_{i=1}^n \big(y_i-m(t_i^{\theta})\big)^2 + \lambda^2\int_D \big| m''(t)\big|^2dt\right].
%\]
Let us define
\begin{equation}\label{smooth:eq:T}
T(\theta) := \inf_{m\in\mathcal{S}}\mathcal{L}_n(m,\theta; \lambda) \quad  \text{ and } \quad  m_{\theta}:= \argmin_{m\in\mathcal{S}}\mathcal{L}_n(m,\theta; \lambda).
\end{equation}
Theorem 2.4 of \cite{greensilverman94} proves that  the  infimum in  \eqref{smooth:eq:T} is attained for every $\theta \in \Theta$ and   the unique minimizer  $m_\theta$ is a natural cubic spline with knots at $\{t_i^{\theta}\}_{i=1}^n.$  Furthermore \cite{greensilverman94} notes that (see Section 2.3.4), $m_\theta$ does not depend on $D$ beyond the condition that $\{t_i^\theta\}_{1\le i\le n} \in D$. Moreover, $m''_\theta$ is zero outside $(t^\theta_{(1)},t^\theta_{(n)})$, where for  $k=1,\ldots,n$,  $t^\theta_{(k)}$ denotes the $k$-th smallest value in $\{t_i^\theta\}_{i=1}^n$.

For every $\theta \in  \Theta$, $m_\theta$ is determined by points in a bounded set, namely  $D_R:= [- t_{\max}, t_{\max}],$ where $t_{\max}$ a finite constant such that $\sup_{\theta\in \Theta} \max_{i\le n} |\theta^\top x_i|  \le t_{\max}$. Note that such a constant always exists as $\Theta\subset S^{d-1}$.  Define \[\Ss_R:= \{ m: D_R \rightarrow \R| \, m' \text{ is absolutely continuous}\},\] and for all $m \in \Ss_R$, define $J^2_R(m) :=\int_{D_R} |m''(t)|^2dt.$ For every $m\in \Ss_R$ and $\theta\in \Theta$, we define
\be\label{smooth:eq:L_R}
\mathcal{L}^R_n(m,\theta; \lambda)= \left[\frac{1}{n}\sum_{i=1}^n \big(y_i-m(t_i^{\theta})\big)^2 + \lambda^2\int_{D_R} \big| m''(t)\big|^2dt\right],
\ee
\be \label{smooth:eq:T_R}
T_R(\theta) := \inf_{m\in\mathcal{S}_R} \mathcal{L}^R_n(m,\theta; \lambda), \quad \text{ and }\quad  m_\theta^R:=\argmin_{m\in\Ss_R} \mathcal{L}^R_n(m,\theta; \lambda).
\ee
\cite{greensilverman94} observes that (see Section 2.3.4), $m_\theta$ is the  linear extrapolation of $m_\theta^R$ to $D$. Moreover, as $m_\theta$ is a linear function outside $D_R$, we have
%\[\int_{D_R} \left\vert \frac{\partial^2{m_\theta^R}}{\partial t^2} (t)\right\vert^2dt=\int_{D} \big| m''_\theta (t)\big|^2dt \quad \text{and} \quad T_R(\theta)=T(\theta).\]
\[\int_{D_R} \left\vert (m_\theta^R)''(t)\right\vert^2dt=\int_{D} \big| m''_\theta (t)\big|^2dt \quad \text{and} \quad T_R(\theta)=T(\theta).\]
Thus we have
\[
\inf_{\theta\in\Theta,m\in\mathcal{S}}\mathcal{L}_n(m,\theta;\lambda)=\inf_{\theta\in\Theta} T(\theta)=\inf_{\theta\in\Theta} T_R(\theta)=\inf_{\theta\in\Theta,m\in\mathcal{S}}\mathcal{L}^R_n(m,\theta;\lambda).\]

As $\Theta$ is a compact set, the existence of the minimizer of $\theta\mapsto T_R(\theta)$ will be established if we can show that $T_R(\theta)$ is a continuous function on $\Theta$; see the Weierstrass extreme value theorem. We now prove that $\theta \mapsto T_R(\theta)$ is a continuous function. Notice that $\sup_{\theta\in\Theta}T_R(\theta) \le \sup_{\theta\in\Theta}\mathcal{L}^R_n(0,\theta;\lambda)=\sum_{i=1}^n y_i^2/n <\infty.$
%$ because, if not then for some $\theta'\in\Theta$, $T(\theta) = \infty$ ($\Theta$ is compact) which implies that for all $m\in\mathcal{S}$, $Q_n(m,\theta') + \lambda^2\int_D \big| m''(t)\big|^2dt= \infty$ and so $Q_n(m,\theta') = \infty$ which is impossible as $Q_n$ is a continuous function of $\theta$ and $\theta$ lies in a compact set.
Hence there is a finite constant $K$ (depending only on $\{y_i\}_{i=1}^n$) such that for all $\theta\in\Theta$,
\begin{align} \label{smooth:eq:th1_2}
Q_n(m^R_\theta,\theta) +\lambda^2 J_R^2(m^R_\theta)&\le K.
\end{align}
We will use the above bound to show that there exists a finite  $L$ (depending only on $\lambda$ and $\{(y_i, x_i)\}_{i=1}^n$)  such that $\|m^R_\theta\|_\infty \le L$ and $J_R(m^R_\theta)\le L$ for all $\theta \in \Theta.$
By \eqref{smooth:eq:th1_2}, we have that
\be \label{smooth:eq:th1_bnd} J_R^2(m^R_\theta) \le K/\lambda^2 \quad \text{ and } \quad |m^R_\theta(t^\theta_{(i)})|\le \sqrt{n K} +\max_{i\le n} |y_i|,
\ee
 for $i=1,\ldots, n.$
If $t^\theta_{(1)}=t^\theta_{(n)},$ then it is easy to see that $m^R_\theta(\cdot)\equiv \sum_{i=1}^n y_i/n$  which implies that  $\|m^R_\theta\|_\infty$ is bounded and $J_R(m^R_\theta)=0$. Now let us assume $t^\theta_{(1)}<t^\theta_{(n)}$.
 By Lemma \ref{smooth:lem:dbounds}, for any  $s\in \R$ such that $|s| \le t_{\max}$, we have
\[
\left\vert (m^R_\theta)' (s)- (m^R_\theta)'(t^\theta_{(1)})\right\vert \le J_R(m_\theta^R) \sqrt{t_{\max}}.
\]
Integrating  the above display with respect to $s$, we get
\be \label{smooth:eq:mb}
\left\vert m^R_\theta(s) - m^R_\theta(t^\theta_{(1)}) - ( m^R_\theta)'(t^\theta_{(1)})(s - t^\theta_{(1)})\right\vert\le  J_R(m^R_\theta)(t_{\max})^{3/2}.
\ee
Taking $s=t^\theta_{(n)}$ in the previous display, we have  $|( m^R_\theta)'(t^\theta_{(1)}) | \le C$,  where the constant $C$ depends only on $ K, \lambda,$ and  $\{(x_i, y_i)\}_{i=1}^n$ (see \eqref{smooth:eq:th1_bnd}). In view of  the bound on $|( m^R_\theta)'(t^\theta_{(1)}) |$, \eqref{smooth:eq:mb} implies that
$$ \sup_{|s|\le t_{\max}} |m^R_\theta(s)| \le C_1,$$
 where the constant $C_1$ depends only on $ K, \lambda,$ and $\{(y_i, x_i)\}_{i=1}^n$. Thus, there exists  a finite $L$ (depending only on $\lambda $ and $ \{(y_i, x_i)\}_{i=1}^n$)  such that $\|m^R_\theta\|_\infty \le L$ and $J_R(m^R_\theta)\le L$. Note that $L$ does not depend on $\theta$. As $\|m^R_\theta\|_\infty \le L$ and $J_R(m^R_\theta)\le L$, we can redefine $T_R(\theta)$ as
\[
T_R(\theta) = \inf_{m\in \{ m \in \mathcal{S}_R: \|m\|_\infty \le L \text{ and }J_R(m)\le L \}} \left[Q_n(m,\theta) + \lambda^2\int_{D_R} \big| m''(t)\big|^2dt\right].
\]
We will now show that  the class of functions \[\{ Q_n(m,\cdot) : \Theta \rightarrow \R|\, m\in\Ss_R, \|m\|_\infty \le L, \text{ and } J_R(m)\le L\}\] is uniformly equicontinuous, i.e., for every $\varepsilon>0$, there exists a $\delta>0$ such that  $|\theta-\eta| \le \delta$ implies that $$ \sup_{m \in \{m \in \Ss_R: \|m\|_\infty \le L \text{ and }J_R(m)\le L\} }|Q_n(m,\theta)-Q_n(m,\eta)|\le \varepsilon.$$Note that
\begin{align} \label{smooth:eq:th1_3}
\begin{split}
&|Q_n(m,\theta) -Q_n(m,\eta)|\\
={}& \frac{1}{n}\left\vert \sum_{i=1}^n \big[ (y_i -m(\theta^\top x_i))^2-(y_i -m(\eta^\top x_i))^2\big]\right\vert\\
={}&\frac{1}{n}\left\vert \sum_{i=1}^n \big[ (m(\eta^\top x_i) -m(\theta^\top x_i))^2 +  2 (y_i -m(\eta^\top x_i))(m(\eta^\top x_i) -m(\theta^\top x_i))\big]\right\vert\\
\le {}& \max_{1\le i\le n} |m(\eta^\top x_i) -m(\theta^\top x_i)|^2 +   \frac{2}{n} \max_{1\le i\le n} |m(\eta^\top x_i) -m(\theta^\top x_i)| \sum_{i=1}^n |y_i -m(\eta^\top x_i)|.
\end{split}
\end{align}
 In view of Lemma \ref{smooth:lem:dbounds},  for $i=1, \ldots, n$ we have
\begin{equation} \label{smooth:eq:th1_1}
|m(\theta^{\top}x_i) - m(\eta^{\top}x_i)|\le\|m'\|_{\infty}|x_i^{\top}(\theta-\eta)|\le C_2 (1+J_R(m))|\theta - \eta|,
\end{equation} where $C_2$ is a constant that depends only on $L$ and $\max_{1\le i \le n} |x_i|$. For every $m \in \{m \in \Ss_R: \|m\|_\infty \le L \text{ and }J_R(m)\le L\}$, \eqref{smooth:eq:th1_3} and \eqref{smooth:eq:th1_1}  imply that \be \label{smooth:eq:th1_10}
\sup_{m \in \{m \in \Ss_R: \|m\|_\infty \le L \text{ and }J_R(m)\le L\} }|Q_n(m,\theta) -Q_n(m,\eta)| \le C_3 |\theta-\eta|,
\ee
 where the constant $C_3$ depends only on $L$ and $\max_{1\le i \le n} |x_i|$.  Observe that  for every $\theta \in \Theta$, $m_\theta^R \in \{m \in \Ss_R: \|m\|_\infty \le L \text{ and }J_R(m)\le L\}$. Fix  $\delta=\varepsilon/C_3$, then uniform equicontinuity of $\{ \theta\mapsto Q_n(m,\theta): m\in\Ss_R, \|m\|_\infty \le L, \text{ and } J_R(m)\le L\}$ implies that,  for all $|\eta-\theta|\le \delta,$ we have
\begin{equation}\label{smooth:eq:th1_be}
Q_n(m^R_\eta, \theta) -\varepsilon \le  Q_n(m^R_\eta,\eta) \quad \text{ and }\quad Q_n(m^R_\theta, \eta)  \le Q_n(m^R_\theta,\theta) +\varepsilon.
\end{equation}
 Recall that for every $\beta \in \Theta$ and  $m \in \{ m\in \Ss_R: J_R(m)<\infty\}$, we have $\mathcal{L}^R_n(m^R_\beta,\beta;\lambda) \le \mathcal{L}^R_n(m,\beta;\lambda).$ Thus, from \eqref{smooth:eq:th1_be}, we have
\begin{align}\label{smooth:eq:th1_4}
\begin{split}
Q_n(m^R_\eta, \theta) -\varepsilon \le Q_n(m^R_\eta,\eta)\;\Leftrightarrow \; \mathcal{L}^R_n(m^R_{\eta},\theta;\lambda) -\varepsilon\le\mathcal{L}^R_n(m^R_{\eta},\eta;\lambda) \\\Rightarrow \;\mathcal{L}^R_n(m^R_{\theta},\theta;\lambda) -\varepsilon\le\mathcal{L}^R_n(m^R_{\eta},\eta;\lambda)\; \Rightarrow \;T_R(\theta)-\varepsilon\le T_R(\eta)
\end{split}
\end{align}
and
\begin{align}
\begin{split}\label{smooth:eq:th1_5}
Q_n(m^R_\theta, \eta)  \le Q_n(m^R_\theta,\theta) +\varepsilon\;\Leftrightarrow\;
\mathcal{L}^R_n(m^R_{\theta},\eta;\lambda)\le\mathcal{L}^R_n(m^R_{\theta},\theta;\lambda) + \varepsilon\\\Rightarrow \;\mathcal{L}^R_n(m^R_{\eta},\eta;\lambda)\le\mathcal{L}^R_n(m^R_{\theta},\theta;\lambda)+\varepsilon \;\Rightarrow \;T_R(\eta)\le T_R(\theta)+\varepsilon.
\end{split}
\end{align}

  Combining \eqref{smooth:eq:th1_4} and \eqref{smooth:eq:th1_5}, we have that $T_R(\theta)-\varepsilon \le T_R(\eta)\le T_R(\theta)+\varepsilon, $ for all $|\eta-\theta|\le \delta.$ Thus, it follows that  $\theta\mapsto T_R(\theta)$ is uniformly continuous and $T_R(\theta)$ attains a minimum  on the compact set $\Theta$ ($S^{d-1}$ is compact and $\Theta$ is closed subset of $S^{d-1}$).  Thus  $$\hat{\theta}=\argmin_{\theta \in \Theta} T_R(\theta)=\argmin_{\theta \in \Theta} T(\theta)$$ is well defined. Moreover by Theorem 2.4 of \cite{greensilverman94} we have that $ m^R_{\hat{\theta}}$ is a unique natural cubic spline with knots at $\{t_i^{\hat{\theta}}\}_{i=1}^n$ and
  $$\hat{m} =m_{\hat{\theta}},$$ where $m_{\hat{\theta}}$  is the linear extrapolation of $m_{\hat\theta}^R$ to $D$.

\section{Proofs of results in Section~\ref{smooth:sec:Assymplse}}
\subsection{Proof of Lemma  \ref{smooth:lem:entropy1}} \label{smooth:app:lem:entrop1_proof}

To prove this lemma, we use the following entropy bound from \cite{VANG}.
We will also use the following result in the proofs of Lemmas \ref{smooth:lem:first_term} and \ref{smooth:lem:second_term} in Sections \ref{smooth:app:proof_first_term} and \ref{smooth:app:proof_second_term}, respectively.
\begin{lemma}\label{smooth:lem:entropy}(Theorem 2.4, \cite{VANG})
Let $\mathcal{F}$ be a class of functions $f:I\to\mathbb{R}$ (for $I$ a compact interval in $\mathbb{R}$) such that for some $M_1,M_2<\infty$, $\|f\|_{\infty}\le M_1,$ the first $k-1$ derivatives are absolutely continuous and $\int_I [f^{(k)} (x)]^2 dx  \le M_2^2$. Then there exists a constant $C$ depending only on $I$ such that,
\[
\log N_{[\;]}(\varepsilon,\mathcal{F},\|\cdot\|_{\infty})\le C\left(\frac{M_1+M_2}{\varepsilon}\right)^{1/k},\quad \mbox{for all }\varepsilon > 0.
\]
\end{lemma}
The above lemma says that  the class of functions
$$\G_{M_1,M_2} := \{m \in \mathcal{S}: \; \|m\|_{\infty}\le M_1, \mbox{ and }J(m)\le M_2\}$$ can be covered by $\exp(C \sqrt{M_1+M_2} \delta^{-1/2})$  balls with radius $\delta$ in the sup-norm, i.e.,
\begin{equation}\label{smooth:eq:ent_geer}
\log N_{[\;]}(\delta,\G_{M_1,M_2},\|\cdot\|_{\infty})\le C\left(\frac{M_1+M_2}{\delta}\right)^{1/2}.
\end{equation}%
For all $\theta_1,\theta_2\in\Theta$, we have that $|\theta_1 - \theta_2| \le 2$. Thus by Lemma 4.1 of \cite{Pollard90}, we have
\[
N(\varepsilon,\Theta,|\cdot|)\lesssim \varepsilon^{-d+1}.
\]
Now define the class of functions
$$\mathcal{H}_{M_1,M_{2}} := \{m(\theta^{\top}x):\;\theta\in\Theta,\;m\in\mathcal{S}, \;\|m\|_{\infty} \le M_{1},\mbox{ and }J(m) \le M_{2}\}.$$
We will show that
\be \label{smooth:eq:ent_H_m_1m_2}
\log N_{[\;]}(\varepsilon,\h_{M_1,M_2},\|\cdot\|_{\infty})\lesssim \left(\frac{M_1+M_2}{\varepsilon}\right)^{1/2}.
\ee
Note that, with respect to $\|\cdot\|_{\infty}$-norm covering number and bracketing number are the same and we can choose an $\varepsilon$-net from within the function class. Thus $\|\cdot\|_{\infty}$ brackets can be chosen from the function class.

Consider an $\varepsilon/[2(1+M_2)T]$-net of $\Theta$, $\{\theta_1,\theta_2,\ldots,\theta_p\}$, $\rchi \subset B_\textbf{0}(T) \subset \R^d,$ the Euclidean ball of radius $T$ around the origin. Choose an $\varepsilon/2$-net for $\G_{M_1,M_2}$, $\{m_1,m_2,\ldots,m_q\}$. We can, without loss of generality, assume that $m_i\in\G_{M_1,M_2}.$ Thus by Lemma \ref{smooth:lem:dbounds}, we have $\|m_i'\|_{\infty}\lesssim 1+M_2$.

Now we will show that the set of functions $\{m_i\circ\theta_j\}_{1 \le i \le q, 1\le j \le p}$ form an $\varepsilon$-net for $\h_{M_1,M_2}$ with respect to $\|\cdot\|_{\infty}$-norm. For any given $m\circ\theta\in\h_{M_1,M_2}$, we can get $m_i$ and $\theta_j$ such that $\|m-m_i\|_{\infty} < \varepsilon/2$ and $|\theta - \theta_j| < \varepsilon/2(1+M_2)T.$ Then
\begin{align}
|m(\theta^{\top}x) &- m_i(\theta_j^{\top}x)|\\&\le|m(\theta^{\top}x) - m(\theta_j^{\top}x)| + |m(\theta_j^{\top}x) - m_i(\theta_j^{\top}x)|\\
&\le \|m'\|_{\infty}|x\|\theta - \theta_j| + \|m - m_i\|_{\infty} \le \frac{(1+M_2)|x|\varepsilon}{2T(1+M_2)}+\frac{\varepsilon}{2} \le \varepsilon.
\end{align}
Hence, the bracketing entropy number in the $\|\cdot\|_{\infty}$-norm for the required set is bounded above by a multiple of $(M/\varepsilon)^{1/2} + \log(C2T(1+M_2)\varepsilon^{-d+1})$ for a suitable constant $C >0,$ which is further bounded by a multiple of $\left(M/{\varepsilon}\right)^{1/2}$, where $M = M_1+M_2$. Thus we have \eqref{smooth:eq:ent_H_m_1m_2}.

Now we will use \eqref{smooth:eq:ent_H_m_1m_2} to prove Lemma \ref{smooth:lem:entropy1}. Let us define,
$$\f_C:= \left\lbrace f(\theta^\top x) : f= \frac{m}{1+J(m_0)+J(m)}, \;\theta \in \Theta,\; m \in \Ss,\text{ and } \frac{\| m\|_\infty}{1+J(m_0) +J(m)} \le C \right\rbrace$$

Since $\f_C \subset \h_{C,1},$ we can choose $\delta/2$ brackets $[g_{1,1},g_{1,2}],\ldots,[g_{q,1},g_{q,2}]$ over $\f_C$ such that for every $f(\theta^\top x) \in \f_C$ there exists a $i$ such that $g_{i,1}(x)\le f(\theta^\top x) \le g_{i,2}(x)$. Let us now define,
$$\f^*:= \left\lbrace h: h= \frac{m_0}{1+J(m_0)+J(m)} \text{ and } m \in \Ss \right\rbrace.$$ Observe that $\f^* \subset \G_{C_1,1},$ where $C_1= \|m_0\|_\infty/ J(m_0).$ Thus we can choose $\delta/2$ brackets $[l_{1,1},l_{1,2}],\ldots,[l_{r,1},l_{r,2}]$ over $\f^*$ such that for every $h \in \f^*$ there exists a $j$ such that $l_{j,1}(\theta_0^\top x)\le h(\theta_0^\top x) \le l_{j,2}(\theta_0^\top x).$ Thus we have,
$$ g_{i,1}(x)-l_{j,2}(\theta_0^\top x)\le \frac{m(\theta^\top x)}{1+J(m_0)+J(m)} -\frac{m_0(\theta_0^\top x)}{1+J(m_0)+J(m)} \le g_{i,2}(x)-l_{j,1}(\theta_0^\top x),$$ where $i$ depends on $(m,\theta)$ and $j$ on $m$.

Brackets of the form $[ g_{i,1}(x)-l_{j,2}(\theta_0^\top x), g_{i,2}(x)-l_{j,1}(\theta_0^\top x)]$ for $i \in \{ 1,\ldots q\}$ and $j \in \{1,\ldots r\}$ cover the required space. Hence, the bracketing entropy satisfies

$$\log N \left(\delta, \mathcal{B}_C,\;\|\cdot\|_\infty\right) \le \frac{(C+1)^\frac{1}{2} + (C_1+1)^\frac{1}{2} }{\delta^\frac{1}{2}},$$where $C_1= \|m_0\|_\infty/ J(m_0).$

%
%We will crucially use Lemma 5.16 and 8.4 of \cite{VANG} and  in the following theorem. We state the Lemmas here for clarity.
\subsection{Proof  of Theorem \ref{smooth:thm:cons}} \label{smooth:sec:proof_cons}
\noindent The following lemma is crucial to the  proof of Theorem \ref{smooth:thm:cons}.
\begin{lemma}\label{smooth:pcomp}
For every fixed $M$, the set of functions $m\in\Ss$ with $J(m) \le M$ and $\|m\|_{\infty} \le M$ is precompact relative to $\|\cdot\|^S_D$.
\end{lemma}
\begin{proof}
Let us define, $\mathcal{D}_M:= \{m \in \Ss :  \|m\|_{\infty} \le M, \text{and}\, J(m) \le M \}$. By Lemma \ref{smooth:lem:bounds} the class of functions $\{ m' : m\in \mathcal{D}_M \}$ is uniformly Lipschitz of order $1/2$. Thus any sequence of functions $\{ m'_{k} : m_k \in \mathcal{D}_M\} $ is equicontinuous. By Lemma \ref{smooth:lem:dbounds}, $\{m'_k\}$  is uniformly bounded. Applying the Arzela-Ascoli theorem, we see that every sequence $\{m_{k}\}$  has a subsequence $\{m_{k_l}\}$ such that $\{m^\prime_{k_l}\}$ converges uniformly on $D$. Since  $\{m^\prime_{k_l}\}$ is uniformly bounded, we have that $\{m_{k_l}\}$ is equicontinuous. Therefore as $\|m_{k_l}\|_\infty \le M$, by Arzela-Ascoli theorem, there exists  a subsequence $\{k_{l_j}\}$ of $\{k_l\}$  such that $\{m_{k_{l_j}}\}$ converge uniformly on $D$. Since these functions converge uniformly on a compact set, by applying the dominated convergence theorem, we see that there exists a subsequence such that functions and derivatives converge. Furthermore, the derivative of the limit equals the limit of the derivative.
\end{proof}
Suppose that $\|m_k\circ\theta_k - m_0\circ\theta_0\|\to 0$, $\|m_k\|_\infty=O(1)$, and $J(m_k) = O(1)$. By Lemma \ref{smooth:pcomp}, every subsequence of $(m_k,\theta_k)$ has a further subsequence $(m_{k_l}, \theta_{k_l})$ such that $\theta_{k_l}\to\theta$ and $\|m_{k_l} - m\|^S_D\to 0$ for some $\theta$ and $m$. Then $\|m_{k_l}\circ\theta_{k_l} - m\circ\theta\| \to 0$ by continuity of the map $(m,\theta)\mapsto m\circ\theta.$ Thus $\|m\circ\theta - m_0\circ\theta_0\| = 0,$ and hence by assumption \ref{smooth:a0}, we get $\theta = \theta_0$ and $m = m_0$ on the support $D_0$. The assumption that $D_0$ is the closure of its interior implies that $m'$ and $m_0'$ agree on $D_0$. Since the convergence in Lemma \ref{smooth:pcomp} is uniform, we get that $\|m - m_0\|_{D_0}= 0$. Combining this with Theorem \ref{smooth:thm:mainc}, we get that $\hat{\theta}\overset{P}{\to}\theta_0$ and $\|\hat{m} - m_0\|^S_{D_0}\overset{P}{\to} 0.$

Let $a$ be a point in $D_0$ and $s\in D$. By Lemma \ref{smooth:lem:bounds}, we have that $|\hat{m}^\prime(s) -\hat{m}^\prime (a)| \le J(\hat{m}) |s-a|^{1/2}= O_p(1)$. Moreover, we have that $|\hat{m}^\prime(a)- m_0^\prime (a)|= o_p(1).$ Thus $\|\hat{m}^\prime\|_\infty=O_p(1).$
%Once consistency has been established, we provide upper bounds on the rate of convergence of $\hat{m}$ and $\hat{\theta}$ as stated below.

% subsection proof_of_existance_of_ (end)
\subsection{Proof of Lemma \ref{smooth:lem:H_lip}} \label{smooth:sec:lem:h_lip_proof}
For every $\theta \in S^{d-1}$ and $\theta\neq \theta_0$, define
\begin{equation}\label{eq:p_d_def}
\theta_{d}:= \frac{\theta_0-\theta}{|\theta_0-\theta|} \quad \text{and}\quad \theta_p := \frac{\theta_0 -\theta\theta_0^\top \theta }{|\theta_0 -\theta\theta_0^\top \theta|}
\end{equation}
Observe that $\theta^\top \theta_p =0$ and $\theta_p\in \Span\{ \theta_0,\theta\},$ where for $a_1,\ldots, a_k \in \R^d$, $\Span\{a_1,\ldots,a_k\}$ denotes the linear span of $a_1,\ldots, a_k$. Consider the following symmetric matrices in $\R^{d\times d}$:
\begin{equation}\label{smooth:eq:T_def}
T^{d}_\theta :=\mathbb{I}_d- 2\theta_{d}\theta_{d}^\top \quad \text {and} \quad  T_\theta^{p}: =\mathbb{I}_d- 2\theta_p\theta_p^\top.
\end{equation}
Note that for every $x\in \R^d$,  $x \mapsto T^{d}_\theta x$ and $x \mapsto T_\theta^{p} x$ define the reflections about the hyperplanes through $0$ which are orthogonal  to $\theta_{d}$ and $\theta_p$, respectively. More generally, for any $a\in S^{d-1}$, $T_a:=\mathbb{I}_d-2aa^\top$ is known as the Householder transformation or elementary reflector matrix;  see Page 324 of \cite{Meyer01}. It is easy to see that $T_a$ is an orthogonal matrix for every $a\in S^{d-1}$ and det$(T_a)=-1$. %Also observe that% \[T_\theta^\top T_\theta= \mathbb{I}_d -\frac{4}{|\theta_0-\theta|^2} (\theta_0-\theta) (\theta_0-\theta)^\top+ \frac{4}{|\theta_0-\theta|^4} (\theta_0-\theta) (\theta_0-\theta)^\top(\theta_0-\theta) (\theta_0-\theta)^\top =\mathbb{I}_d.\] Thus $T_\theta$ is an orthogonal matrix. $T_\theta$ is called the Householder transformation or elementary reflector matrix. We will now show that $T_\theta \theta_0 =\theta$.
As $|\theta_0|=|\theta|=1$, we have \[1=\theta_{d}^\top\theta_{d}=\frac{1}{|\theta_0-\theta|^2}(\theta_0-\theta)^\top (\theta_0-\theta)=\frac{1}{|\theta_0-\theta|^2} \big[2\theta_0^\top\theta_0 -2 \theta^\top \theta_0 \big]= \frac{2}{|\theta_0-\theta|} \theta_{d}^\top \theta_0. \]
Thus \[ T^d_\theta \theta_0 =\theta_0 - 2\theta_{d}\theta_{d}^\top \theta_0=\theta_0 -  \theta_{d} |\theta_0-\theta|   =\theta\] and as $\theta_p^\top \theta=0,$ we have  $T_\theta^p \theta =\theta.$ Now, let $\{e_{1},\ldots,e_{d}\}$ be an orthonormal basis of $\R^d$ such that $e_{1}=\theta_0$. Define
\begin{equation}\label{eq:H_def}
 H_{\theta_0} := [ e_2,\ldots,e_d] \text{ and } H_\theta:= T^{p}_\theta T^{d}_\theta H_{\theta_0}, \; \forall \theta\neq\theta_0.
\end{equation}
 As $T_\theta^p T_\theta^d$ is an orthogonal matrix, it is easy to see that $H_{\theta_0}$ and $H_\theta$ satisfy  conditions (a) and (b).  Now we will prove that $\|H_\theta-H_{\theta_0}\|_2 \le |\theta_0-\theta|$. Observe that
\begin{align*}
\|H_\theta-H_{\theta_0}\|_2 &=\sup_{\eta \in S^{d-2}} |H_\theta\eta-H_{\theta_0} \eta|\\
&=\sup_{\eta \in S^{d-2}} |T^{p}_\theta T^{d}_\theta H_{\theta_0}\eta-H_{\theta_0} \eta|\\
&=\sup_{x^\top \theta_0=0, \, x \in S^{d-1}} |T^{p}_\theta T^{d}_\theta x-x|\\
&\le\sup_{ x \in S^{d-1}} |T^{p}_\theta T^{d}_\theta x- \mathbb{I}_d x|= \|T^{p}_\theta T^{d}_\theta -\mathbb{I}_d\|_2.
\end{align*}
We will now show that $\|T^{p}_\theta T^{d}_\theta -\mathbb{I}_d\|_2 =|\theta_0-\theta|.$ The following argument shows that $T^{p}_\theta T_\theta^{d}$ is essentially a rotation operator on $\Span\{\theta,\theta_0\}$ that fixes $\Span\{\theta,\theta_0\}^\perp.$
Fix $\theta\in \Theta$.  Observe that for any orthogonal matrix $Q$, we have
\begin{align}\label{smooth:eq:H_0}
 \|T^{p}_\theta T^{d}_\theta -\mathbb{I}_d\|_2=  \|Q^\top (T^{p}_\theta T^{d}_\theta -\mathbb{I}_d)Q\|_2=\|Q^\top T^{p}_\theta T^{d}_\theta Q -\mathbb{I}_d\|_2.
\end{align}
We will try to compute the right hand side of the above display by using  a convenient choice of $Q$. Consider any orthogonal matrix $Q$ such that $\theta$ and $\theta_p$ are the first two columns of $Q$. Such a $Q$ exists as  $\theta\perp \theta_p$ and $|\theta|=|\theta_p|=1.$ By \eqref{smooth:eq:T_def} and the fact that $\theta_d \in \Span\{\theta, \theta_p\}$, we have
 \begin{align} \label{smooth:eq:H_1}
Q^\top T^{p}_\theta T^{d}_\theta Q&= \mathbb{I}_d-2Q^\top \big[\theta_{d}\theta_{d}^\top  +\theta_p\theta_p^\top -2\theta_{d}\theta_{d}^\top\theta_p\theta_p^\top\big] Q
=\begin{bmatrix}
A_\theta & \,\textbf{0}_{2\times (d-2)}\\
\;\textbf{0}_{(d-2)\times 2} & \mathbb{I}_{(d-2)}
\end{bmatrix},
\end{align}
where $A_\theta \in \R^{2\times 2}.$  As $Q^\top T^{p}_\theta T^{d}_\theta Q$ is an orthogonal matrix and det$(Q^\top T^{p}_\theta T^{d}_\theta Q)=1$,  $A_\theta$ is an orthogonal matrix and det$(A_\theta)=1$, i.e., $A_\theta$ is a rotation matrix for $\R^2$.
Note that by \eqref{smooth:eq:H_1}, we have
\begin{align}\label{smooth:eq:H_2}
Q^\top T^{p}_\theta T^{d}_\theta Q x -x=  A_\theta \begin{bmatrix} x_1\\ x_2\end{bmatrix}- \begin{bmatrix} x_1\\ x_2\end{bmatrix} \quad\text{where } x:=(x_1,x_2, \ldots, x_d)^\top  \in \R^d.
\end{align}
Thus
\begin{align}\label{smooth:eq:H_21}
\sup_{x\in S^{d-1}}\left|Q^\top T^{p}_\theta T^{d}_\theta Q x -x\right|=  \sup_{x\in S^{d-1}}\left| A_\theta \begin{bmatrix} x_1\\ x_2\end{bmatrix}- \begin{bmatrix} x_1\\ x_2\end{bmatrix}\right| =\sup_{y\in S^1} |A_\theta y-y|.
\end{align}

  However, as $A_\theta$ is a rotation  matrix and in two dimension rotation is completely determined by a angle of rotation, we have that
\be\label{smooth:eq:H_4}
 \sup_{y \in S^1} |A_\theta y- y|= |A_\theta z- z|
\ee
  for all $z\in S^1$; see Page 326, \cite{Meyer01}.
Let $z^0:=(z^0_1, z^0_2)^\top \in S^1$  be such that $\theta_0= z^0_1\theta +z^0_2\theta_p$. Define $x^0:=(z^0_1, z^0_2, 0, \ldots,0)^\top \in S^{d-1}$.  By~\eqref{smooth:eq:H_2}, we have
\begin{align}\label{smooth:eq:H_6}
|A_\theta z^0 -z^0|=|Q^\top T^{p}_\theta T^{d}_\theta Q x^0 -Q^\top Qx^0|= |Q^\top (\theta - \theta_0)| = |\theta_0-\theta|,
\end{align}
where the second equality is true due the following observation: as $Qx^0= z_1^0 \theta + z_2^0 \theta_p=\theta_0$ and $T^{p}_\theta T^{d}_\theta \theta_0=\theta$, we have  $T^{p}_\theta T^{d}_\theta Q x^0=\theta$. The last equality in the above display is true as $Q$ is an orthogonal matrix.  Thus  combining \eqref{smooth:eq:H_0}, \eqref{smooth:eq:H_2}, \eqref{smooth:eq:H_4}, and \eqref{smooth:eq:H_6}, we have $\|T_\theta^p T_\theta^d -\mathbb{I}_d\|_2 = |\theta_0-\theta|$.

Before proving $(d)$, we show that for $x\in \R^{d-1}, |H_\theta x|=|x|$ and for $y\in \R^d, |H_\theta^\top y| \le |y|$. Recall that  $T^{p}_\theta T^{d}_\theta$ is an orthogonal matrix. For $x\in \R^{d-1}$ observe that $|H_\theta x|= |H_{\theta_0} x|= |\sum_{i=1}^{d-1} x_i e_{i+1}|$, where $e_1,\ldots, e_d$ is defined in~\eqref{eq:H_def}. As $e_1,\ldots, e_d$ form an orthonormal set, we have that $|H_\theta x|= \sqrt{\sum_{i=1}^{d-1} x_i^2}=|x|$. Recall that  $T^{p}_\theta T^{d}_\theta$ is an orthogonal matrix. Thus to prove  $|H_\theta^\top y|\le | y|$, it is enough to show that $|H_{\theta_0}^\top y|\le |y|$. Let $y\in R^d$, then  $y= \sum_{i=1}^d(e_i^\top y) e_i.$ Observe that $H_{\theta_0}^\top y =\sum_{j=2}^d \sum_{i=1}^d e_i^\top y e_j^\top e_i.$ As $e_1,\ldots, e_d$ form an orthonormal set, we have $e_j^\top e_i=0$ for all $j\neq i$ and $e_i^\top e_i=1$. Thus $|H_{\theta_0}^\top y | = \sqrt{\sum_{j=2}^d (e_j^\top y)^2}\le \sqrt{\sum_{j=1}^d (e_j^\top y)^2}= |y|.$

 Now we verify that $\{H_\theta: \theta\in\Theta\}$ defined in~\eqref{eq:H_def} satisfies condition $(d)$ of Lemma~\ref{smooth:lem:H_lip}. Let $\eta, \beta \in \Theta\setminus \theta_0$ such that $ |\eta-\theta_0|<1/2$, $|\beta-\theta_0|<1/2.$  Note that
 \begin{align}
\|H_\eta^\top-H_\beta^\top\|_2 &= \|H_{\theta_0}^\top \big[ T_\eta^d T_\eta^p- T_\beta^d T_\beta^p\big]\|_2\\
&= \sup_{x\in S^{d-1}} \big|H_{\theta_0}^\top \big[ T_\eta^d T_\eta^p- T_\beta^d T_\beta^p\big] x\big|\\
&\le \sup_{x\in S^{d-1}} \big| \big( T_\eta^d T_\eta^p- T_\beta^d T_\beta^p\big) x \big|\\
&\le \sup_{x\in S^{d-1}} \big| \big( T_\eta^d T_\eta^p- T_\eta^d T_\beta^p\big) x\big|+  \sup_{x\in S^{d-1}} \big| \big(  T_\eta^d T_\beta^p- T_\beta^d T_\beta^p\big) x \big|\\
&= \sup_{x\in S^{d-1}} \big| T_\eta^d \big( T_\eta^p- T_\beta^p\big) x\big|+  \sup_{x\in S^{d-1}} \big| \big(  T_\eta^d - T_\beta^d \big) T_\beta^p x \big|\\
&= \sup_{x\in S^{d-1}} \big| \big( T_\eta^p- T_\beta^p\big) x\big|+  \sup_{x\in S^{d-1}} \big| \big(  T_\eta^d - T_\beta^d \big) x \big|\\
&=  \| T_\eta^p- T_\beta^p \|_2+  \|   T_\eta^d - T_\beta^d\|_2, \label{eq:H_diff_calc}
 \end{align}
here the first inequality is true as $|H_{\theta_0}^\top x| \le |x|$ for all $x\in\R^d$ and the penultimate equality is true as both $T_\eta^d$ and $T_\beta^p$ are orthogonal matrices in $\R^{d\times d}$. We will next show that
\begin{equation}\label{eq:diffs_H_lip}
  \| T_\eta^d- T_\beta^d \|_2 \le 4|\eta_d-\beta_d|\quad\text{and} \quad   \| T_\eta^p- T_\beta^p \|_2 \le 4|\eta_p-\beta_p|,
\end{equation} where $\eta_p, \eta_d,\beta_p,$ and $\beta_d$ are defined as in~\eqref{eq:p_d_def}. Observe that
\begin{align}\label{eq:t_diff}
\begin{split}
\| T_\eta^d- T_\beta^d \|_2 &= 2 \|\beta_d \beta_d^\top-\eta_d \eta_d^\top\|_2\\
&\le 2 \|\beta_d \beta_d^\top-\beta_d \eta_d^\top\|_2+2 \|\beta_d\eta_d^\top-\eta_d \eta_d^\top\|_2\\
&= 2 \|\beta_d (\beta_d^\top-\eta_d^\top)\|_2+ 2\|(\beta_d-\eta_d) \eta_d^\top\|_2\\
&= 2 \sup_{x\in S^{d-1}} |\beta_d (\beta_d^\top-\eta_d^\top) x|+2|\beta_d-\eta_d| \sup_{x\in S^{d-1}}| \eta_d^\top x|\\
&=2 \sup_{x\in S^{d-1}} |(\beta_d^\top-\eta_d^\top) x|+ 2|\beta_d-\eta_d|\\
&= 4 |\beta_d-\eta_d|.
\end{split}
\end{align}
A similar calculation will show the second equality in~\eqref{eq:diffs_H_lip}. The proof of~\eqref{eq:H_lip_gen} will be complete if we can show that
\begin{equation}\label{eq:bound_diff}
|\eta_d-\beta_d|\le 2 \frac{|\eta-\beta|}{|\eta-\theta_0| +|\beta-\theta_0|} \quad \text{and} \quad |\eta_p-\beta_p|\le \frac{16|\eta-\beta|/\sqrt{15}}{|\eta-\theta_0| +|\beta-\theta_0|}.
\end{equation}
Observe that by properties of projection onto the unit sphere (see Lemma 3.1 of~\cite{Kalajetal16}), we have
\begin{align}\label{eq:d_diff}
\begin{split}
|\eta_d-\beta_d|= \left|\frac{\eta-\theta_0}{|\eta-\theta_0|}-\frac{\beta-\theta_0}{|\beta-\theta_0|}\right| \le \frac{2|\eta-\beta|}{|\eta-\theta_0|+|\beta-\theta_0|} .
\end{split}
\end{align}
and
\begin{align}\label{eq:p_diff}
\begin{split}
|\eta_p-\beta_p|&= \left|\frac{\theta_0-\eta \theta_0^\top \eta}{|\theta_0-\eta \theta_0^\top \eta|}-\frac{\theta_0-\beta \theta_0^\top \beta}{|\theta_0-\beta \theta_0^\top \beta|}\right|\le \frac{2|\eta \theta_0^\top \eta-\beta \theta_0^\top \beta|}{|\theta_0-\eta \theta_0^\top \eta|+|\theta_0-\beta \theta_0^\top \beta|}.
\end{split}
\end{align}

We now try to simplify~\eqref{eq:p_diff}. First note that
$|\eta-\theta_0|\le 1/2$ implies that  $1+\theta_0^\top \eta \ge 15/8$. Now observe that
\begin{align}\label{eq:11}
\begin{split}
|\theta_0-\eta \theta_0^\top \eta |^2 &= 1- (\theta_0^\top \eta)^2 = (1- \theta_0^\top \eta) (1+\theta_0^\top \eta)\\
&= \frac{|\eta-\theta_0|^2}{2} (1+\theta_0^\top \eta)\ge\frac{|\eta-\theta_0|^2}{2} \inf_{\eta\in \Theta} (1+\theta_0^\top \eta) \ge \frac{15}{16}|\eta-\theta_0|^2.
\end{split}
\end{align}
For the numerator of~\eqref{eq:p_diff}, we have
\begin{align}\label{eq:p_dii_num}
\begin{split}
|\eta \theta_0^\top \eta-\beta \theta_0^\top \beta| \le |\eta \theta_0^\top \eta-\eta \theta_0^\top \beta|+|\eta \theta_0^\top \beta-\beta \theta_0^\top \beta|\le 2|\eta-\beta|.
\end{split}
\end{align}
Combining the above two displays, we have
\begin{equation}\label{eq:p_diff_final}
|\eta_p-\beta_p| \le   \frac{4|\eta-\beta|}{\sqrt{\frac{15}{16}}(|\eta-\theta_0|+|\beta-\theta_0|)} \le\frac{16|\eta-\beta|/\sqrt{15}}{|\eta-\theta_0|+|\beta-\theta_0|}  .
\end{equation}
Combining \eqref{eq:H_diff_calc},~\eqref{eq:diffs_H_lip}, and~\eqref{eq:bound_diff}, we have that
\[\|H_\eta^\top-H_\beta^\top\|_2 \le (8+64/\sqrt{15}) \frac{|\eta-\beta|}{|\eta-\theta_0|+|\beta-\theta_0|} \]

\subsection{Proof of Theorem \ref{smooth:thm:ratest}} \label{smooth:sec:proof_ratest}

We first state and  prove a lemma that we will use to prove Theorem~\ref{smooth:thm:ratest}.
%We use Lemma 5.7 of \cite{VANC} (stated in Section \ref{smooth:app:used_lemmas}) and the following lemma to prove Theorem \ref{smooth:thm:ratest}. See Section for more details.
\begin{lemma} \label{smooth:lem:ratest_1}
Suppose $m\in \Ss,\; J(m) < \infty,\; \text{and } \theta \in \Theta.$ Then
\begin{align}
P_X\big|m(\theta^{\top}X) - m(\theta_0^{\top}X)&- m_0'(\theta_0^{\top}X)X^{\top}(\theta - \theta_0)\big|^2\\&\lesssim |\theta_0-\theta|^{3}J^2(m) + |\theta - \theta_0|^2P_X\big|(m-m_0)'(\theta_0^{\top}X)\big|^2.
\end{align}
\end{lemma}
\begin{proof}
By the mean value theorem, we have
\begin{align}
m(\theta^{\top}x) - m(\theta_0^{\top}x) - m_0'(\theta_0^{\top}x)x^{\top}(\theta - \theta_0)&= m'(\xi^{\top}x)x^{\top}(\theta - \theta_0) - m_0'(\theta_0^{\top}x)x^{\top}(\theta - \theta_0)\\
&= \{m'(\xi^{\top}x) - m_0'(\theta_0^{\top}x)\}x^{\top}(\theta - \theta_0),
\end{align}
where $\xi^{\top}x$ lies between $\theta^{\top}x$ and $\theta_0^{\top}x$. Since $\rchi$ is bounded (see \ref{smooth:a2}), by an application of the Cauchy-Schwarz inequality, we have
\begin{align}
\big|m(\theta^{\top}x) - m(\theta_0^{\top}X) - m_0'(\theta_0^{\top}x)x^{\top}(\theta - \theta_0)\big|^2 &\lesssim |\theta - \theta_0|^2\big|m'(\xi^{\top}x) - m_0'(\theta_0^{\top}x)\big|^2\\
&\lesssim |\theta - \theta_0|^2\big|m'(\xi^{\top}x) - m^\prime(\theta_0^{\top}x)\big|^2\\
&\quad + |\theta - \theta_0|^2\big|m'(\theta_0^{\top}x) - m_0'(\theta_0^{\top}x)\big|^2.
\end{align}
By Lemma \ref{smooth:lem:bounds}, we have
\begin{align}
|m'(\xi^{\top}x) - m'(\theta_0^{\top}x)|\le J(m)|\xi^{\top}x - \theta_0^{\top}x|^{1/2}&\le J(m)|\theta^{\top}x - \theta_0^{\top}x|^{1/2}\\&\lesssim J(m)|\theta - \theta_0|^{1/2}.
\end{align}
Thus we have
\begin{align}
\big|m(\theta^{\top}x) - m_0(\theta_0^{\top}x) &-  m_0'(\theta_0^{\top}x)x^{\top}(\theta - \theta_0)\big|^2\\&\lesssim\big|m'(\theta_0^{\top}x) - m_0'(\theta_0^{\top}x)\big|^2|\theta - \theta_0|^2+ J^2(m)|\theta - \theta_0|^3,
\end{align}
and hence
\begin{align}
P_X \big| m(\theta^{\top}X) - m(\theta_0^{\top}X) &- m_0'(\theta^{\top}X)X^{\top}(\theta - \theta_0)\big|^2\\&\lesssim |\theta - \theta_0|^2P_X \big|(m-m_0)'(\theta_0^{\top}X)\big|^2+J^2(m)|\theta_0-\theta|^{3} . \qedhere
\end{align}
\end{proof}
The proof of Theorem \ref{smooth:thm:ratest} is a small modification of the proof of Theorem 3.3 of~\cite{VANC}. 
Let us define $A(x): =\hat{m}(\hat{\theta}^{\top}x) - m_0(\theta_0^{\top}x)$ and $B(x):=m_0'(\theta_0^{\top}x)x^{\top}(\hat{\theta} - \theta_0) + (\hat{m} - m_0)(\theta_0^{\top}x).$
Observe that $$A(x)-B(x)=\hat{m}(\hat{\theta}^{\top}x) -m_0'(\theta_0^{\top}x)x^{\top}(\hat{\theta} - \theta_0) - \hat{m}(\theta_0^{\top}x).$$ Recall that $|\hat{\theta} - \theta_0|\overset{P}{\to}0$, $P_X\big|(\hat{m}-m_0)'(\theta_0^{\top}X)\big|^2\overset{P}{\to}0$ and $J(\hat{m}) = O_p(1).$  Thus by Lemma \ref{smooth:lem:ratest_1}, we have that
$$P_X|A(X)-B(X)|^2 \lesssim |\hat{\theta} -\theta_0|^3 J^2(\hat{m}) + |\hat{\theta}-\theta_0|^2 P_X| (\hat{m}'-m_0')(\theta_0^\top X)|^2 =o_p(1)|\hat{\theta} -\theta_0|^2. $$
and
\begin{align}
P_X|A(X)|^2\ge \frac{1}{2} P_X|B(X)|^2- P_X|A(X)-B(X)|^2 \ge\frac{1}{2} P_X|B(X)|^2- o_p(1)|\hat{\theta} -\theta_0|^2.
\end{align}
However by Theorem \ref{smooth:thm:mainc}, we have that $P_X|A(X)|^2 =O_p(\hat{\lambda}_n^2).$
%\begin{align}\label{smooth:eq:temp1}
%\begin{split}
%O_p(\hat{\lambda}_n^2) = P_X \big|&\hat{m}(\hat{\theta}^{\top}X) - m_0(\theta_0^{\top}X)\big|^2 \\
%&\ge \frac{1}{2}P_X\big|m_0'(\theta_0^{\top}X)X^{\top}(\hat{\theta} - \theta_0) + (\hat{m} - m_0)(\theta_0^{\top}X)\big|^2 - o_p(1)|\hat{\theta} - \theta_0|^2.\end{split}
%\end{align}
Thus we have $$ P_X\big|m_0'(\theta_0^{\top}X)X^{\top}(\hat{\theta} - \theta_0) + (\hat{m} - m_0)(\theta_0^{\top}X)\big|^2  \le O_p(\hat{\lambda}_n^2) +o_p(1)|\hat{\theta} - \theta_0|^2.$$
Now define
\begin{equation}\label{smooth:eq:def_g_1}
\gamma_n := \frac{\check{\theta}_n - \theta_0}{|\check{\theta}_n - \theta_0|},\quad g_1(x) := m_0'(\theta_0^{\top}x)x^{\top}(\hat{\theta} - \theta_0)  \text{ and}\quad  g_2(x) := (\hat{m} - m_0)(\theta_0^{\top}x)
\end{equation}
%  and note that by assumption \ref{smooth:a6} there exists a $\lambda_1>0$ such that
% \begin{equation}\label{smooth:eq:step_th6}
% P_Xg_1^2 = (\hat{\theta} - \theta_0)^{\top}P_X[XX^{\top}|m_0'(\theta_0^{\top}X)|^2](\hat{\theta} - \theta_0)\ge \lambda_1(\hat{\theta} - \theta_0)^\top (\hat{\theta} - \theta_0) = \lambda_1|\hat{\theta} - \theta_0|^2.
% \end{equation}
Note that for all $n$,
\begin{align}\label{smooth:eq:step_th6}
\begin{split}
P_Xg_1^2 &= (\hat{\theta}_n - \theta_0)^{\top}P_X[XX^{\top}|m_0'(\theta_0^{\top}X)|^2](\hat{\theta}_n - \theta_0)\\
&= |\hat{\theta}_n - \theta_0|^2\, \gamma_n^\top P_X[XX^{\top}|m_0'(\theta_0^{\top}X)|^2]\gamma_n\\
&\ge |\hat{\theta}_n - \theta_0|^2\, \gamma_n^\top\E\big[\Var(X|\theta_0^\top X) |m_0'(\theta_0^{\top}X)|^2\big]\gamma_n.
\end{split}
\end{align}
Since $\gamma_n^{\top}\theta_0$ converges in probability to zero, we get by Lemma~\ref{lem:gamma_n} and assumption~\ref{smooth:a6} that with probability converging to one,
\begin{equation}\label{eq:G1SquareBound}
\frac{P_Xg_1^2}{|\hat{\theta}_n - \theta_0|^2} \ge \frac{\lambda_{\min}(H_{\theta_0}^\top\E\big[\Var(X|\theta_0^\top X) |m_0'(\theta_0^{\top}X)|^2\big] H_{\theta_0})}{2} > 0.
\end{equation}
With \eqref{smooth:eq:step_th6} in mind, we can see that proof of this theorem will be complete if we can show that
\begin{equation}\label{smooth:eq:th6_con}
P_X g_1^2 +P_X g_2^2 \lesssim  P_X\big|m_0'(\theta_0^{\top}X)X^{\top}(\hat{\theta} - \theta_0) + (\hat{m} - m_0)(\theta_0^{\top}X)\big|^2.
\end{equation}
%Note that our proof will be complete if we can show that
%%{\color{red}[The proof for the lemma above as I did was using $J(m_0)$ but the proof above uses $J(m)$ and so in the following, we need to use $J(\hat{m})$ which is $O_p(1)$ from above. Is this what is being done?]}
% If we can now show that the expectation on the right side of \eqref{smooth:eq:temp1} is bounded below by a multiple of $|\hat{\theta} - \theta_0| + P_X \big|(\hat{m} - m_0)'(\theta^{\top}X)\big|^2$, then we get the rates of convergence of $\hat{m}$ and $\hat{m}$ given in Theorem \ref{smooth:thm:ratest}. For proving this, we use the following lemma.
The following theorem gives a sufficient condition for \eqref{smooth:eq:th6_con} to hold.
\begin{lemma}\label{smooth:lem:ineq}(Lemma 5.7 of \cite{VANC})
Let $g_1$ and $g_2$ be measurable functions such that $|P_X (g_1g_2)|^2 \le cP_Xg_1^2P_Xg_2^2$ for a constant $c<1$. Then
\[
P_X(g_1+g_2)^2\ge (1-\sqrt{c})(P_Xg_1^2 + P_Xg_2^2).
\]
\end{lemma}
The following arguments show that $g_1$ and $g_2$ (defined in \eqref{smooth:eq:def_g_1}) satisfy the condition of Lemma \ref{smooth:lem:ineq}. Observe that
\begin{align}
P_X[m_0'(\theta_0^{\top}X)g_2(X) & X^{\top}(\hat{\theta} - \theta_0)]^2\\
&= P_X \big| m_0'(\theta_0^{\top}X)g_2(X)E(X^{\top}(\hat{\theta} - \theta_0)|\theta_0^{\top}X)\big|^2\\
&\le P_X\big[\{m_0'(\theta_0^{\top}X)\}^2E^2[X^{\top}(\hat{\theta} - \theta_0)|\theta_0^{\top}X]\big]P_Xg_2^2(X)\\
&= |\hat{\theta} - \theta_0|^2 \gamma_n^\top P_X\big[|m_0'(\theta_0^{\top}X)|^2E[X|\theta_0^{\top}X] E[X^{\top}|\theta_0^{\top}X]\big] \gamma_n P_Xg_2^2(X)\\
&= c_n |\hat{\theta} - \theta_0|^2 \gamma_n^\top P_X\big[|m_0'(\theta_0^{\top}X)|^2X^{\top}X\big] \gamma_n P_Xg_2^2(X)\\
&= c_n P_X g_1^2 P_Xg_2^2(X),
% &< P_X\big[\{m_0'(\theta_0^{\top}X)\}^2E[\{X^{\top}(\hat{\theta} - \theta_0)\}^2|\theta_0^{\top}X]\big]P_Xg_2^2(X)\\
% &= P_X\big[\mathbb{E}[\{m_0'(\theta_0^{\top}X)X^{\top}(\hat{\theta} - \theta_0)\}^2|\theta_0^{\top}X)]\big]P_Xg_2^2(X)\\
% &= P_X[m_0'(\theta_0^{\top}X)X^{\top}(\hat{\theta} - \theta_0)]^2P_X g_2^2(X)\\
% &= P_Xg_1^2\,P_Xg_2^2.
\end{align}
where \[c_n:= \frac{ \gamma_n^\top P_X\big[|m_0'(\theta_0^{\top}X)|^2E[X^{\top}|\theta_0^{\top}X] E[X|\theta_0^{\top}X]\big] \gamma_n }{\gamma_n^\top P_X\big[|m_0'(\theta_0^{\top}X)|^2X^{\top}X\big] \gamma_n}.
\]
To show that with probability converging to one, $c_n < 1$, observe that
\[
1 - c_n = \frac{ \gamma_n^\top\E\big[\Var(X|\theta_0^\top X) |m_0'(\theta_0^{\top}X)|^2\big]\gamma_n}{ \gamma_n^\top\E\big[XX^{\top} |m_0'(\theta_0^{\top}X)|^2\big]\gamma_n}
\]
and by Lemma~\ref{lem:gamma_n} along with assumption~\ref{smooth:a6}, with probability converging to one,
\[
1 - c_n > \frac{ 4\lambda_{\min}\left(H_{\theta_0}^{\top}\E\big[\Var(X|\theta_0^\top X) |m_0'(\theta_0^{\top}X)|^2\big]H_{\theta_0}\right)}{\lambda_{\max}\left(H_{\theta_0}^{\top}\E\big[XX^{\top} |m_0'(\theta_0^{\top}X)|^2\big]H_{\theta_0}\right)} > 0.
\]
This implies that with probability converging to one, $c_n < 1.$
% &< P_X\big[\{m_0'(\theta_0^{\top}X)\}^2E[\{X^{\top}(\hat{\theta} - \theta_0)\}^2|\theta_0^{\top}X]\big]P_Xg_2^2(X)\\
% &= P_X\big[\mathbb{E}[\{m_0'(\theta_0^{\top}X)X^{\top}(\hat{\theta} - \theta_0)\}^2|\theta_0^{\top}X)]\big]P_Xg_2^2(X)\\
% &= P_X[m_0'(\theta_0^{\top}X)X^{\top}(\hat{\theta} - \theta_0)]^2P_Xg_2^2(X)\\
% &= P_Xg_1^2\,P_Xg_2^2.
% \end{align}
% Strict inequality in the above sequence of inequalities holds under the assumption that the conditional distribution of $X$ given $\theta_0^{\top}X$ is non-degenerate.

\begin{lemma}\label{lem:gamma_n}
Suppose $A\in\mathbb{R}^{d\times d}$ and let $\{\gamma_n\}$ be any sequence of random vectors in $S^{d - 1}$ satisfying $\theta_0^{\top}\gamma_n =o_p(1).$ Then 
\[
\p\bigg(0.5 \lambda_{\min}\left(H_{\theta_0}^{\top}AH_{\theta_0}\right) \le \gamma_n^{\top}A\gamma_n \le 2 \lambda_{\max}\left(H_{\theta_0}^{\top}AH_{\theta_0}\right)\bigg) \to 1,
\]
where for any symmetric matrix $B$, $\lambda_{\min}(B)$ and $\lambda_{\max}(B)$ denote, respectively, the minimum and the maximum eigenvalues of $B.$
\end{lemma}
\begin{proof}
Note that $\text{Col}(H_{\theta_0}) \oplus\{\theta_0\} =\R^d$, thus
\begin{equation}\label{eq:OrthogonalDecomposition}
\gamma_n = \left(\gamma_n^{\top}\theta_0\right)\theta_0 + H_{\theta_0}\left(H_{\theta_0}^{\top}\gamma_n\right).
\end{equation}
Therefore,
\begin{align*}
\gamma_n^{\top}A\gamma_n &= \left[\left(\gamma_n^{\top}\theta_0\right)\theta_0 + H_{\theta_0}\left(H_{\theta_0}^{\top}\gamma_n\right)\right]^{\top}A\left[\left(\gamma_n^{\top}\theta_0\right)\theta_0 + H_{\theta_0}\left(H_{\theta_0}^{\top}\gamma_n\right)\right]\\
&= \left(\gamma_n^{\top}\theta_0\right)^2\theta_0^{\top}A\theta_0 + \left(\gamma_n^{\top}\theta_0\right)\theta_0^{\top}AH_{\theta_0}\left(H_{\theta_0}^{\top}\gamma_n\right)\\ &\qquad+ \left(\gamma_n^{\top}\theta_0\right)\left(H_{\theta_0}^{\top}\gamma_n\right)^{\top}H_{\theta_0}^{\top}A\theta_0 + \left(H_{\theta_0}^{\top}\gamma_n\right)^{\top}H_{\theta_0}^{\top}AH_{\theta_0}\left(H_{\theta_0}^{\top}\gamma_n\right).
\end{align*}
Note that $H_{\theta_0}^{\top}\gamma_n$ is a bounded sequence of vectors. Because of $\gamma_n^{\top}\theta_0$ in the first three terms above, they converge to zero in probability and so,
\[
\left|\gamma_n^{\top}A\gamma_n - \left(\gamma_n^{\top}H_{\theta_0}\right)H_{\theta_0}^{\top}AH_{\theta_0}\left(H_{\theta_0}^{\top}\gamma_n\right)\right| = o_p(1).
\]
Also, note that from~\eqref{eq:OrthogonalDecomposition},
\[
|H_{\theta_0}^{\top}\gamma_n|^2 - 1 = |\gamma_n|^2 - \left(\gamma_n^{\top}\theta_0\right)^2 - 1 = - \left(\gamma_n^{\top}\theta_0\right)^2 = o_p(1).
\]
Therefore, as $n\to\infty$,
\begin{equation}\label{eq:InterimoPOneStatement}
\left|\gamma_n^{\top}A\gamma_n - \frac{\left(\gamma_n^{\top}H_{\theta_0}\right)H_{\theta_0}^{\top}AH_{\theta_0}\left(H_{\theta_0}^{\top}\gamma_n\right)}{|H_{\theta_0}^{\top}\gamma_n|^2}\right| = o_p(1).
\end{equation}
By the definition of the minimum and maximum eigenvalues,
\[
\lambda_{\min}\left(H_{\theta_0}^{\top}AH_{\theta_0}\right) \le \frac{\left(\gamma_n^{\top}H_{\theta_0}\right)H_{\theta_0}^{\top}AH_{\theta_0}\left(H_{\theta_0}^{\top}\gamma_n\right)}{|H_{\theta_0}^{\top}\gamma_n|^2} \le \lambda_{\max}\left(H_{\theta_0}^{\top}AH_{\theta_0}\right).
\]
Thus using~\eqref{eq:InterimoPOneStatement} the result follows.
\end{proof}
\section{Proofs of results in Section~\ref{smooth:sec:SemiInf}}\label{smooth:Sec3_proof}

\subsection{Proof of existence of $r$ in~\eqref{eq:D_r}} \label{rem:StrictSubset}
% \begin{remark}
Without loss of generality assume that $\sup_{t\in D_{\theta_0}} t > 0$  and $\inf_{t\in D_{\theta_0}} t < 0$\footnote{Note that we can always do this by a simple translation of $\rchi$.}. Then there exist constants $\bar{k}$ and $\underline{k}$ and $r<1$ such that $\sup_{t\in D_{\theta}} t > \bar{k} >0$  and $\inf_{t\in D_{\theta}} t < \underline{k} <0$ for all $\theta \in B_{\theta_0}(r).$ We will prove that  $D_\theta \subsetneqq D$ for every $\theta \in B_{\theta_0}(r/2)$.

Fix $\bar{\theta}\in B_{\theta_0}(r/2)$. By definition of supremum there exists an $\bar{x}\in \rchi$ such that 
\[ \sup_{t\in D_{\bar{\theta}}} t \le \bar{\theta}^\top \bar{x}+ r\bar{k}/3.
\]
However, we have that $sup_{t\in D_{\bar{\theta}}} t > \bar{k}$. Thus $\bar{\theta}^\top \bar{x}  \ge (1-r/3)\bar{k} > 2 \bar{k}/3$. By Cauchy-Schwarz inequality, we have that $|\bar{x}| > 2\bar{k}/3.$
% \quad \text{and}\quad |\bar{x}| \ge 2\varepsilon/r.
As $\bar{\theta} \in B_{\theta_0}(r/2)$, we have that \[
  \{\theta: |\theta- \bar{\theta}| =r/2\} \subset \{\theta: |\theta- {\theta_0}| \le r\}.
\] Thus
\[ \sup_{\theta:|\theta - \theta_0| \le r} (\theta - \bar{\theta})^{\top}\bar{x} \ge \sup_{\theta:|\theta - \bar{\theta}| = r/2} (\theta - \bar{\theta})^{\top}\bar{x} = r|\bar{x}|/2 > r \bar{k}/3.\] 
Thus there exists a $\theta_*$ such that $|\theta_*-\theta_0|\le r$ and $\theta_*^\top \bar{x} >\bar{\theta}^{\top}\bar{x} + r\bar{k}/3$. Therefore
\[ \sup_{t\in D} t> \bar{\theta}^{\top}\bar{x} + r\bar{k}/3 \ge \sup_{t\in D_{\bar{\theta}}} t.
\]
Similarly, we can show that   $\inf_{t\in D} t < \inf_{t\in D_{\bar{\theta}}} t$. Thus $D_{\bar\theta} \subsetneqq D.$ 
\subsection{Unbiasedness of \texorpdfstring{$\tilde{\ell}_{\hat{\theta}, \hat{m}}$}{Lg} }\label{smooth:thm:nobias_proof}
We start with some notation. Let $P_{\theta,m}^{Y|X}$ denote the conditional distribution of $Y$ given $X$, where $Y=m(\theta^\top X) +\epsilon$. For any $(\theta,m) \in \Theta\times \mathcal{S}$ and $f \in L_2(P_{\theta,m})$,  define
% \begin{align}
\begin{equation}\label{smooth:eq:E_def}
 E_{\theta,m}( f) :=\int f dP_{\theta,m},  \quad E_{\theta,m}^X( f) :=\int_\R f dP_{\theta,m}^{Y|X}, \quad \text{ and }\quad   E_X (f) :=\int f dP_X.
\end{equation}
For $f:\rchi\rightarrow\R$ we have  $P_{\theta_0,m_0}[f(X)]= P_X(f(X))$ and $$P_{\theta_0,m_0} \big[\big(Y-m_0(\theta_0^\top X)\big)^2 f(X)\big]= E_X\big[ E_{\theta_0,m_0}^X\big[f(X)\big(Y-m_0(\theta_0^\top X)\big)^2\big] \big]=E_X\big[ f(X) \sigma^2(X)\big],$$ where $\sigma^2(x) = \E(\epsilon^2| X=x)$.
For the rest of the paper, we use $ E_{\theta,m}$ and $P_{\theta,m}$ interchangeably.
\begin{thm} \label{smooth:thm:nobias}
Under assumptions \ref{smooth:a0}--\ref{smooth:a3}, \ref{smooth:b3}, and \ref{smooth:dens},
\be \label{smooth:eq:no_bias}
P_{\theta, m_0} \tilde{\ell}_{\theta,m} =0,
\ee
for all  $\theta \in \Theta$ and $m \in \{ g \in  \Ss :   J(g) <\infty\}$.
\end{thm}
\begin{proof}
Note that by definition \eqref{smooth:eq:E_def}, we have $ E_{\theta,m_0}^X\big[ Y-m(\theta ^\top X)\big]=m_0(\theta^\top X)-m(\theta^\top X).$ Thus \begin{align}
P_{\theta, m_0} \tilde{\ell}_{\theta,m}={}& E_{\theta,m_0}[(Y-m(\theta ^\top X)) m^\prime(\theta ^\top X) K_1(X;\theta)]\\
={}&E_X \big[E_{\theta,m_0}^X[(Y-m(\theta ^\top X)) m^\prime(\theta ^\top X) K_1(X;\theta)]\big]\\
={}&E_{\theta,m_0}[(m_0\,m^\prime-m\,m^\prime)(\theta ^\top X)K_1(X;\theta) ]\\
={}& E_{\theta,m_0}\big[\E\big((m_0\,m^\prime-m\,m^\prime)(\theta ^\top X) K_1(X;\theta) |\theta ^\top X \big)\big]\\
={}&E_{\theta,m_0}\big[(m_0\,m^\prime-m\,m^\prime)(\theta ^\top X) \E\big(K_1(X;\theta) |\theta ^\top X\big)\big]\\
={}&0.\qedhere
\end{align}
\end{proof}

\subsection{Proof of (\ref{smooth:eq:ParaScore_approx}) in Theorem \ref{smooth:thm:Main_rate} } \label{smooth:app:proof_step3}

To prove \eqref{smooth:eq:ParaScore_approx}, we will need some auxiliary results on the asymptotic behavior of  $\tilde{\ell}_{\hat\theta,\hat{m}}$. We summarize them in the following lemma.
\begin{lemma} \label{smooth:lem:Consistencyofell}
Under assumptions \ref{smooth:a1}--\ref{smooth:a6}, \ref{smooth:b3}, and \ref{smooth:dens}, the PLSE satisfies
\begin{align}
P_{\theta_0, m_0} |\tilde{\ell}_{\hat{\theta},\hat{m}}- \tilde{\ell}_{\theta_0,m_0}|^2&=o_p(1),\label{smooth:eq:L_2conv}\\
P_{\hat{\theta}, m_0} |\tilde{\ell}_{\hat{\theta},\hat{m}}|^2&=O_p(1). \label{smooth:eq:L_2bound}
\end{align}
\end{lemma}
\begin{proof}
Recall that $ K_1(x;\theta) = H_\theta^\top \big(x- h_{\theta}( \theta^\top x)\big).$  To prove \eqref{smooth:eq:L_2conv}, observe that
\begin{align}
\begin{split}\label{smooth:eq:main_conti_1}
&P_{\theta_0, m_0} \big| \tilde{\ell}_{\hat{\theta},\hat{m}} - \tilde{\ell}_{\theta_0,m_0}\big|^2\\
={}&P_{\theta_0, m_0} \big| (Y-\hat{m}(\hat{\theta} ^\top X)) \hat{m}^\prime (\hat{\theta} ^\top X) K_1(X;\hat{\theta}) \\
&\quad-(Y-m_0(\theta_0 ^\top X)) {m_0}^\prime (\theta_0 ^\top X) K_1(X;\theta_0)\big|^2\\
={}&P_{\theta_0, m_0}  \big|\{(Y-m_0(\theta_0 ^\top X)) +(m_0(\theta_0 ^\top X)-\hat{m}(\hat{\theta} ^\top X))\} \hat{m}^\prime (\hat{\theta} ^\top X) K_1(X;\hat{\theta}) \\
&\quad-(Y-m_0(\theta_0 ^\top X)) {m_0}^\prime (\theta_0 ^\top X) K_1(X;\theta_0)\big|^2\\
={}&P_{\theta_0, m_0} \big|(Y-m_0(\theta_0 ^\top X))\{ \hat{m}^\prime (\hat{\theta} ^\top X) K_1(X; \hat{\theta})- {m_0}^\prime (\theta_0 ^\top X) K_1(X;\theta_0)\}\\
&\quad+(m_0(\theta_0 ^\top X)-\hat{m}(\hat{\theta} ^\top X)) \hat{m}^\prime (\hat{\theta} ^\top X) K_1(X;\hat{\theta})\big|^2\\
={}&  P_{\theta_0, m_0} [(Y-m_0(\theta_0 ^\top X))^2 \big| \hat{m}^\prime (\hat{\theta} ^\top X) K_1(X; \hat{\theta})- {m_0}^\prime (\theta_0 ^\top X) K_1(X;\theta_0)\big|^2]\\
&\quad+ P_{\theta_0, m_0} \big| (m_0(\theta_0 ^\top X)-\hat{m}(\hat{\theta} ^\top X)) \hat{m}^\prime (\hat{\theta} ^\top X) K_1(X;\hat{\theta}) \big|^2,\\
={}&  P_X \big[\sigma^2(X) | \hat{m}^\prime (\hat{\theta} ^\top X) K_1(X;\hat{\theta})- {m_0}^\prime (\theta_0 ^\top X) K_1(X;\theta_0)|^2\big]\\
&\quad+ P_{\theta_0, m_0} \big| (m_0(\theta_0 ^\top X)-\hat{m}(\hat{\theta} ^\top X)) \hat{m}^\prime (\hat{\theta} ^\top X) K_1(X;\hat{\theta}) \big|^2,\\
\le{}&\|\sigma^2(\cdot)\|_\infty P_X \big[ | \hat{m}^\prime (\hat{\theta} ^\top X) K_1(X;\hat{\theta})- {m_0}^\prime (\theta_0 ^\top X) K_1(X;\theta_0)|^2\big]\\
&\quad+ P_{\theta_0, m_0} \big| (m_0(\theta_0 ^\top X)-\hat{m}(\hat{\theta} ^\top X)) \hat{m}^\prime (\hat{\theta} ^\top X) K_1(X;\hat{\theta}) \big|^2,\\
={}& \|\sigma^2(\cdot)\|_\infty \textbf{\Rome{1}}+\textbf{\Rome{2}}\\
%\textbf{\Rome{1}}
\end{split}
\end{align}
where  in the fourth equality, the cross product term is zero as $E_{\theta_0,m_0}^X\big(Y-m_0(\theta_0^\top X)\big)=0$ and
\begin{align}
\textbf{\Rome{1}}&:= P_X \big[ | \hat{m}^\prime (\hat{\theta} ^\top X) K_1(X;\hat{\theta})- {m_0}^\prime (\theta_0 ^\top X) K_1(X;\theta_0)|^2\big], \\
\textbf{\Rome{2}}&:= P_X \big[|(m_0(\theta_0 ^\top X)-\hat{m}(\hat{\theta} ^\top X)) \hat{m}^\prime (\hat{\theta} ^\top X) K_1(X;\hat{\theta})|^2\big].
\end{align}

Recall that  for all $a \in \R^d$, we have  $|H_\theta^\top a | \le |a|$; see proof of Lemma~\ref{smooth:lem:H_lip}. We will now show that  $\textbf{\Rome{1}} =o_p(1).$ Observe that
\begin{align}
 \textbf{\Rome{1}}&\le  2 P_X \Big[ \Big| H_{\theta_0}^\top \Big( ( \hat{m}^\prime (\hat{\theta} ^\top X)- {m_0}^\prime (\theta_0 ^\top X) ) X + ({m_0}^\prime\, h_{\theta_0}) (\theta_0 ^\top X)-  (\hat{m}^\prime\, h_{\hat{\theta}})(\hat{\theta} ^\top X)\Big)\Big|^2\Big]\\
 &\qquad +2 P_X \Big[ \Big| \big(H^\top_{\hat\theta}-H^\top_{\theta_0}\big) \hat{m}^\prime (\hat{\theta} ^\top X) (X-h_{\hat{\theta}}(\hat{\theta} ^\top X)) \Big|^2 \Big] \\
&\le 2 P_X \Big[ \Big| H_{\theta_0}^\top \Big( ( \hat{m}^\prime (\hat{\theta} ^\top X)- {m_0}^\prime (\theta_0 ^\top X) ) X + ({m_0}^\prime\, h_{\theta_0}) (\theta_0 ^\top X)-  (\hat{m}^\prime\, h_{\hat{\theta}})(\hat{\theta} ^\top X)\Big)\Big|^2\Big]\\
&\qquad +[4CT(1+J(\hat{m}))]^2 |\hat\theta-\theta_0|^2 \\
&\le 2 P_X \Big[ \Big|  ( \hat{m}^\prime (\hat{\theta} ^\top X)- {m_0}^\prime (\theta_0 ^\top X) ) X + ({m_0}^\prime\, h_{\theta_0}) (\theta_0 ^\top X)-  (\hat{m}^\prime\, h_{\hat{\theta}})(\hat{\theta} ^\top X)\Big|^2\Big],\\
&\qquad +[4CT(1+J(\hat{m}))]^2 |\hat\theta-\theta_0|^2,
\end{align}
where the second inequality follows from $(c)$ of Lemma~\ref{smooth:lem:H_lip}. Let us define  $$ \textbf{\Rome{3}}:=  4  P_X\big|({m_0}^\prime\, h_{\theta_0}) (\theta_0 ^\top X)-  (\hat{m}^\prime\, h_{\hat{\theta}})(\hat{\theta} ^\top X)\big|^2.$$  Using Lemma \ref{smooth:lem:bounds} and the fact that  $\sup_{x\in\rchi} |x|\le T$ (see \ref{smooth:a2}), we have
\begin{align}
 \textbf{\Rome{1}}&\leq 4T^2 P_X |\hat{m}^\prime (\hat{\theta} ^\top X)- {m_0}^\prime ({\theta_0} ^\top X)|^2 +  \textbf{\Rome{3}} +o_p(1)\\
&\leq 8 T^2 P_X| \hat{m}^\prime (\hat{\theta} ^\top X)- \hat{m}^\prime ({\theta_0} ^\top X)|^2 + 8 T^2 P_X |(\hat{m}^\prime - {m_0}^\prime) ({\theta_0} ^\top X) |^2+  \textbf{\Rome{3}}  +o_p(1)\\
&\leq 8 T^2 J^2(\hat{m}) P_X \big[|\hat{\theta}^\top X-{\theta_0}^\top X|\big] + 8 T^2 \|\hat{m}^\prime- {m_0}^\prime\|_{D_0} ^2+  \textbf{\Rome{3}}  +o_p(1)\\
&\leq8 T^2 J^2(\hat{m}) T|\hat{\theta}-{\theta_0}|+ 8 T^2 \|\hat{m}^\prime- {m_0}^\prime\|_{D_0} ^2 +  \textbf{\Rome{3}}+o_p(1).
\end{align}
 Recall that both $|\hat{\theta}-\theta_0|$ and $\|\hat{m}^\prime- {m_0}^\prime\|_{D_0}$ are $o_p(1)$; see  Theorem \ref{smooth:thm:cons}.
Thus we  have $\textbf{\Rome{1}}=o_p(1)$, if we can show  that $\textbf{\Rome{3}}=o_p(1).$ First observe that by Theorem \ref{smooth:thm:cons} and assumption \ref{smooth:b3}, we have that   $P_X \big|  h_{\theta_0} (\theta_0 ^\top X) -   h_{\hat{\theta}} (\hat{\theta} ^\top X)\big|^2\stackrel{P}{\rightarrow} 0$. Hence we can bound $\textbf{\Rome{3}}$ from above:  \begin{align}
\textbf{\Rome{3}} &=  4  P_X\big| ({m_0}^\prime h_{\theta_0}) (\theta_0 ^\top X)-  {m_0}^\prime ({\theta_0} ^\top X) h_{\hat{\theta}} (\hat{\theta} ^\top X)   + {m_0}^\prime ({\theta_0} ^\top X)  h_{\hat{\theta}} (\hat{\theta} ^\top X)   -  (\hat{m}^\prime  h_{\hat{\theta}}) (\hat{\theta} ^\top X)\big|^2\\
&\leq  8  P_X\big| ({m_0}^\prime ~ h_{\theta_0}) (\theta_0 ^\top X) -  {m_0}^\prime ({\theta_0} ^\top X) ~ h_{\hat{\theta}} (\hat{\theta} ^\top X)\big|^2   + 8P_X\big|{m_0}^\prime ({\theta_0} ^\top X)  h_{\hat{\theta}} (\hat{\theta} ^\top X)   -  (\hat{m}^\prime  h_{\hat{\theta}}) (\hat{\theta} ^\top X)\big|^2\\
&\leq  8  \| {m_0}^\prime\|^2_\infty P_X\big|  h_{\theta_0} (\theta_0 ^\top X) -   h_{\hat{\theta}} (\hat{\theta} ^\top X)\big|^2   + 8\| h_{\hat{\theta}} \big\|^2_{2, \infty} P_X|{m_0}^\prime ({\theta_0} ^\top X) -  \hat{m}^\prime (\hat{\theta} ^\top X)|^2\\
&\leq  8  \| {m_0}^\prime\|^2_\infty P_X\big|  h_{\theta_0} (\theta_0 ^\top X) -   h_{\hat{\theta}} (\hat{\theta} ^\top X)\big|^2     \\
&\qquad + 16\| h_{\hat{\theta}} \big\|^2_{2, \infty} \Big[ P_X|({m_0}^\prime  -  \hat{m}^\prime) (\theta_0 ^\top X)\big|^2 +P_X| \hat{m}^\prime (\theta_0 ^\top X) -  \hat{m}^\prime (\hat{\theta} ^\top X)|^2\Big]\\
&\leq  8  \| {m_0}^\prime\|^2_\infty P_X\big|  h_{\theta_0} (\theta_0 ^\top X) -   h_{\hat{\theta}} (\hat{\theta} ^\top X)\big|^2   + 16 \| h_{\hat{\theta}} \big\|^2_{2, \infty} \Big[ \|{m_0}^\prime  -  \hat{m}^\prime \|_{D_0}^2 + J^2(\hat{m})  T^2 |\hat{\theta} -\theta_0|^2 \Big].\end{align}
As each of the terms in the last inequality of the above display are $o_p(1)$, we have that $\textbf{\Rome{3}}=o_p(1).$ The proof of \eqref{smooth:eq:L_2conv} will be complete, if we can show that $\textbf{\Rome{2}}=o_p(1).$ First  note that for all $x\in \rchi$,
\begin{equation} \label{smooth:eq:K_bound}
|K_1(x;{\theta})|\le |H_\theta^\top (x- h_{\theta} (\theta^\top x))| \le  |x- h_{\theta} (\theta^\top x)| \le {2T}.
\end{equation}
By Theorem \ref{smooth:thm:mainc}  and assumption  \ref{smooth:a4}, we have
\begin{align}
\textbf{\Rome{2}}&=P_X \big[|(m_0(\theta_0 ^\top X)-\hat{m}(\hat{\theta} ^\top X)) \hat{m}^\prime (\hat{\theta} ^\top X) K_1(X;\hat{\theta})|^2\big]\\
&\le 4T^2  \|\hat{m}^\prime\|^2 _\infty P_X| (m_0(\theta_0 ^\top X)-\hat{m}(\hat{\theta} ^\top X))|^2  \stackrel {P}{\rightarrow} 0.
\end{align}
All these facts combined prove that $P_{\theta_0, m_0} | \tilde{\ell}_{\hat{\theta},\hat{m}}- \tilde{\ell}_{\theta_0,m_0}|^2=o_p(1).$

Next we prove \eqref{smooth:eq:L_2bound}. Observe that
\begin{align}
P_{\hat{\theta}, m_0} | \tilde{\ell}_{\hat{\theta},\hat{m}}|^2
&=P_{\hat{\theta}, m_0} \big| (Y-\hat{m}(\hat{\theta} ^\top X)) \hat{m}^\prime (\hat{\theta} ^\top X) K_1(X;\hat{\theta})\big|^2 \\
&=P_{\hat{\theta}, m_0} \big| (Y-{m}_0 (\hat{\theta} ^\top X)+ {m}_0 (\hat{\theta} ^\top X) -\hat{m}(\hat{\theta} ^\top X)) \hat{m}^\prime (\hat{\theta} ^\top X) K_1(X;\hat{\theta})\big|^2 \\
&\leq 4T^2  \|\hat{m}^\prime\|_\infty^2  P_{\hat{\theta}, m_0} [ (Y-{m}_0 (\hat{\theta} ^\top X)+ {m}_0 (\hat{\theta} ^\top X) -\hat{m}(\hat{\theta} ^\top X))  ]^2 \\
&=  4T^2 \|\hat{m}^\prime\|_\infty^2  P_{\hat{\theta}, m_0} [ (Y-{m}_0 (\hat{\theta} ^\top X))^2+ ({m}_0 (\hat{\theta} ^\top X) -\hat{m}(\hat{\theta} ^\top X)) ^2] \\
&=  4T^2 \|\hat{m}^\prime\|_\infty^2  \Big[ P_X |\sigma^2(X)|+  P_X |{m}_0 ({\hat\theta} ^\top X) -\hat{m}(\hat{\theta} ^\top X)| ^2 \Big]
%&= 4 T^2\|\hat{m}^\prime\|_\infty^2 [ \sigma^2 +  \|{m}_0 ({\theta_0} ^\top X) -\hat{m}(\hat{\theta} ^\top X)\|^2 ]
=O_p(1),
\end{align}
where  in the penultimate equality, the cross product term is zero as $E_{\hat\theta,m_0}^X(Y-m_0(\hat \theta^\top X))=0.$
%{\color{blue} [Fourth line above would again be an equality with 4 instead of 8.]}
\end{proof}
%Let  $\hat\eta\in \R^{d-1}$ be defined be the unique solution to the following equation
%\begin{equation}\label{smooth:eq:eta}
%\hat{\theta}=\sqrt{1-|\hat\eta|^2} \theta_0 +H_{\theta_0}\hat{\eta}.
%\end{equation}
%As  $H_{\theta_0}^\top \theta_0=0,$ we have
%\begin{equation}
%|\hat{\theta}-\theta_0| =| \sqrt{1-|\hat{\eta}|^2} -1 | +|\hat{\eta}|,
%\end{equation}
%and $|\hat{\eta}| \le |\hat{\theta}-\theta_0|$. Moreover we have
%\begin{equation}
%\hat{\eta}- H_{\theta_0}^\top (\hat{\theta}-\theta_0) =H_{\theta_0}^\top\theta_0 [1- \sqrt{1-|\hat{\eta}|^2} ].
%\end{equation}
Now we prove  \eqref{smooth:eq:ParaScore_approx}. For $\theta \in \Theta$ and $m\in \Ss$, define $p_{\theta,m}(y,x) := p_{\epsilon|X}(y-m(\theta^\top x), x) p_X(x)$ to be the joint density of $(Y,X)$ with respect to the dominating measure $\mu$, where $Y=m(\theta^\top X)+\epsilon$ and $X\sim P_X$.  Now consider the following submodel for $\theta_0$:
\begin{equation}
\zeta_{\eta,\theta_0} = \sqrt{1-|\eta|^2} \theta_0 + H_{\theta_0} \eta.
\end{equation}
By definition of $\hat{\eta}$ (see \eqref{smooth:eq:local_hat}), we have that $\zeta_{\hat{\eta},\theta_0}=\hat{\theta}$.  As $\hat{\eta}=o_p(1)$ (see Theorem \ref{smooth:thm:ratest} and \eqref{smooth:eq:eta_hat})  differentiability  in quadratic mean  of  model \eqref{smooth:simsl} implies that

\begin{equation}\label{smooth:eq:qmd_theta_0}
\int \left(\sqrt{p_{\hat{\theta},m_0}} -\sqrt{p_{\theta_0,m_0}} -\frac{1}{2} \hat{\eta}^\top S_{\theta_0,m_0} \sqrt{p_{\theta_0,m_0}} \right)^2 d\mu =o_p(|\hat{\eta}|^2) = o_p(|\hat{\theta}-\theta_0|^2).
\end{equation}
With Lemma \ref{smooth:lem:Consistencyofell} in hand, we now show that \eqref{smooth:eq:ParaScore_approx} holds. The following steps are very similar to the proof of Theorem 6.20 of~\cite{MR1915446}.  Note that by \eqref{smooth:eq:eta_hat}, we have
\begin{align}
%&\sqrt{n} (P_{\hat{\theta}, m_0} - P_{\theta_0, m_0}) \tilde{\ell}_{\hat{\theta},\hat{m}} -   \sqrt{n}\tilde{I}_{\theta_0,m_0} H_{\theta_0}^\top (\hat{\theta}- \theta_0)\\
%={}&\sqrt{n} (P_{\hat{\theta}, m_0} - P_{\theta_0, m_0}) \tilde{\ell}_{\hat{\theta},\hat{m}} - \sqrt{n}  P_{\theta_0, m_0}( \tilde{\ell}_{\theta_0, m_0} \tilde{\ell}_{\theta_0, m_0}^\top) H_{\theta_0}^\top (\hat{\theta}- \theta_0)\\
%=
\sqrt{n} (P_{\hat{\theta}, m_0} - P_{\theta_0, m_0}) \tilde{\ell}_{\hat{\theta},\hat{m}} - \sqrt{n}  P_{\theta_0, m_0}( \tilde{\ell}_{\theta_0, m_0} S_{\theta_0,m_0}^\top) H_{\theta_0}^\top (\hat{\theta}- \theta_0)
={}& \textbf{\Rome{4}} +\frac{1}{2}\textbf{\Rome{5}} +  \textbf{\Rome{6}},
\end{align}
where \begin{align}
\textbf{\Rome{4}}&= \sqrt{n} \int \tilde{\ell}_{\hat{\theta},\hat{m}} (\sqrt{p_{\hat{\theta}, m_0}} + \sqrt{p_{\theta_0, m_0}}) \left(\sqrt{p_{\hat{\theta}, m_0}} -
 \sqrt{p_{\theta_0, m_0}} -\frac{1}{2}\hat{\eta}^\top S_{\theta_0,m_0}  \sqrt{p_{\theta_0, m_0}} \right) d\mu,\\
 \textbf{\Rome{5}}&= \sqrt{n} \left[\int \tilde{\ell}_{\hat{\theta},\hat{m}} (\sqrt{p_{\hat{\theta}, m_0}} - \sqrt{p_{\theta_0, m_0}}) S_{\theta_0,m_0}^\top \sqrt{p_{\theta_0, m_0}}d\mu\right] \hat{\eta},\\
  \textbf{\Rome{6}}&= \sqrt{n} \left[\int [\tilde{\ell}_{\hat{\theta}, \hat{m}}-\tilde{\ell}_{\theta_0, m_0}] S_{\theta_0,m_0}^\top p_{\theta_0, m_0} d\mu\right]\hat{\eta}.
\end{align}
Observe that $\textbf{\Rome{4}}, \textbf{\Rome{5}},$ and  $\textbf{\Rome{6}}$  are elements of $\R^d$. In the following, we  show that  $\textbf{\Rome{4}}, \textbf{\Rome{5}},$ and  $\textbf{\Rome{6}}$ are $o_p(\sqrt{n} |\hat\theta-\theta_0|)$.
Using the Cauchy-Schwarz inequality and  the fact that $(a+b)^2 \le 2(a^2+b^2)$, we have
\begin{align}%\label{smooth:eq:4_final}
\big|\textbf{\Rome{4}}\big|^2 &\leq
 % 2 n \int |\tilde{\ell}_{\hat{\theta},\hat{m}}|^2 ({p_{\hat{\theta}, m_0}} + {p_{\theta_0, m_0}}) d\mu \int \left(\sqrt{p_{\hat{\theta}, m_0}} - \sqrt{p_{\theta_0, m_0}} -\frac{1}{2} \hat{\eta}^\top S_{\theta_0,m_0}\sqrt{p_{\theta_0, m_0}} \right)^2 d\mu\\
 2 n \left[P_{\hat{\theta},m_0} | \tilde{\ell}_{\hat{\theta},\hat{m}}|^2 + P_{\theta_0, m_0} | \tilde{\ell}_{\hat{\theta},\hat{m}}- \tilde{\ell}_{\theta_0,m_0}|^2 + P_{\theta_0, m_0} |\tilde{\ell}_{\theta_0,m_0}|^2 \right] o_p(|\hat \eta|^2 )
= o_p(n |\hat{\theta}-\theta_0|^2),
\end{align}
where the equality is due to Lemma \ref{smooth:lem:Consistencyofell}, \eqref{smooth:eq:qmd_theta_0}, and the fact that $\tilde{\ell}_{\theta_0,m_0} \in L_2(P_{\theta_0,m_0})$; see \ref{smooth:a1}, \ref{smooth:a2}, and Lemma \ref{smooth:lem:dbounds}.

Now we will show that $|\textbf{\Rome{6}}| =o_p(|\sqrt{n}(\theta-\theta_0)|).$ For a matrix $\mathbb{A} \in \R^{d\times d }$, let $\|\mathbb{A}\|_F$ denote the Frobenius norm of $\mathbb{A}$. Then we have
 \begin{align}\label{smooth:eq:6_1}
\begin{split}
 \big|\textbf{\Rome{6}}\big|^2&\le  \left\|\int [\tilde{\ell}_{\hat{\theta}, \hat{m}}-\tilde{\ell}_{\theta_0, m_0}] S_{\theta_0,m_0}^\top p_{\theta_0, m_0} d\mu\right\|_F^2 |\sqrt{n}\hat{\eta}|^2.
 \end{split}
 \end{align}
Let $f= (f_1,\ldots,f_d)$ and  $g=(g_1,\ldots,g_d)$ be two functions that map a separable metric space $\Re$ to  $\R^d$. If $\nu$ is a finite measure on $\Re$ such that  $|f| $ and $|g| $ are $ L_2(\nu)$, then by the Cauchy-Schwarz inequality, we have
\begin{align}\label{smooth:eq:forb_pdt}
\begin{split}
\left\| \int_\Re f g^\top d\nu \right\|_F^2 &= \sum_{i,j} \left[\int_\Re f_i g_j d\nu\right]^2 \le \sum_{i,j} \int_\Re f_i^2 d\nu \int_\Re g_j^2 d\nu\\
 &=  \left[\sum_{i} \int_\Re f_i^2 d\nu \right]\left[\sum_{j}\int_\Re g_j^2 d\nu\right] =\int_\Re |f|^2 d\nu\int_\Re |g|^2 d\nu.
\end{split}
\end{align} Thus from Lemma \ref{smooth:lem:Consistencyofell}, \eqref{smooth:eq:6_1}, and the fact that $S_{\theta_0,m_0} \in L_2(P_{\theta_0,m_0})$, we have
  \begin{align}\label{smooth:eq:6_final}
  \begin{split}
 \big|\textbf{\Rome{6}}\big|^2
  &\leq |\sqrt{n}\hat{\eta}|^2 \int |\tilde{\ell}_{\hat{\theta},\hat{m}}-\tilde{\ell}_{\theta_0, m_0}|^2 p_{\theta_0, m_0} d\mu \int |S_{\theta_0,m_0}|^2 p_{\theta_0, m_0} d\mu \\
&= |\sqrt{n}\hat{\eta}|^2 P_{\theta_0, m_0} | \tilde{\ell}_{\hat{\theta},\hat{m}}- \tilde{\ell}_{\theta_0,m_0}|^2 P_{\theta_0, m_0} |S_{\theta_0,m_0}|^2  =o_p( |\sqrt{n}\hat{\eta}|^2) =o_p(|\sqrt{n}(\hat\theta-\theta_0)|^2).
\end{split}
\end{align}
%In Lemma \ref{smooth:lem:Consistencyofell}, we showed that  $P_{\hat{\theta},m_0} | \tilde{\ell}^2_{\hat{\theta},\hat{m}}|^2=O_p(1)$ and $P_{\theta_0, m_0} | \tilde{\ell}_{\hat{\theta},\hat{m}}- \tilde{\ell}_{\theta_0,m_0}|^2= o_p(1).$ Thus, we have $\textbf{\Rome{4}}= o_p(\sqrt{n}(\hat{\theta}- \theta_0))$ and  $\textbf{\Rome{6}}= o_p(1).$
%In the following, we use the Cauchy-Schwarz inequality to to prove that $\textbf{\Rome{2}}$ and $ \textbf{\Rome{3}}$ are $o_p(1)$.
We will now prove that
\begin{equation}\label{eq:part_5_1}
|\textbf{\Rome{5}}|^2=o_p(|\sqrt{n}(\hat\theta-\theta_0)|^2).
\end{equation}
Observe that
\begin{equation}\label{smooth:eq:bd_5}
|\textbf{\Rome{5}}|^2\le  \left\|\int \tilde{\ell}_{\hat{\theta},\hat{m}} (\sqrt{p_{\hat{\theta}, m_0}} - \sqrt{p_{\theta_0, m_0}}) S_{\theta_0,m_0}^\top \sqrt{p_{\theta_0, m_0}}d\mu\right\|_F^2 |\sqrt{n}\hat{\eta}|^2.
\end{equation}
Thus the proof of~\eqref{eq:part_5_1} will be complete, if we can show that
\begin{equation}\label{smooth:eq:bd_51}
\left\|\int \tilde{\ell}_{\hat{\theta},\hat{m}} (\sqrt{p_{\hat{\theta}, m_0}} - \sqrt{p_{\theta_0, m_0}}) S_{\theta_0,m_0}^\top \sqrt{p_{\theta_0, m_0}}d\mu\right\|_F^2 =o_p(1).
\end{equation}
 We will show this by splitting the integral in the above display into two regions that depend on $n$. More specifically by splitting the integral into $\{ (y,x): |S_{\theta_0,m_0}(y,x)| > r_n\}$ and $\{ (y,x): |S_{\theta_0,m_0}(y,x)| \leq r_n\},$ where  $\{r_n\}$ is a sequence of constants to be chosen later.

  Observe that  by \eqref{smooth:eq:forb_pdt}, we have
\begin{align} \label{smooth:eq:(b)_part1}
 &\left\|\int_{|S_{\theta_0,m_0}| \leq r_n} \tilde{\ell}_{\hat{\theta},\hat{m}}S_{\theta_0,m_0}^\top (\sqrt{p_{\hat{\theta}, m_0}} - \sqrt{p_{\theta_0, m_0}}) \sqrt{p_{\theta_0, m_0}}d\mu\right\|_F^2 \nonumber\\
 ={}&\sum_{i,j}  \left[\int_{|S_{\theta_0,m_0}|  \leq r_n}\left\{\tilde{\ell}_{\hat{\theta},\hat{m}}S_{\theta_0,m_0}^\top\right\}_{i,j} \left(\sqrt{p_{\hat{\theta}, m_0}} - \sqrt{p_{\theta_0, m_0}}\right) \sqrt{p_{\theta_0, m_0}}d\mu\right]^2  \nonumber\\
 \le{}&  \left[ \int \left( \sqrt{p_{\hat{\theta}, m_0}} - \sqrt{p_{\theta_0, m_0}} \right)^2 d\mu \right]  \sum_{i,j} \left[\int_{|S_{\theta_0,m_0}|  \leq r_n}\left\{\tilde{\ell}_{\hat{\theta},\hat{m}}S_{\theta_0,m_0}^\top\right\}_{i,j}^2 p_{\theta_0, m_0}d\mu\right]\nonumber \\
 \le{}& 2 \left[ \int \left(\frac{1}{2}S_{\theta_0,m_0}^\top (\hat{\theta}- \theta_0) \sqrt{p_{\theta_0, m_0}} \right)^2 + \left( \sqrt{p_{\hat{\theta}, m_0}} - \sqrt{p_{\theta_0, m_0}} -\frac{1}{2}S_{\theta_0,m_0}^\top (\hat{\theta}- \theta_0) \sqrt{p_{\theta_0, m_0}} \right)^2 d\mu \right]\nonumber \\
 &\qquad \qquad \times \sum_{i,j} \left[\int_{|S_{\theta_0,m_0}|  \leq r_n}\left\{\tilde{\ell}_{\hat{\theta},\hat{m}}S_{\theta_0,m_0}^\top\right\}_{i,j}^2 p_{\theta_0, m_0}d\mu\right] \\
 ={}&2 \left[ \int \left(\frac{1}{2}S_{\theta_0,m_0}^\top (\hat{\theta}- \theta_0) \sqrt{p_{\theta_0, m_0}} \right)^2 + \left( \sqrt{p_{\hat{\theta}, m_0}} - \sqrt{p_{\theta_0, m_0}} -\frac{1}{2}S_{\theta_0,m_0}^\top (\hat{\theta}- \theta_0) \sqrt{p_{\theta_0, m_0}} \right)^2 d\mu \right] \nonumber\\
 &\qquad\qquad\times  \int_{|S_{\theta_0,m_0}|  \leq r_n} |\tilde{\ell}_{\hat{\theta},\hat{m}}|^2 |S_{\theta_0,m_0}^\top|^2 p_{\theta_0, m_0}d\mu \nonumber\\
\leq{}&2 r_n^2 P_{\theta_0, m_0} |\tilde{\ell}_{\hat{\theta},\hat{m}} |^2\, \big[O_p(|\hat{\theta}- \theta_0|^2) + o_p(|\hat{\theta}- \theta_0|^2)\big]= r_n^2 o_p(1)\nonumber,
\end{align}
where the last equality follows from~Theorem~\ref{smooth:thm:ratest} and \eqref{smooth:eq:L_2bound}. Now to bound the second part of the integral, observe that
\begin{align} \label{smooth:eq:(b)_part2}
\begin{split}
&\left\|\int_{|S_{\theta_0,m_0}| > r_n} \tilde{\ell}_{\hat{\theta},\hat{m}}S_{\theta_0,m_0}^\top (\sqrt{p_{\hat{\theta}, m_0}} - \sqrt{p_{\theta_0, m_0}}) \sqrt{p_{\theta_0, m_0}}d\mu\right\|_F^2\\
\le{}&2\int_{|S_{\theta_0,m_0}| > r_n} |S_{\theta_0,m_0}|^2 p_{\theta_0, m_0} d\mu \int |\tilde{\ell}_{\hat{\theta},\hat{m}}|^2 ({p_{\hat{\theta}, m_0}} + {p_{\theta_0, m_0}}) d\mu\\
\le{}& O_p(1) \int_{|S_{\theta_0,m_0}| > r_n} |S_{\theta_0,m_0}|^2 p_{\theta_0, m_0} d\mu.
\end{split}
\end{align}
Since $ P_{\theta_0, m_0}|S_{\theta_0,m_0}|^2=O_p(1),$ it is easy to see that we can find a sequence $\{r_n\}$ such that both \eqref{smooth:eq:(b)_part1} and \eqref{smooth:eq:(b)_part2} are $o_p(1).$ Thus we have~\eqref{smooth:eq:bd_51}.

\subsection{Proof of Lemma \ref{smooth:lem:first_term}} \label{smooth:app:proof_first_term}
Before proceeding to prove Lemma~\ref{smooth:lem:first_term}, we find the entropy of the class of matrices $\{ H_\theta: \theta\in \Theta\}$, where $H_\theta$ satisfies properties of Lemma~\ref{smooth:lem:H_lip}.
\begin{lemma}\label{lem:H_ent}
We can construct a cover $\{\eta_1, \ldots, \eta_{N_\varepsilon}\}$ of $\Theta \cap B_{\theta_0}(1/2)$ such that $N_\varepsilon\lesssim \varepsilon^{-2d}$ and for every $\theta\in\Theta \bigcap B_{\theta_0}(1/2)$, there exists an $i\le N_\varepsilon$ such that
\begin{equation}\label{smooth:eq:H_cover}
|\theta-\eta_i|\le \varepsilon \text{  and  }\|H_\theta^\top-H_{\eta_i}^\top\|_2\le\varepsilon.
\end{equation}
% There exists an $\varepsilon$-cover of size $\varepsilon^{-2d}$ $\{\eta_1, \ldots, \eta_{N_\varepsilon}$of $\Theta \cap B_{\theta_0}(1/2)\}$ such that for every $\theta \in \Theta \cap B_{\theta_0}(1/2)\}$   \label{lem:H_entopy}$
%  N(\varepsilon,\{ H_\theta: \theta\in \Theta \cap B_{\theta_0}(1/2)\}, \|\cdot\|_2 )\lesssim \varepsilon^{-2d}.$
\end{lemma}
\begin{proof}
To find the entropy with respect to the matrix $2$-norm, we construct a $\varepsilon$-cover for the set $\{ H_\theta^\top: \theta\in \Theta\}$. By Lemma 4.1 of~\cite{Pollard90}, we have that
\[ N(\varepsilon^2/(8+64/\sqrt{15}),\Theta\cap B_{\theta_0}(1/2)\setminus B_{\theta_0}(\varepsilon/2), |\cdot| )\lesssim \varepsilon^{-2d}.\]
 Let $\{\theta_i\}_{1\le i\le N_\varepsilon}$ for $N_\varepsilon \lesssim \varepsilon^{-2d}$ form a cover of $\Theta\cap B_{\theta_0}(1/2)\setminus B_{\theta_0}(\varepsilon/2)$. We can without loss of generality assume that $|\theta_i-\theta_0| \ge \varepsilon/2$ for all $1\le i\le N_\varepsilon.$
We claim that $H_{\theta_0}^\top \cup \{ H_{\theta_i}^\top\}_{1\le i\le N_\varepsilon}$ forms a $\varepsilon$-cover for $\{ H_\theta^\top: \theta\in \Theta\}.$ It is enough to show that for every $\eta\in \Theta$, we can find $i^* \in \{0,1,\ldots, N_\varepsilon\}$ such that $\|H_\eta^\top-H_{\theta_{i^*}}^\top\|_2 \le \varepsilon$. If $\eta \in B_{\theta_0}(\varepsilon/2)$ then choose $i^*=0$. By condition $(c)$ of Lemma~\ref{smooth:lem:H_lip}, we have $\|H_\eta^\top-H_{\theta_0}^\top\|_2 \le |\eta-\theta_0|\le \varepsilon.$ If $\eta \notin B_{\theta_0}(\varepsilon/2)$ then choose $i^*$ such that $|\eta-\theta_{i^*}| \le \varepsilon^2/(8+64/\sqrt{15})$. Thus by condition $(d)$ of Lemma~\ref{smooth:lem:H_lip}, we have \[\|H_\eta^\top-H_{\theta_{i^*}}^\top\|_2 \le (8+64/\sqrt{15}) \frac{|\eta-\theta_{i^*}|}{|\eta-\theta_0| +|\theta_{i^*}-\theta_0|} \le (8+64/\sqrt{15}) \frac{ \varepsilon^2/(8+64/\sqrt{15}) }{\varepsilon} \le \varepsilon.\qedhere \]
\end{proof}
Now we will show that  $D_{M_1, M_2,M_3}(n)$ is an envelope of $\mathcal{D}_{M_1, M_2,M_3}(n).$ For every $(m,\theta) \in \mathcal{C}_{M_1,M_2, M_3}(n)$ and $x\in \rchi$, we have
\begin{align}
&|(m_0 (\theta_0 ^\top x) -m(\theta ^\top x)) m^\prime (\theta ^\top x) K_1(x;\theta)|\\
\le{}& \big(|(m_0 (\theta_0 ^\top x)-m(\theta_0 ^\top x)|+ |m(\theta_0 ^\top x) -m(\theta ^\top x))|\big) M_2 2T\\
\le{}& \big(\|m_0 -m\|_{D_0}+ \|m^\prime\|_\infty |\theta_0-\theta| |x|) M_2 2T \\
\le{}& 2TM_2 (a_n^{-1}+ \|m^\prime\|_\infty |\theta_0-\theta| T)\\
\le{}& 2TM_2 (a_n^{-1}+T M_2 \hat\lambda_n^{1/2})=D_{M_1, M_2,M_3}(n),
\end{align}
 where the first and second inequality follow from the facts that  $\sup_{x\in\rchi}|x|\le T$ and $\|K_1(\cdot;\theta)\|_{2,\infty} \le 2T$, see \ref{smooth:a2} and \eqref{smooth:eq:K_bound}.
Next we prove that there exists finite $c$ depending only on $M_1,M_2,$ and $M_3,$ such that
\be \label{smooth:eq:entropy_Wstar}
 N(\varepsilon, \mathcal{D}^*_{M_1, M_2,M_3}, \|\cdot\|_{2,\infty})  \leq c \exp \left(\frac{c}{\varepsilon}+ \frac{c}{\sqrt{\varepsilon}}\right)\varepsilon^{-2(d-1)}.
 \ee  We first find covers for $\mathcal{C}^{m*}_{M_1, M_2,M_3}$,  $\{ f^\prime : f \in \mathcal{C}^{m*}_{M_1, M_2,M_3} \big\}$, and  $\Theta\cap B_{\theta_0}(1/2)$ and use  them to construct a cover for $\mathcal{D}^*_{M_1,M_2, M_3}$.
By Lemma \ref{smooth:lem:entropy} (for $k=1$ and $2$, respectively), we have
\begin{align}
N(\varepsilon,  \mathcal{C}^{m*}_{M_1, M_2,M_3}, \|\cdot\|_\infty) &\le \exp(c/\sqrt{\varepsilon}),\\
N(\varepsilon, \big\{ f^\prime : f \in \mathcal{C}^{m*}_{M_1, M_2,M_3} \big\}, \|\cdot\|_\infty) &\le \exp(c/{\varepsilon}),
\end{align}
where $c$ is a constant depending only on $M_1,M_2,$ and $M_3.$  Let us denote the functions in the $\varepsilon$-cover of $\mathcal{C}^{m*}_{M_1, M_2,M_3}$ by $r_1,\ldots, r_q$ and the functions in the $\varepsilon$-cover of  $\{ f^\prime : f \in \mathcal{C}^{m*}_{M_1, M_2,M_3} \big\}$ by $l_1,\ldots, l_{t}.$ By Lemma~\ref{lem:H_ent}, we have that there exists  $\theta_1,\ldots, \theta_s$ for $s\lesssim \varepsilon^{-4d}$ such that $\{\theta_i\}_{1\le i \le s}$ form an $\varepsilon^2$-cover  of $\Theta \cap B_{\theta_0}(1/2)$ and satisfies~\eqref{smooth:eq:H_cover}.  Fix $(m,\theta) \in \mathcal{C}_{M_1,M_2, M_3}(n).$ Without loss of generality  assume that the function nearest to $m$ in the $\varepsilon$-cover of $\mathcal{C}^{m*}_{M_1, M_2,M_3}$ is $r_1,$  the function nearest to $m^\prime$ in the $\varepsilon$-cover of $\{ f^\prime : f \in \mathcal{C}^{m*}_{M_1, M_2,M_3} \big\}$ is $l_1,$ and the vector nearest to $\theta$ in the $\varepsilon^2$-cover of $\Theta \cap B_{\theta_0}(1/2)$ is $\theta_1$, i.e.,
\begin{equation}\label{eq:local_cond}
\|m-r_1\|_{\infty} \le \varepsilon, \quad  \|m^{\prime}-l_1\|_{\infty} \le \varepsilon,\quad \|H_{\theta_1}^\top-H_\theta^\top\|_2 \le \varepsilon^2 \quad \text{and} \quad |\theta_1-\theta| \le \varepsilon^2.
\end{equation}

 Now for every $x\in \rchi$, observe that
\begin{align}
\begin{split}\label{smooth:eq:first_part_split}
&\big|(m_0 (\theta_0 ^\top x) -m(\theta ^\top x)) m^\prime (\theta ^\top x) K_1(x;\theta)- (m_0 (\theta_0 ^\top x) -r_1(\theta_1 ^\top x)) l_1 (\theta_1 ^\top x) K_1(x;\theta_1)\big|\\
={}&\big|(m_0 (\theta_0 ^\top x) -m(\theta ^\top x)) m^\prime (\theta ^\top x) K_1(x;\theta)-\\
 &\qquad \big(m_0 (\theta_0 ^\top x)-m(\theta ^\top x)+m(\theta ^\top x) -r_1(\theta_1 ^\top x)\big) l_1 (\theta_1 ^\top x) K_1(x;\theta_1)\big|\\
\le{}&\big| m_0 (\theta_0 ^\top x) -m(\theta ^\top x)\big| \big| m^\prime (\theta ^\top x) K_1(x;\theta)-  l_1 (\theta_1 ^\top x) K_1(x;\theta_1)\big|\\
&\qquad +\big|m(\theta ^\top x) -r_1(\theta_1 ^\top x)\big| \big|l_1 (\theta_1 ^\top x) K_1(x;\theta_1)\big|\\
={}&\textbf{A}+\textbf{B},
\end{split}
\end{align}
where
\begin{align}
\textbf{A}&:=\big| m_0 (\theta_0 ^\top x) -m(\theta ^\top x)\big|m^\prime (\theta ^\top x) K_1(x;\theta)-  l_1 (\theta_1 ^\top x) K_1(x;\theta_1)\big|\\
\textbf{B}&:=\big|m(\theta ^\top x) -r_1(\theta_1 ^\top x)\big| \big|l_1 (\theta_1 ^\top x) K_1(x;\theta_1)\big|.
\end{align}
We next find an upper bound for \textbf{A}. First, by Lemma \ref{smooth:lem:H_lip} and assumption \ref{smooth:b3}, we have
\begin{align}
\begin{split}\label{smooth:eq:K_diff_bound}
&\big |K_1(x;\theta)-  K_1(x;\theta_1)\big| \\
={}&\big|H_\theta^\top (x- h_{\theta} (\theta^\top x))- H_{\theta_1}^\top(x- h_{\theta} (\theta^\top x))+  H_{\theta_1}^\top(x- h_{\theta} (\theta^\top x))   - H_{\theta_1}^\top (x- h_{\theta_1} (\theta_1^\top x))\big| \\
\le{}&\big|(H_\theta^\top - H_{\theta_1}^\top) (x- h_{\theta} (\theta^\top x))\big|+ \big| H_{\theta_1}^\top \big[(x- h_{\theta} (\theta^\top x))   -  (x- h_{\theta_1} (\theta_1^\top x))\big]\big| \\
\le{}&\|H_\theta^\top - H_{\theta_1}^\top\|_2 2T + \big|h_{\theta} (\theta^\top x)   -  h_{\theta_1} (\theta_1^\top x)\big| \\
\le{}&2 T\varepsilon^2+ (\bar{M}+\|h_{\theta_0}'\|_\infty) |\theta-\theta_1| \lesssim \varepsilon^2.
\end{split}
\end{align}
Now observe that
\begin{align}
\begin{split}\label{smooth:eq:A}
\textbf{A}\le{}&2M_1\big| m^\prime (\theta ^\top x) K_1(x;\theta)-  l_1 (\theta_1 ^\top x) K_1(x;\theta_1)\big|\\
\le{}&2 M_1\big| m^\prime (\theta ^\top x) K_1(x;\theta)-  l_1 (\theta ^\top x) K_1(x;\theta)\big|+ \big| l_1 (\theta ^\top x) K_1(x;\theta)-  l_1 (\theta_1 ^\top x) K_1(x;\theta)\big|\\
&\qquad +2 M_1 \big| l_1 (\theta_1 ^\top x) K_1(x;\theta)-  l_1 (\theta_1 ^\top x) K_1(x;\theta_1)\big|\\
\le{}&2 M_1 |K_1(x;\theta)| \big| m^\prime (\theta ^\top x) -  l_1 (\theta ^\top x)\big| + | K_1(x;\theta)|\big| l_1 (\theta ^\top x)-  l_1 (\theta_1 ^\top x)\big|\\
& \qquad +2 M_1 \|l_1\|_\infty \big|  K_1(x;\theta)-  K_1(x;\theta_1)\big|\\
\lesssim{}& 4TM_1 \bigg( \varepsilon  + \left[\int_D {l_1^\prime}^2(z) dz \right] |\theta-\theta_1|^{1/2} T^{1/2}\bigg) + 2 M_1M_2(2T+\bar{M}) \varepsilon^2\\
\le{}& 4T M_1 (\varepsilon  +  M_3 |\theta-\theta_1|^{1/2} T^{1/2}) + (2T+\bar{M})2 M_1 M_2  \varepsilon^2 \\
\lesssim{}& \varepsilon,
\end{split}
\end{align}
where the penultimate inequality follows from~\eqref{eq:local_cond} and the last inequality follows from \ref{smooth:a2}, \eqref{smooth:eq:K_diff_bound}, and Lemma \ref{smooth:lem:bounds}. To find an upper bound for $\textbf{B}$, observe that
\begin{align}\label{smooth:eq:B}
\begin{split}
\textbf{B}={}&\big|m(\theta ^\top x) -r_1(\theta_1 ^\top x)\big| \big|l_1 (\theta_1 ^\top x) K_1(x;\theta_1)\big|\\
\le{}&\Big[ \big|m(\theta ^\top x) -r_1(\theta^\top x)\big|+ \big|r_1(\theta ^\top x) -r_1(\theta_1 ^\top x)\big|\Big] \big|l_1 (\theta_1 ^\top x) K_1(x;\theta_1)\big|\\
\le{}& \big[ \varepsilon +  \|r_1^\prime\|_\infty |\theta-\theta_1| T\big] \|l_1\|_\infty 2T \lesssim \varepsilon.
\end{split}
\end{align}
Combining  \eqref{smooth:eq:first_part_split}, \eqref{smooth:eq:A}, and \eqref{smooth:eq:B} we get  that $\{(m_0 (\theta_0 ^\top x) -r_i(\theta_k^\top x)) l_j^\prime (\theta_k ^\top x) K_1(x;\theta_k)\}_{i,j,k}$ for ${1\le i\le q , 1\le j \le t,}$ and  ${1\le k \le s}$ form an (constant multiple of) $\varepsilon$-cover (with respect to $\|\cdot\|_{2,\infty}$ norm) of $\mathcal{D}^*_{M_1, M_2,M_3}.$ Thus we have \eqref{smooth:eq:entropy_Wstar}. Moreover, as $N_{[\,]}(\varepsilon, \mathcal{D}^*_{M_1, M_2,M_3}, \|\cdot\|_{2, P_{\theta_0, m_0}}) \lesssim N(\varepsilon, \mathcal{D}^*_{M_1, M_2,M_3}, \|\cdot\|_{2,\infty})$ and
  $$\mathcal{D}_{M_1, M_2,M_3}(n) \subset \mathcal{D}^*_{M_1, M_2,M_3},$$
   for every $n \in \mathbb{N},$ we have $N_{[\,]}(\varepsilon,\mathcal{D}_{M_1, M_2,M_3}(n), \|\cdot\|_{2, P_{\theta_0, m_0}}) \lesssim N_{[\,]}(\varepsilon, \mathcal{D}^*_{M_1, M_2,M_3}, \|\cdot\|_{2, P_{\theta_0, m_0}})$ and  $J_{[\,]}(\gamma,\mathcal{D}^*_{M_1, M_2,M_3}(n),\|\cdot\|_{2, P_{\theta_0, m_0}}) \lesssim c \gamma^{1/2}.$
Observe that $f\in \mathcal{D}_{M_1, M_2,M_3}(n)$  maps $\rchi$ to $\R^{d-1}$. For any $f\in \mathcal{D}_{M_1, M_2,M_3}(n),$ let $f_1, \ldots, f_{d-1}$ denote each of the real valued components, i.e., $f(\cdot):= (f_1(\cdot), \ldots, f_{d-1}(\cdot))$. With this notation, we have
\begin{align}\label{eq:dim_split}
\begin{split}
&\mathbb{P}\bigg( \sup_{f \in \mathcal{D}_{M_1, M_2,M_3}(n) }|\g_n f | > \delta \bigg)\\
&{}\le  \sum_{i=1}^{d-1} \mathbb{P}\bigg( \sup_{f \in \mathcal{D}_{M_1, M_2,M_3}(n) }|\g_n f_i | > \delta/\sqrt{d-1} \bigg).
\end{split}
\end{align}
We can bound each term in the summation of~\eqref{eq:dim_split} using the maximal inequality in Corollary 19.35 of \cite{VdV98}. We have
\begin{align}
\mathbb{P}\bigg( \sup_{f \in \mathcal{D}_{M_1, M_2,M_3}(n) }|\g_n f_1 | > \delta \bigg)\le{}& \delta^{-1} \mathbb{E} \bigg( \sup_{f \in \mathcal{D}_{M_1, M_2,M_3}(n) }|\g_n f_1 |\bigg)\\
\le{}& \delta^{-1} J_{[\,]}(\| D_{M_1, M_2,M_3}(n)\|,\mathcal{D}^*_{M_1, M_2,M_3}(n),\|\cdot\|_{2, P_{\theta_0, m_0}}) \\
\lesssim{}& \delta^{-1} \|D_{M_1, M_2,M_3}(n) \|^{1/2}\\
\lesssim{}& \Big[ \hat{\lambda}_n^{1/2} +a_n^{-1}\Big]^{1/2}\rightarrow 0, \qquad \text{as } n\rightarrow \infty. \label{eq:each_dim_1}
\end{align}
In the last inequality, we have used \eqref{smooth:eq:g_beta0_U_envelope} and the fact that $D^2_{M_1, M_2,M_3}(n)$ is non-random. The lemma follows by combining \eqref{eq:each_dim_1} and~\eqref{eq:dim_split}.

\subsection{Proof of Lemma \ref{smooth:lem:second_term}} \label{smooth:app:proof_second_term}
We will first  show that, for every $ (m,\theta) \in \mathcal{C}_{M_1, M_2,M_3}(n)$ and $x \in \rchi,$ we have
\begin{equation}\label{smooth:eq:env_epsilion}
\Big|\epsilon \big[ U_{\theta,m} (x)-U_{\theta_0,m_0}(x)\big]\Big| \le |\epsilon| W_{M_1, M_2,M_3}(n).
\end{equation}
 Observe that for every $(m,\theta) \in \mathcal{C}_{M_1,M_2, M_3}(n)$ and $x \in \rchi$, we have
\begin{align}
&| U_{\theta,m} (x)-U_{\theta_0,m_0}(x)|\\
\le{}&|m^\prime (\theta^\top x) K_1(x;\theta)-m^\prime (\theta_0^\top x) K_1(x;\theta)|+|m^\prime (\theta_0^\top x) K_1(x;\theta)-m_0^\prime (\theta_0^\top x) K_1(x;\theta_0)|\\
\le{}&|m^\prime (\theta^\top x) K_1(x;\theta)-m^\prime (\theta_0^\top x) K_1(x;\theta)|+|m^\prime (\theta_0^\top x) K_1(x;\theta)-m_0^\prime (\theta_0^\top x) K_1(x;\theta)|\\
&\qquad +|m_0^\prime (\theta_0^\top x) K_1(x;\theta)-m_0^\prime (\theta_0^\top x) K_1(x;\theta_0)|\\
\le{}&|m^\prime (\theta^\top x) -m^\prime (\theta_0^\top x)| |K_1(x;\theta)|+|m^\prime (\theta_0^\top x) -m_0^\prime (\theta_0^\top x) | |K_1(x;\theta)|\\
& \qquad +|m_0^\prime (\theta_0^\top x)|| K_1(x;\theta)-K_1(x;\theta_0)|\\
\le{}&J(m)|\theta_0-\theta|^{1/2} T^{1/2} |K_1(x;\theta)|+\|m-m_0\|_{D_0}^S |K_1(x;\theta)| +\|m_0^\prime \|_\infty (2T+\bar{M}+\|h_{\theta_0}'\|_\infty)|\theta_0-\theta|\\
\le{}& \big[  2 T^{3/2} M_3 \hat{\lambda}_n^{1/4}+  2 T a_n^{-1} +M_2 (2T+\bar{M}+\|h_{\theta_0}'\|_\infty)\big] \hat{\lambda}_n^{1/2}=W_{M_1,M_2,M_3}(n),
\end{align}
where for the third term in the penultimate inequality follows from~\eqref{smooth:eq:K_diff_bound}.

Next, we will prove that there exists a constant $c$ depending only on $M_1,M_2,$ and $M_3$ such that
\be \label{smooth:eq:ent_w_n}
N_{[\,]}(\varepsilon, \w_{M_1, M_2,M_3}(n), \|\cdot\|_{2, P_{\theta_0, m_0}}) \leq c \exp (c/\varepsilon) \varepsilon^{-2(d-1)}.
\ee
  As in  proof of Lemma  \ref{smooth:lem:first_term}, we first find covers for the class of functions $\{ f^\prime : f \in \mathcal{C}^{m*}_{M_1, M_2,M_3} \big\}$ and the set $\Theta\cap B_{\theta_0}(1/2)$ and use  them to construct a cover for $\w^*_{M_1, M_2,M_3}$.
By Lemma \ref{smooth:lem:entropy}, we have
\begin{align}
N(\varepsilon, \big\{ f^\prime : f \in \mathcal{C}^{m*}_{M_1, M_2,M_3} \big\}, \|\cdot\|_\infty) &\le \exp(c/{\varepsilon}),
\end{align} where $c$ is a constant depending only on $d,M_1,M_2,$ and $M_3.$ We denote the functions in the $\varepsilon$-cover of $\{ f^\prime : f \in \mathcal{C}^{m*}_{M_1, M_2,M_3} \big\}$ by $l_1,\ldots, l_{t}$. By Lemma~\ref{lem:H_ent}, we have that there exists  $\theta_1,\ldots, \theta_s$ for $s\lesssim \varepsilon^{-4d}$ such that $\{\theta_i\}_{1\le i \le s}$ form an $\varepsilon^2$-cover  of $\Theta \cap B_{\theta_0}(1/2)$ and satisfies~\eqref{smooth:eq:H_cover} (with $\varepsilon^2$ instead of $\varepsilon$). Fix $(m,\theta) \in \mathcal{C}_{M_1,M_2, M_3}(n).$  Without loss of generality  assume that the function nearest to $m^\prime$ in the $\varepsilon$-cover of $\{ f^\prime : f \in \mathcal{C}^{m*}_{M_1, M_2,M_3} \big\}$ is $l_1$ and the vector nearest to $\theta$ in the $\varepsilon^2$-cover of $\Theta\cap B_{\theta_0}(1/2)$ is $\theta_1$, i.e.,
\[
\|m^{\prime}-l_1\|_{\infty} \le \varepsilon, \quad \quad |\theta_1-\theta| \le \varepsilon^2, \quad \text{and} \quad \|H_\theta^\top-H_{\theta_1}^\top\|_2 \le \varepsilon^2.
\]
Let us define $r_1, \ldots, r_t$ to be  anti-derivatives of $l_1,\ldots, l_t$, i.e., $l_1= r_1^\prime, \ldots l_t=r_t^\prime$. Then for every $x\in \rchi$, observe that
\begin{align}
&|U_{\theta,m}(x)-U_{\theta_1,r_1}(x)|\\
\le{}&|U_{\theta,m}(x)-U_{\theta,r_1}(x)|+|U_{\theta,r_1}(x)-U_{\theta_1,r_1}(x)|\\
\le{}&|m^\prime (\theta^\top x) K_1(x;\theta)-r_1^\prime (\theta^\top x) K_1(x;\theta)|+|r_1^\prime (\theta^\top x) K_1(x;\theta)-r_1^\prime (\theta_1^\top x) K_1(x;\theta_1)|\\
\le{}&|m^\prime (\theta^\top x)-r_1^\prime (\theta^\top x)|| K_1(x;\theta)|+|r_1^\prime (\theta^\top x) K_1(x;\theta)-r_1^\prime (\theta_1^\top x) K_1(x;\theta)|\\
&\qquad +|r_1^\prime (\theta_1^\top x) K_1(x;\theta)-r_1^\prime (\theta_1^\top x) K_1(x;\theta_1)|\\
\le{}&\varepsilon | K_1(x;\theta)|+|r_1^\prime (\theta^\top x) -r_1^\prime (\theta_1^\top x)| |K_1(x;\theta)| +\|r_1^\prime\|_\infty| K_1(x;\theta)-K_1(x;\theta_1)|\\
\le{}&\varepsilon \| K_1(\cdot;\theta)\|_{2,\infty}+J(r_1)|\theta-\theta_1|^{1/2} T^{1/2} \|K_1(\cdot;\theta)\|_{2,\infty} +M_1 (2T+\bar{M})|\theta-\theta_1|\lesssim{}\varepsilon.
\end{align}
Here the last inequality follows from \ref{smooth:a2}, \eqref{smooth:eq:K_diff_bound}, and Lemma \ref{smooth:lem:bounds}.
 Thus, $\{U_{\theta_i,r_j}-U_{\theta_0,m_0}\}_{1\le i\le t , 1\le j \le s}$ form an (constant multiple of) $\varepsilon$-cover (with respect to $\|\cdot\|_{2,\infty}$ norm) of $\w^*_{M_1, M_2,M_3}$. Moreover, as $N_{[\,]}(\varepsilon, \w^*_{M_1, M_2,M_3}, \|\cdot\|_{2,P_{\theta_0,m_0}}) \lesssim N(\varepsilon, \w^*_{M_1, M_2,M_3}, \|\cdot\|_{2,\infty})$ and
$\w_{M_1, M_2,M_3}(n) \subset \w^*_{M_1, M_2,M_3},$
  for every $n \in \mathbb{N},$ we have $$N_{[\,]}(\varepsilon,\w_{M_1, M_2,M_3}(n), \|\cdot\|_{2,P_{\theta_0,m_0}}) \lesssim N_{[\,]}(\varepsilon, \w^*_{M_1, M_2,M_3}, \|\cdot\|_{2,P_{\theta_0,m_0}})\lesssim c \exp (c/\varepsilon) \varepsilon^{-4d} .$$
Observe that if  $[\hbar_1, \hbar_2]$ is a bracket for $U_{\theta, m}-U_{\theta_0,m_0},$ then $[\hbar_1 \epsilon^+ -\hbar_2 \epsilon^-, \hbar_2 \epsilon^+ -\hbar_1 \epsilon^-]$ is a bracket (here the ordering is coordinate-wise) for $\epsilon (U_{\theta, m}-U_{\theta_0,m_0}).$ Therefore, we have
\be \label{smooth:eq:bracket_entropy_epsilon_w}
 N_{[\,]}\big(\varepsilon \|\sigma(\cdot)\|_\infty, \{ \epsilon f : f \in \w_{M_1, M_2,M_3}(n)\}, \|\cdot\|_{2,P_{\theta_0,m_0}}\big)  \le c \exp(c/\varepsilon) \varepsilon^{-2(d-1)}.
 \ee
Now we prove \eqref{smooth:eq:second_part_main}.  As in \eqref{smooth:eq:insetoutset}, we have
\begin{align}
&\p( \big|\g_n \big[ \epsilon \big(U_{\hat{\theta},\hat{m}}(X)-U_{\theta_0,m_0}(X)\big)\big]\big|> \delta) \\
\le{}& \p \Big(\sup_{(m,\theta) \in \mathcal{C}_{M_1,M_2, M_3}(n)} \big|\g_n \big[ \epsilon \big(U_{\theta,m}(X)-U_{\theta_0,m_0}(X)\big)\big]\big|> \delta\Big) +\p((\hat{\theta}, \hat{m}) \notin \mathcal{C}_{M_1,M_2, M_3}(n))
\end{align}
  By discussions similar to those after Theorem \ref{smooth:thm:ConsistencyofG_n}, we only need to show that for every fixed $M_1, M_2,$ and $M_3,$ we have
$$\p \Big(\sup_{f  \in \w_{M_1, M_2,M_3}(n)} |\g_n \epsilon f| >\delta\Big)\rightarrow 0,$$ as $n \rightarrow 0.$
Note that by \eqref{smooth:eq:bracket_entropy_epsilon_w},  for $\gamma >0$ we have  $$J_{[\,]}\Big(\gamma,\{ \epsilon f : f \in \w_{M_1, M_2,M_3}(n)\}, \|\cdot\|_{2,P_{\theta_0,m_0}}\Big) \lesssim \gamma^{\frac{1}{2}}.$$
 By arguments similar to~\eqref{eq:dim_split} and~\eqref{eq:each_dim_1}, we have
\begin{align}
\p\bigg(\sup_{f \in \w_{M_1,M_2, M_3}(n) } |\g_n \epsilon f|> {\delta}\bigg) \lesssim& \delta^{-1} \E  \bigg(\sup_{f \in \w_{M_1,M_2, M_3}(n) } |\g_n \epsilon f|\bigg)\\
\lesssim& J_{[\,]}\Big(P_{\theta_0,m_0}\big(|\epsilon^2| W_{M_1, M_2,M_3}^2(n)\big)^{\frac{1}{2}},\w_{M_1,M_2, M_3}(n), L_2(P_{\theta_0,m_0})\Big) \\
\lesssim& \left[\hat{\lambda}_n^{1/4}+a_n^{-1}+ \hat{\lambda}_n^{1/2}\right]^{1/2}  \rightarrow{} 0,\,   as \, n \rightarrow \infty.%\qedhere
%\qquad \qqaud \text{as } n\rightarrow \infty. \qedhere
\end{align}
\end{document}